\documentclass[12pt]{article}
\usepackage[left=1in,top=1in,right=1in,bottom=1in,head=.1in,nofoot]{geometry}

\usepackage{setspace,url,bm,amsmath} %

\usepackage{titlesec} 
\titlelabel{\thetitle.\quad} %
\titleformat*{\section}{\bf\Large\center} %

\usepackage{graphicx} %
\usepackage{bbm}
\usepackage{latexsym}
\usepackage{caption}
\usepackage[margin=20pt]{subcaption}
\usepackage{hyperref}
\usepackage{booktabs}
\usepackage{multirow}
\usepackage{enumitem}
\usepackage[table]{xcolor}

\newcommand{\GG}[1]{}
\usepackage{amsthm}
\usepackage{amssymb}
\usepackage{amsmath}
\usepackage{color}

\usepackage{comment}
\theoremstyle{definition}

\newtheorem*{theorem*}{Theorem}
\newtheorem{theorem}{Theorem}
\newtheorem*{rmk*}{Remark}
\newtheorem{rmk}{Remark}

\newtheorem{lemma}{Lemma}

\newtheorem{condition}{Condition}

\newtheorem{definition}{Definition}

\newtheorem{corollary}{Corollary}
\newtheorem*{corollary*}{Corollary}

\def\pr{\mathbb{P}}

\def\sta{\star}

\usepackage{natbib} %
\bibpunct{(}{)}{;}{a}{}{,} %

\usepackage{etoolbox} %
\apptocmd{\sloppy}{\hbadness 10000\relax}{}{} %

\usepackage{color}
\usepackage{listings}

\newcommand{\lxr}{\color{black}}

\def\ind{\begin{picture}(9,8)
         \put(0,0){\line(1,0){9}}
         \put(3,0){\line(0,1){8}}
         \put(6,0){\line(0,1){8}}
         \end{picture}
        }

\def\Var{\text{Var}}

\def\convergeas{\stackrel{\textup{a.s.}}{\longrightarrow}}
\def\converged{\stackrel{d}{\longrightarrow}}

\def\iidsim{\stackrel{i.i.d.}{\sim}}

\def\I{\mathbbm{1}}
\def\E{\mathbb{E}}

\def\bs{\bm}

\def\obs{\text{obs}}

\allowdisplaybreaks

\def\vecge{\succcurlyeq}
\def\vecle{\preccurlyeq}
\usepackage{textcomp}

\def\thmadj{}
\def\simu{\text{sim}}

\usepackage{algorithm}
\usepackage{algorithmic}

\usepackage{soul}

\def\Unif{\text{Unif}}

\usepackage[textsize=scriptsize, textwidth = 2cm, shadow]{todonotes}

\RequirePackage[normalem]{ulem}

\def\add{\color{black}}

\def\DW{\color{black}}

\usepackage{tikz}
\usetikzlibrary{shapes.geometric, arrows}
\tikzstyle{io} = [trapezium, trapezium left angle=70, trapezium right angle=110, minimum width=1cm, minimum height=1cm, text centered, draw=black,  trapezium stretches=true, thick]
\tikzstyle{process} = [rectangle, minimum width=3cm, minimum height=1cm, text centered, draw=black, thick]
\tikzstyle{decision} = [diamond, minimum width=1cm, minimum height=1cm, text centered, draw=black, thick]
\tikzstyle{arrow} = [thick,->,>=stealth]

\usepackage{subcaption}

\usepackage{accents}
\newcommand{\dbtilde}[1]{\accentset{\approx}{#1}}
\newcommand{\vardbtilde}[1]{\tilde{\raisebox{0pt}[0.85\height]{$\tilde{#1}$}}}

\DeclareFontFamily{OT1}{pzc}{}
\DeclareFontShape{OT1}{pzc}{m}{it}{<-> s * [1.10] pzcmi7t}{}
\DeclareMathAlphabet{\mathpzc}{OT1}{pzc}{m}{it}

\begin{document}

\title{\bf 
Sensitivity Analysis for Quantiles of Hidden Biases in Matched Observational Studies
}
\author{
Dongxiao Wu and Xinran Li
\footnote{
Dongxiao Wu is Doctoral Candidate, Department of Statistics, University of Illinois at Urbana-Champaign, Champaign, IL 61820 (e-mail: \href{mailto:dw12@illinois.edu}{dw12@illinois.edu}). 
Xinran Li is Assistant Professor, Department of Statistics, University of Chicago, Chicago, IL 60637 (e-mail: \href{mailto:xinranli@uchicago.edu}{xinranli@uchicago.edu}). 
}
}
\date{}
\maketitle

\begin{abstract}
\onehalfspacing
Causal conclusions 
from observational studies 
may be sensitive to unmeasured confounding. In such cases, a sensitivity analysis is often conducted, 
which
tries to infer 
the minimum amount of hidden biases 
or the minimum strength of unmeasured confounding
needed in order to explain away the observed association between treatment and outcome.  
If the needed bias is large, 
then the treatment is likely to have significant effects.  
The Rosenbaum sensitivity analysis is a modern 
approach for conducting sensitivity analysis in matched observational studies. It investigates what magnitude the maximum of hidden biases from all matched sets needs to be in order to explain away the observed association.  
However, such a sensitivity analysis can be overly conservative and pessimistic, especially when 
investigators 
suspect 
that some matched sets may have exceptionally large hidden biases. 
In this paper, we generalize 
Rosenbaum's framework to conduct sensitivity analysis on 
quantiles of hidden biases from all matched sets, 
which are more robust than the maximum.  
Moreover, 
the proposed sensitivity analysis is simultaneously valid across all quantiles of hidden biases and is thus a free lunch added to the conventional 
sensitivity analysis.
The proposed approach works for general outcomes, general matched studies and general test statistics. 
In addition, 
we demonstrate that the proposed sensitivity analysis also works for bounded null hypotheses 
when 
the test statistic satisfies certain properties.
An R package implementing the proposed approach 
is 
available online. 
\end{abstract}


\onehalfspacing

\section{Introduction}
\subsection{Observational studies, confounding and sensitivity analysis}

Randomized experiments are the gold standard for drawing causal inference, under which all potential confounders are balanced on average. 
However, in observational studies, we do not have control of the treatment assignment, 
and the choice of it is often based on the subjects themselves. Thus, different treatment groups may differ significantly in some covariates, 
making them not directly comparable 
{\lxr when}
inferring causal effects. 
Various approaches have been proposed to adjust for covariate imbalance between different treatment groups 
\citep[see, e.g.,][]{Imbens15a, rosenbaum2010book, robins2020}. 
{\lxr However,}
such adjustments can generally {\lxr only} be performed for observed covariates.
Inevitably, one may be concerned that the unmeasured confounding will drastically change the inferred causal conclusions. 

In the presence of unmeasured confounding, causal effects are generally not identifiable. 
A sensitivity analysis is often conducted, 
which asks what amount of unmeasured confounding is needed 
in order 
to explain away the observed treatment-outcome association, 
assuming that the treatment has no causal effect. 
The idea of sensitivity analysis dates back to 
\citet{cornfield1959smoking}. 
They demonstrated that, 
assuming smoking has no effect on lung cancer and 
to attribute the observed 
smoking-cancer
association 
to a binary confounder,  
the relative risk of smoking on this confounder must be 
no less than 
that of smoking on lung cancer.  
If hidden confounding of such strength is not likely to exist, then smoking is likely to cause lung cancer. 
There have been many efforts generalizing Cornfield’s sensitivity analysis.
We mention a few below; see \citet{liu2013introduction}, \citet{Richardson2014}, \citet{dingli2018} and references therein for a more comprehensive review. 
First, 
\citet{rubinsen1983} and \citet{imbens2003} conducted model-based sensitivity analysis involving an unobserved 
confounder for both treatment assignment and potential outcome models. 
Second, 
\citet{Rosenbaum1987, Rosenbaum02a} proposed sensitivity analysis for matched observational studies with bounds on the odds ratio of 
treatment probabilities 
for units within 
each 
matched set. 
Third, 
\citet{Vanderweele2014, ding2016sensitivity} generalized Cornfield's inequality 
by allowing a categorical unmeasured confounder 
and extended it to other measures such as risk difference. 
Fourth, 
\citet{robins2000} and \citet{Franks2020} 
considered sensitivity analysis using the contrast of potential outcome distributions between treated and control units. 

\vspace{-0.5ex}
\subsection{Matched observational studies and Rosenbaum's sensitivity analysis}\label{sec:match_Rosen_sen}

We focus on matched observational studies with two treatment arms, 
{\lxr where}
matching {\lxr is used}
to adjust for the imbalance of observed covariates.  
When there is no unmeasured confounding and matching is exact, 
units within each matched set have the same 
treatment probability. 
In this case, the matched study reduces to a matched randomized experiment, and Fisher's randomization test 
can be used to test, say, the null hypothesis 
of no treatment effect 
for any unit \citep{Fisher35a}. 
However, in practice,
we can rarely rule out the existence of unmeasured confounding, and units within each matched set, although similar in observed covariates, 
{\lxr may}
have rather different 
unmeasured confounders. 
Thus, 
randomization tests pretending completely randomized treatment assignments within each matched set generally provide 
biased or invalid 
inferences for causal effects.

\citet{Rosenbaum02a} proposed an assumption-lean sensitivity analysis framework for matched observational studies. 
He measures the hidden bias or strength of unmeasured confounding 
for each matched set by the maximum odds ratio of  
treatment probabilities 
for any two units, and 
assumes that the hidden biases for all matched sets are bounded by some constant $\Gamma_0$.  
He then constructs valid 
$p$-values for testing the sharp null hypothesis of no causal effect under each value of 
$\Gamma_0$, and investigates how the $p$-value varies with 
$\Gamma_0$. 
Usually, the maximum $\Gamma_0$, denoted by $\hat{\Gamma}_0$, such that the $p$-value is significant at, say, $5\%$ level, will be reported. 
If the hidden biases in these matched sets are unlikely to exceed $\hat{\Gamma}_0$, then the causal effects are likely to be significant. 
As discussed later,
$\hat{\Gamma}_0$ is actually a $95\%$ lower confidence limit for the maximum hidden bias from all matched sets, assuming that the treatment has no causal effect. 
This is analogous to \citet{cornfield1959smoking}.

Rosenbaum's sensitivity analysis considers mainly a uniform bound on the hidden biases in all matched sets,
and thus 
focuses on the maximum hidden bias from all matched sets. 
However, in practice, the hidden bias is likely to vary across sets, 
and 
such a sensitivity analysis 
may lose power if one suspects that some sets have exceptionally large 
hidden biases. 
As discussed in \citet{HS2017} and \citet{FH2017}, 
the sensitivity analysis on maximum hidden bias will be conservative if the worst-case hidden bias does not hold in every matched pair, 
and will yield an unduly pessimistic perception of a study's robustness to hidden confounding if rare cases of extreme hidden biases are suspected. 
To overcome this issue, 
they proposed sensitivity analysis for the average hidden bias from matched pairs.
Instead, we 
will 
generalize the conventional 
sensitivity analysis 
to quantiles of hidden biases, 
which provides 
a new way to mitigate its overly conservative issue.  
Our proposal is closely related to \citet[][Section 4]{Rosenbaum1987}, who considered 
cases where 
hidden biases of most matched pairs are bounded by some constant while the remaining can be unbounded.

\subsection{A motivating example}\label{sec:motivate}

We consider a matched observational study consisting of 512 matched sets, each containing one smoker and two matched nonsmokers, from the 2005-2006 US National Health and Nutrition Examination Survey.
We are interested in the causal effect of smoking on the 
{\lxr levels}
of cadmium ($\mu$g/l) and lead ($\mu$g/dl) in the blood; see Section \ref{sec:application} for details. 
Many studies have shown {\lxr associations} between the levels of these heavy metals and the risk of cancer and blood pressure elevation \citep[see, e.g.,][]{an2017association,oh2014blood,staessen1984urinary}.

We consider the effect of smoking on the blood cadmium level.
Assuming matching has 
taken into account all the confounding, the $p$-value 
from the randomization test for 
Fisher's null of no effect is almost zero, indicating a significant effect of smoking. 
However, we may suspect that such a significant smoking-cadmium association is due to unmeasured confounding. 
We then conduct Rosenbaum's sensitivity analysis.  
When the 
maximum hidden bias is 
bounded by {\DW $82.44$}, the $p$-value for the sharp null of no effect 
is significant at the $5\%$ level; when the maximum hidden bias
exceeds {\DW $82.44$}, the $p$-value becomes insignificant. 
Such an evidence is generally considered to have a high degree of insensitivity to unmeasured confounding. 

However, although it is unlikely for all matched sets or a moderate proportion of them to have hidden biases greater than {\DW $82.44$}, 
it 
{\lxr is} 
possible that a 
{\lxr small}
proportion of the sets have hidden biases exceeding {\DW $82.44$}. 
{\lxr In this case,}
the maximum hidden bias may not be representative of the majority of 
the matched sets, 
{\lxr rendering}
the conventional analysis 
{uninformative}.
For example, if one 
unit 
is almost certain to smoke due to, say, unmeasured 
high peer pressure
\citep{HS2017}, 
then the corresponding matched set will have infinite hidden bias.
In this case, the conventional sensitivity analysis loses power to detect significant causal effect 
and we will conclude that the results from the study are sensitive to unmeasured confounding. 

\begin{table}
    \caption{Sensitivity analysis for the effect of smoking. 
    The 1st and 2nd columns show bounds on hidden biases and 
    the 
    corresponding 
    $p$-values for testing Fisher's null of no effect. The 3rd column gives conditions under which the $p$-value becomes insignificant.
    {\add The conclusions of the study are sensitive only if all conditions in the 3rd column are likely to hold simultaneously.} 
    }
    \label{tab:motivate}
    \centering
    \resizebox{0.9\columnwidth}{!}{%
    \begin{tabular}{rrr}
        \toprule
        Constraints on Hidden biases & $p$-value &   conclusions of the study are sensitive  \\
         & & if all of the following are likely to be true 
        \\
        \midrule
        the maximum hidden bias $\le$  {\DW $82.44$} &  {\DW $0.04999661$} & 
        at least one matched set has bias $>$ {\DW $82.44$}\\
        the 52nd largest hidden bias $\le$ {\DW $72.52$} & {\DW $0.04999786$} &  
        at least $10\%$ of matched sets have biases $>$ {\DW $72.52$} \\
        the 154th largest hidden bias $\le${\DW $46.90$} & {\DW $0.04995089$} & 
        at least $30\%$ of matched sets have biases $>$ {\DW $46.90$}\\
        the 256th largest hidden bias $\le$ {\DW $26.88$} & {\DW $0.04989957$} & 
        at least $50\%$ of matched sets have biases $>$ {\DW $26.88$}\\
        the 359th largest hidden bias $\le$ {\DW $11.66$} & {\DW $0.04985813$} & 
        at least $70\%$ of matched sets have biases $>$ {\DW $11.66$} \\
        \bottomrule
    \end{tabular}%
    }
\end{table}

We
generalize Rosenbaum's sensitivity analysis and consider bounding quantiles of hidden biases instead of merely the maximum. 
In particular, we show that the $p$-value for the null of no effect is significant at $5\%$ level 
when the 52nd largest hidden bias from all sets is bounded by {\DW $72.52$}
and when the 359th largest hidden bias is bounded by {\DW$11.66$}, respectively. 
Equivalently, 
to attribute the observed smoking-cadmium association to unmeasured confounding, at least $52/512 \approx 10\%$ of the matched sets need to have biases greater than {\DW $72.52$}, 
and 
at least  $359/512 \approx 70\%$ of the matched sets need to have biases greater than {\DW $11.66$}. 
Table \ref{tab:motivate} also shows the sensitivity analysis for other quantiles of hidden biases. 
These additional results on quantiles of hidden biases strengthen the evidence for a causal effect of smoking on the blood cadmium level, and they are
not sensitive to extreme hidden biases in a small proportion of the matched sets, which 
mitigates the limitation of the conventional sensitivity analysis.

Given the rich results from the proposed sensitivity analysis (see, e.g., Table \ref{tab:motivate}), 
one may be concerned about which quantiles to focus on and how to interpret the results from the analyses based on multiple quantiles of hidden biases.  
In practice, 
we suggest investigators to conduct sensitivity analyses for all quantiles of hidden biases, because they are simultaneously valid without the need of any adjustment due to multiple analyses.
Taking Table \ref{tab:motivate} as an example,
to attribute the observed smoking-cadmium association to unmeasured confounding, we need to believe that, with $95\%$ confidence level,  all the statements in the last column of Table \ref{tab:motivate} hold simultaneously, 
i.e., there are at least one, $10\%$, $30\%$, $50\%$ and $70\%$ of the matched sets having hidden biases greater than {\DW $82.44, 72.52,46.90,26.88$, and $11.66$}, respectively.
Equivalently, {\add the conclusions of the study are sensitive only if the investigator believes that all the statements in the last column of Table \ref{tab:motivate} are likely to hold simultaneously,} i.e., there are small proportions of matched sets having huge hidden biases and also moderate proportions of sets having moderate hidden biases.
Therefore, 
our 
sensitivity 
analysis on all quantiles of hidden biases 
is 
actually
a free lunch added to the conventional sensitivity analysis. 

\subsection{Our contribution}
We propose a new approach to {\lxr address} the 
{\lxr overly conservative}
issue of the conventional sensitivity analysis. 
Instead of considering only the maximum hidden bias, 
we propose sensitivity analysis based on quantiles of hidden biases from all matched sets, which can be more robust in the presence of heterogeneous and extreme hidden biases.  
{\lxr In particular, we consider the constraint that the $(1-\beta)$th quantile of hidden biases is bounded by a constant for a given $\beta \in [0,1]$. This constraint is equivalent to stating that at least $(1-\beta)$ proportion of matched sets have hidden biases within this bound, while the remaining $\beta$ proportion can have unbounded hidden biases.}
Our proposal is indeed 
{\lxr in}
the same spirit as \citet[][Section 4]{Rosenbaum1987}
for matched pair studies. We generalize it to matched observational studies allowing more than one control (or treated) units and varying set sizes. 
Moreover, 
we show that the sensitivity analysis 
{\lxr across}
all quantiles of hidden biases can be simultaneously valid. 

We first construct (asymptotic) valid $p$-values
for testing any sharp null hypothesis with bounds on a certain quantile of hidden biases from all matched sets. 
The proposed approach 
works for general outcomes, general test statistics and general matched studies. 
It 
is also computationally efficient, where the involved optimizations have closed-form solutions.

We then provide an equivalent view of Rosenbaum's sensitivity analysis framework and 
demonstrate that our sensitivity analysis on all quantiles of hidden biases
is simultaneously valid. 
Specifically, the test for the null of no causal effect under a certain constraint on the hidden biases is equivalent to the test for whether
the true hidden biases satisfy the constraint, 
assuming that the treatment has no causal effect. 
So 
finding the minimum amount of hidden biases 
that renders the $p$-value insignificant
is equivalent to 
performing {\lxr a standard test inversion} to construct confidence sets for the true hidden biases.
For example, 
the cutoff $\hat{\Gamma}_0$ from the conventional sensitivity analysis is also an (asymptotic) lower confidence limit for the maximum of true hidden biases from all matched sets, assuming that the null of no effects holds. 
Our sensitivity analysis for quantiles of hidden biases can also be viewed as  
constructing confidence sets for quantiles of the true hidden biases. We show that the confidence sets for all quantiles of true hidden biases are simultaneously valid, i.e., the intersection of them can be viewed as an $I$-dimensional confidence set for all the true hidden biases. 
Therefore, our sensitivity analysis is 
a free lunch added to the conventional Rosenbaum's sensitivity analysis, 
and it helps yield more robust and powerful causal conclusions. 

{\lxr 
Finally, we extend the proposed sensitivity analysis for quantiles of hidden biases, including the conventional one for the maximum hidden bias, to test bounded null hypotheses \citep{li2020quantile} that, for example, assume the treatment is not harmful for every unit.}

The paper proceeds as follows. 
Section \ref{sec:notation} introduces the framework and reviews Rosenbaum's sensitivity analysis. 
Section \ref{sec:pair} studies 
sensitivity analysis for quantiles of hidden biases in matched pair studies.
Section \ref{sec:set} 
extends to 
general matched studies. 
Section \ref{sec:application} applies the proposed method to real data. 
Section \ref{sec:discussion} concludes 
with extensions and discussions.

\section{Framework and Notation}\label{sec:notation}
\subsection{Matched sets, potential outcomes and treatment assignments}
Consider a matched observational study with $I$ matched sets, 
where 
each set $i$ contains $1$ treated unit and $n_i-1$ control units; the approach can be extended straightforwardly to cases where each set has either one treated or one control unit
\citep[][pp. 161--162]{Rosenbaum02a}.
We invoke the potential outcome framework \citep{Neyman23a, Rubin:1974wx} 
and assume the stable unit treatment value assumption \citep{rubin1980randomization}.  
For each unit $j$ in matched set $i$, 
let $Y_{ij}(1)$ and $Y_{ij}(0)$ denote the treatment and control potential outcomes, 
$\bs{X}_{ij}$ denote the pretreatment covariate vector,
and 
$Z_{ij}$ denote the treatment assignment, 
where $Z_{ij}=1$ if the unit 
receives 
treatment and 0 otherwise. 
The corresponding observed outcome 
is then 
$Y_{ij} = Z_{ij} Y_{ij}(1) + (1-Z_{ij}) Y_{ij}(0)$. 
Let 
$\bs{Y}(1) = (Y_{11}(1), Y_{12}(1), \ldots, Y_{In_I}(1))^\top$, 
$\bs{Y}(0) = ( Y_{11}(0), Y_{12}(0), \ldots, Y_{In_I}(0) )^\top$, 
$\bs{Z} = ( Z_{11}, Z_{12}, \ldots, Z_{In_I} )^\top$
and 
$\bs{Y} = ( Y_{11}, Y_{12}, \ldots, Y_{In_I})^\top$ 
denote the treatment potential outcome, control potential outcome, treatment assignment and observed outcome vectors 
for all $N = \sum_{i=1}^I n_i$ units. 
Let {\lxr $\mathcal{F} = \{(Y_{ij}(1), Y_{ij}(0), \bs{X}_{ij}): 1\le i \le I, 1\le j \le n_i\}$}, and let  
$\mathcal{Z} = \{\bs{z}\in \{0,1\}^N: \sum_{j=1}^{n_i} z_{ij}=1, 1\le i\le I\}$ denote the set of treatment assignments such that each matched set contains exactly one treated unit.

Throughout the paper, 
we will conduct the 
finite population inference, 
also called the
randomization-based inference \citep{lidingclt2016, ShiLi2024}. 
That is, all the potential outcomes and covariates
are viewed as fixed constants (or equivalently being conditioned on), 
and the randomness in the observed data 
comes solely from the random assignment $\bs{Z}$. 
Thus, 
the distribution of $\bs{Z}$, often called the treatment assignment mechanism, 
is crucial for statistical inference, 
as discussed below. 

\subsection{Treatment assignment mechanism and hidden biases due to confounding}\label{sec:mechanism}

If all confounding has been considered in matching and matching is exact, 
then units within the same matched set have the same 
probability to receive treatment, 
under which we can analyze the data as if they were from a matched randomized experiment. 
However, in practice, matching may be inexact, 
and more importantly, there may exist
unmeasured confounding. 
We will therefore consider sensitivity analysis to assess the robustness of causal conclusions to hidden biases.  
Below we first characterize the hidden bias in the treatment assignment within each matched set.

{\lxr We assume that, conditional on all potential outcomes, covariates, and $\bs{Z} \in \mathcal{Z}$ (abbreviated as conditional on $(\mathcal{F}, \mathcal{Z})$), the treatment assignments are mutually independent across all matched sets.
Consequently, the treatment assignment $(Z_{i1}, Z_{i2}, \ldots, Z_{in_i})$ within each matched set $i$ follows a multinomial distribution
and the distribution of the whole assignment vector $\bs{Z}$ has the following equivalent forms:
\begin{align}\label{eq:true_treat_assign_mech}
    \pr\left( 
    \bs{Z} = \bs{z} \mid \mathcal{F}, \mathcal{Z}
    \right)
    & 
    = 
    \prod_{i=1}^I  \prod_{j=1}^{n_i} (p_{ij}^\star)^{z_{ij}}
    = 
    \prod_{i=1}^I 
    \frac{\exp(\gamma_i^{\sta} \sum_{j=1}^{n_i} z_{ij} u^\sta_{ij})}{\sum_{j=1}^{n_i} \exp(\gamma_i^{\sta} u^\sta_{ij})}. 
\end{align}
Below we explain the notation in \eqref{eq:true_treat_assign_mech}. 
For any $i$ and $j$, $p_{ij}^\star$ denotes the probability that unit $j$ in matched set $i$ receives the treatment, given that the set has exactly one treated unit. 
Obviously, $\sum_{j=1}^{n_i} p_{ij}^\star = 1$ for all $i$. 
The last equivalent form in \eqref{eq:true_treat_assign_mech} 
aligns with the typical form in Rosenbaum’s sensitivity analysis; see, e.g., \citet[][Section 2.2]{Rosenbaum2007m} for the derivation.
Specifically, $\gamma_i^\star = \log \Gamma_i^\star$ with
\begin{align}\label{eq:bias}
    \Gamma_i^\star = 
    \max_{1\le j \le n_i} p_{ij}^\star/ \min_{1\le k\le n_i} p_{ik}^\star
    \ge 1, 
\end{align}
and $u_{ij}^\star = ( \log p_{ij}^\star - \min_{1\le k\le n_i} \log p_{ik}^\star )/\gamma_i^\star \in [0,1]$ with $\min_{j} u_{ij}^\sta = 0$ and $\max_{j} u_{ij}^\sta = 1$ for all $i$\footnote{
Note that, 
if $\min_{1\le k\le n_i} p_{ik}^\star = 0$ for some $i$, then the corresponding $\Gamma_i^\star, \gamma_i^\star, u_{ij}^\star$ may not be well-defined. 
For descriptive convenience, we will let $\Gamma_i^\star = \gamma_i^\star = \infty$ and assume that there still exist $u_{ij}^\star$'s in $[0,1]$ such that \eqref{eq:true_treat_assign_mech} holds.}.} 
Obviously, if $\Gamma_i^\sta= 1$, then all units within the matched set $i$ have the same chance to receive the treatment. Otherwise, 
the chance depends on $u_{ij}^\sta$'s, where  
units with larger $u_{ij}^\sta$'s are more likely to receive the treatment.
Moreover, some units are $\Gamma_i^\star$ times more likely to receive the treatment than other units within the same matched set. 
Therefore, 
we can understand $\Gamma_i^\star$ as the hidden bias (i.e., the strength of hidden confounding) for matched set $i$, 
and intuitively view the $u_{ij}^\sta$'s as the hidden confounders associated with all the units.

\begin{rmk}\label{rmk:match}
The model in \eqref{eq:true_treat_assign_mech} can be justified when we have independently sampled units in the original data prior to matching and matching is exact, under which the hidden bias for each matched set $i$ is equivalently the maximum odds ratio between the 
treatment probabilities 
of any two units in the set, i.e., 
$\Gamma_{i}^{\sta} = \max_{1\le j\le n_i} \pi_{ij}/(1-\pi_{ij})/\{ \min_{1\le k \le n_i}\pi_{ik}/(1-\pi_{ik}) \}$ for all $i$, 
where $\pi_{ij} = \Pr(Z_{ij} = 1 \mid \bs{X}_{ij}, Y_{ij}(1), Y_{ij}(0))$ is the 
treatment probability
for unit $j$ in set $i$. 
If the covariates contain all the confounding, i.e., there exists a function $\lambda(\cdot)$ such that $\pi_{ij} = \lambda(\bs{X}_{ij})$ for any $i,j$, 
then $\Gamma_{i}^{\sta}$'s will reduce to $1$, 
and the matched observational study becomes equivalent to a matched randomized experiment. We relegate the details to the supplementary material. 
When matching is inexact, the conditional mutual independence of treatment assignments across all matched sets as in \eqref{eq:true_treat_assign_mech} may not hold, an issue that has been noted by, e.g., \citet{Pashley2021} and \citet{pimentel2022covariate}. 
How to take into account the matching algorithm into Rosenbaum's sensitivity analysis framework is still an open problem, and is beyond the scope of this paper; see \citet{li2024sensitivity} for some recent progress. 
In the paper we will focus on model \eqref{eq:true_treat_assign_mech} for the treatment assignment mechanism of the matched units. 
\end{rmk}

In practice, the actual treatment assignment mechanism (including both true hidden biases $\Gamma_i^\sta$'s and the associated hidden confounding $u_{ij}^\sta$'s in \eqref{eq:true_treat_assign_mech}) is generally unknown. 
Motivated by \eqref{eq:true_treat_assign_mech}, we will
consider the following class of models for the treatment assignments.

\begin{definition}\label{def:sen_model}
The treatment assignment is said to follow a sensitivity model with hidden bias $\bs{\Gamma} = (\Gamma_1, \ldots, \Gamma_I)^\top \in [1,\infty)^I$ and hidden confounding $\bs{u} = (u_{11}, \ldots, u_{In_I}) \in [0,1]^N$, denoted by $\bs{Z} \sim \mathcal{M}(\bs{\Gamma}, \bs{u})$, if and only if 
the distribution of $\bs{Z}$ 
has the following form:
\begin{align}\label{eq:assign_mechanism_u}
\pr(\bs{Z} = \bs{z} \mid 
\mathcal{F}, \mathcal{Z}) = \prod_{i=1}^I \frac{\exp(\gamma_i \sum_{j=1}^{n_i} z_{ij} u_{ij})}{\sum_{j=1}^{n_i} \exp(\gamma_i u_{ij})}, 
\qquad
\text{where $\gamma_i = \log \Gamma_i$ for $1\le i \le I$. }
\end{align}
\end{definition}

In practice, 
we often conduct sensitivity analysis assuming that the true hidden biases are bounded from the above, under which the assignment mechanism must fall in a class of sensitivity models in Definition \ref{def:sen_model}.  
Let $\bs{\Gamma}^{\sta} = (\Gamma^{\sta}_1, \ldots, \Gamma^{\sta}_I)^\top$ be the vector of true hidden biases in all $I$ matched sets. 
For any $\bs{\Gamma} \in \mathbb{R}^I$, 
we say $\bs{\Gamma}^{\sta} \vecle \bs{\Gamma}$ if and only if 
$\Gamma_i^{\sta} \le \Gamma_i$ for all $1\le i \le I$.
From \eqref{eq:true_treat_assign_mech},
we can verify that
$\bs{\Gamma}^\sta \vecle \bs{\Gamma}$ 
if and only if 
there exists $\bs{u} = (u_{11},\ldots, u_{In_I}) \in [0,1]^{N}$ such that $\bs{Z} \sim \mathcal{M}(\bs{\Gamma}, \bs{u})$. 
Thus, we say that the treatment assignment follows a sensitivity model with hidden biases at most $\bs{\Gamma}$ or $\bs{\Gamma}^\sta \vecle \bs{\Gamma}$ if $\bs{Z} \sim \mathcal{M}(\bs{u}, \bs{\Gamma})$ for some $\bs{u}\in [0,1]^N$. 
Similar to the discussion after \eqref{eq:true_treat_assign_mech}, we will allow $\bs{\Gamma}$ to have infinite elements in Definition \ref{def:sen_model}, 
so that the sensitivity model can allow some units to have probability 0 or 1 for receiving treatment.

\subsection{Randomization test and sensitivity analysis}\label{sec:test_rand_sen}

We consider Fisher's null hypothesis of no causal effect: 
\begin{align}\label{eq:H0}
    H_0: Y_{ij}(0)=Y_{ij}(1) \text{ for }  1\le j\le n_i \text{ and } 1\le i \le I 
    \text{ or equivalently } \bs{Y}(0) = \bs{Y}(1). 
\end{align}
We will discuss extensions to other null hypotheses at the end of the paper. 
Under Fisher's $H_0$, we are able to impute the potential outcomes for all units based on the observed data,  
i.e., 
$\bs{Y}(1) = \bs{Y}(0) = \bs{Y}$. 
We then consider testing $H_0$ using a general test statistic of the following form: 
\begin{align}\label{eq:test_stat}
    T = t(\bs{Z}, \bs{Y}(0)) = \sum_{i=1}^I T_i = \sum_{i=1}^I \sum_{j=1}^{n_i} Z_{ij} q_{ij}, 
\end{align}
where $T_i = \sum_{j=1}^{n_i} Z_{ij} q_{ij}$ denotes the test statistic for set $i$, 
and 
$q_{ij}$'s can be arbitrary 
functions of 
$\bs{Y}(0)$. 
The test statistic in \eqref{eq:test_stat} can be quite general and requires only an additive form: it is the summation of some set-specific statistics; 
see, e.g., m-statistics suggested by \citet{Rosenbaum2007m} with details in the supplementary material. 
{\lxr 
Throughout the remaining of the paper, we will conduct inference conditional on $(\mathcal{F}, \mathcal{Z})$, under which the randomness in the test statistic $T$ comes solely from the random  
assignment 
$\bs{Z}$. For conciseness, we will make the conditioning implicit. 
}

In a matched randomized experiment, there is no hidden bias in treatment assignment, i.e., $\bs{\Gamma}^\sta=\bs{1}$, and thus $\bs{Z}$ is uniformly distributed over $\mathcal{Z}$. 
Consequently, the null distribution of $T$ 
under $H_0$ is known exactly, whose tail probability has the following form:
\begin{align*}
    \pr\left( T \ge c \right)
    = 
    \pr\left\{ t(\bs{Z}, \bs{Y}(0)) \ge c \right\}
    = 
    |\{\bs{z} \in \mathcal{Z}: t(\bs{z}, \bs{Y}(0)) \ge c\}|/
    |\mathcal{Z}|
    ,
\end{align*}
where $|\cdot|$ denotes the cardinality of a set. 
A valid $p$-value for testing $H_0$ can then be achieved by evaluating the tail probability at the observed value of the test statistic. 

In a matched observational study, 
however, there may exist hidden confounding such that $\bs{\Gamma}^\sta\vecge \bs{1}$. 
\citet{Rosenbaum02a} 
assumed 
that hidden biases in treatment assignments from all matched sets are bounded by some constant $\Gamma_0$, i.e., $\bs{\Gamma}^\sta \vecle \Gamma_0 \bs{1}$, 
and proposed valid tests for $H_0$ at various values of $\Gamma_0$. 
When $\Gamma_0 > 1$, 
unlike matched randomized experiments,
the null distribution of the test statistic $T$ is no longer known exactly. 
In particular,  
when $\bs{\Gamma}^\sta \vecle \Gamma_0 \bs{1}$, 
$\bs{Z} \sim \mathcal{M}(\Gamma_0 \bs{1}, \bs{u})$ for some unknown hidden confounding $\bs{u}$, 
and the tail probability of $T$ is: 
\begin{align*}
    \pr\left( T\ge c \right) 
    = 
    \pr_{\bs{Z} \sim \mathcal{M}(\Gamma_0 \bs{1}, \bs{u})} \left(t(\bs{Z}, \bs{Y}(0)) \ge c \right)
    =
    \sum_{\bs{z} \in \mathcal{Z}}
    \left[ 
    \I\left\{ t(\bs{z}, \bs{Y}(0)) \ge c \right\}
    \prod_{i=1}^I \frac{\exp(\gamma_0 \sum_{j=1}^{n_i} z_{ij} u_{ij})}{\sum_{j=1}^{n_i} \exp(\gamma_0 u_{ij})} 
    \right], 
\end{align*}
where $\gamma_0 = \log \Gamma_0$. 
\citet{Rosenbaum02a} derived an (asymptotic) upper bound of the tail probability over all possible $\bs{u} \in [0,1]^{N}$, 
whose value evaluated at the observed test statistic 
then provides an (asymptotic) valid $p$-value for testing $H_0$.

In practice, 
we 
investigate how the $p$-value for $H_0$ varies with $\Gamma_0$. 
Note that the class of assignment mechanisms satisfying $\bs{\Gamma}^\sta \vecle \Gamma_0 \bs{1}$ increases with $\Gamma_0$. 
Thus, the $p$-value will also increase with $\Gamma_0$. 
In practice, we often report the maximum value of $\Gamma_0$, denoted by $\hat{\Gamma}_0$, such that the $p$-value is significant, and conclude that the causal effect is robust to hidden biases at magnitude $\hat{\Gamma}_0$ but becomes sensitive when 
hidden biases 
are likely to 
exceed 
$\hat{\Gamma}_0$. 
Intuitively, the larger the $\hat{\Gamma}_0$, the more robust the inferred causal effect 
is to unmeasured confounding. 
Such a value of $\hat{\Gamma}_0$ is also termed as sensitivity value in \citet{senvalue2019}. 

\subsection{Another way to understand Rosenbaum's sensitivity analysis}\label{sec:equiv_sen}

We give another way to understand the sensitivity value $\hat{\Gamma}_0$. 
We can always view a valid $p$-value for testing $H_0$ under the sensitivity model $\bs{\Gamma}^\sta \vecle \Gamma_0 \bs{1}$
equivalently as a valid $p$-value for testing the null hypothesis $\bs{\Gamma}^\sta \vecle \Gamma_0 \bs{1} \Leftrightarrow \max_{1\le i\le I} \Gamma_i^\sta \vecle \Gamma_0$ assuming that $H_0$ holds. 
Thus, $\hat{\Gamma}_0$ is equivalently a $1-\alpha$ lower confidence limit of $\max_{1\le i\le I} \Gamma_i^\sta$ from the 
{\lxr standard}
test inversion, assuming that $H_0$ holds. 
That is, if $H_0$ holds,
then we are $1-\alpha$ confident that at least one  matched set has hidden bias exceeding $\hat{\Gamma}_0$. 
If $\max_{1\le i \le I} \Gamma_i^\sta> \hat{\Gamma}_0$ is not plausible, then $H_0$ is not likely to hold, indicating a significant causal effect. 
Otherwise, 
the observed treatment-outcome association can be explained by a plausible amount of hidden biases and the inferred causal effect is sensitive to hidden confounding. 

\subsection{Quantiles of hidden biases in all matched sets}\label{sec:quantile}

As discussed 
in Section \ref{sec:motivate}, 
the conventional sensitivity analysis 
focuses 
on the maximum hidden bias from all matched sets and is thus sensitive to extreme hidden biases,
which 
deteriorates its power. 
For instance, it is possible that one matched set has an exceptionally large hidden bias while all other matched sets have small hidden biases, under which $\max_{1\le i \le I} \Gamma_i^\sta>\Gamma_0$ is likely to hold even for a relatively large $\Gamma_0$. 
To make it more robust and powerful,  
we generalize the conventional sensitivity analysis to quantiles of hidden biases across matched sets.  
Let $\Gamma^\sta_{(1)} \le \ldots \le \Gamma^\sta_{(I)}$ be the true hidden biases for all matched sets sorted in an increasing order\footnote{
Note that 
any quantile of the hidden biases across all matched sets (or more precisely their empirical distribution)  must be $\Gamma^\sta_{(k)}$ for some $1\le k\le I$.}. 
We will conduct test for 
$H_0$ under the sensitivity model assuming that the hidden bias at rank $k$ is bounded by some constant $\Gamma_{0}$, i.e., 
$\Gamma^\sta_{(k)} \le \Gamma_0$; 
when $k=I$, 
it 
reduces to the conventional sensitivity analysis. 
Define further $I^\sta(\Gamma_0) = \sum_{i=1}^I \I(\Gamma^\sta_i > \Gamma_0)$ as the number of matched sets with hidden biases exceeding $\Gamma_0$. 
Then $\Gamma^\sta_{(k)} \le \Gamma_0$ is equivalent to $I^\sta(\Gamma_0) \le I - k$. 
Therefore, the sensitivity analysis for $\Gamma^\sta_{(k)}$ can equivalently provide sensitivity analysis for $I^\sta(\Gamma_0)$. Again, these measures, $\Gamma^\sta_{(k)}$'s and $I^\sta(\Gamma_0)$'s, can be more robust than the maximum hidden bias.

From Section \ref{sec:equiv_sen}, 
the sensitivity analyses for $\Gamma^\sta_{(k)}$ and $I^\sta(\Gamma_0)$ equivalently infer confidence sets
for 
them 
given that 
$H_0$ holds. 
In practice, it is often desirable to infer $\Gamma^\sta_{(k)}$ and $I^\sta(\Gamma_0)$ for multiple $k$'s and $\Gamma_0$'s. 
As shown later, our sensitivity analysis for all of them is simultaneously valid, 
without the need of any adjustment due to multiple analyses. 
Because the sensitivity analysis for 
$\Gamma_{(I)}^\sta$ is the same as the conventional one,  
the simultaneous analysis for all the quantiles $\Gamma^\sta_{(k)}$'s as well as 
$I^\sta(\Gamma_0)$'s is actually a free lunch added to the conventional 
sensitivity analysis.

\section{Sensitivity Analysis for Matched Pairs}\label{sec:pair}

In this section 
we study sensitivity analysis 
in a matched pair study, i.e., $n_i = 2$ for all $i$. 
We first consider sensitivity models with pair-specific bounds on the hidden biases $\Gamma_i^\sta$'s. 
We then extend to sensitivity models bounding only a certain quantile of the $\Gamma_i^\sta$'s.
Finally, we demonstrate that our sensitivity analyses for all quantiles $\Gamma^\sta_{(k)}$'s and numbers $I^\sta(\Gamma_0)$'s are simultaneously valid over all $1\le k\le I$ and $\Gamma_0 \ge 1$. 
Let $[1, \infty] \equiv [1, \infty) \cup \{\infty\}$. 

\subsection{Sensitivity analysis with pair-specific bounds on hidden biases}\label{sec:pair_specific}

For any given 
$\bs{\Gamma} = (\Gamma_1, \ldots, \Gamma_I)^\top \in [1, \infty]^I$,
we are interested in conducting valid tests for Fisher's null $H_0$ 
under the sensitivity model with hidden biases at most $\bs{\Gamma}$. 
The exact distribution of the treatment assignment as shown in \eqref{eq:true_treat_assign_mech} depends on hidden confounding and is generally unknown. 
To ensure the validity of the test, 
similar to 
Section \ref{sec:test_rand_sen}, 
we seek 
an upper bound of the tail probability of the test statistic under the sensitivity model with $\bs{\Gamma}^\sta \vecle \bs{\Gamma}$. 
By almost the same logic as the conventional sensitivity analysis \citep[][Chapter 4]{Rosenbaum02a}, 
we can derive the sharp upper bound of the tail probability in closed form.

\begin{theorem}\label{thm:pair_vector_gamma}
Assume 
that 
$H_0$ in \eqref{eq:H0} holds. 
If the true hidden bias $\bs{\Gamma}^\sta$ for all pairs is bounded by $\bs{\Gamma}\in [1, \infty]^I$, i.e., $\bs{\Gamma}^\sta\vecle \bs{\Gamma}$, then $T$ is stochastically smaller than or equal to $\overline{T}(\bs{\Gamma})$, i.e., $\pr (T \ge c) \le \pr(\overline{T}(\bs{\Gamma}) \ge c)$ for $c\in \mathbb{R}$, where $\overline{T}(\bs{\Gamma})$ 
follows the distribution of the test statistic $T = t(\bs{
Z}, \bs{Y}(0))$ in \eqref{eq:test_stat} when 
$\bs{Z}
\sim 
\mathcal{M}(\bs{\Gamma}, \bs{u})$
with 
$u_{i1} = 1 - u_{i2} = \I(q_{i1} \ge q_{i2})$ 
for $1\le i \le I$.
\end{theorem}

Theorem \ref{thm:pair_vector_gamma} is a straightforward extension of \citet[][Chapter 4.3]{Rosenbaum02a}
with identical $\Gamma_i$'s. 
Below we give some intuition. 
The test statistic $T$ is the summation of $T_i = Z_{i1} q_{i1} + Z_{i2} q_{i2}$ over all pairs, 
where $q_{i1}$ and $q_{i2}$ are prespecified constants.  
Due to the constraint $\Gamma_i^\sta \le \Gamma_i$, 
$T_i$ equals $q_{i1}$ (or $q_{i2}$) with probability between $1/(1+\Gamma_i)$ and $\Gamma_i/(1+\Gamma_i)$. 
Thus, 
$T_i$ must be stochastically smaller than or equal to the random variable $\overline{T}_{i}(\Gamma_i)$ that equals $\max\{q_{i1}, q_{i2}\}$ with probability $\Gamma_i/(1+\Gamma_i)$ and $\min\{q_{i1}, q_{i2}\}$ with probability $1/(1+\Gamma_i)$. 
Theorem \ref{thm:pair_vector_gamma} then follows from  the mutual independence of treatment assignments across matched pairs.

\begin{rmk}\label{rmk:pair_gamma_infinite}
In Theorem \ref{thm:pair_vector_gamma} we allow elements of $\bs{\Gamma}$ to take infinite values. 
This is important when we consider sensitivity models bounding only a certain quantile of hidden biases, under which some hidden biases are unconstrained and are allowed to take infinite values. 
\end{rmk}
From Theorem \ref{thm:pair_vector_gamma}, we can immediately obtain a valid $p$-value under the null $H_0$ with true hidden biases satisfying $\bs{\Gamma}^\sta\vecle \bs{\Gamma}$, based on which we can test the significance of treatment effects under various upper bounds on the true hidden biases. 
\begin{corollary}\label{cor:pair_vector_gamma}
If 
Fisher's null 
$H_0$ 
holds and 
the true hidden bias vector $\bs{\Gamma}^\sta$
for all pairs is bounded by $\bs{\Gamma} \in [1, \infty]^I$, i.e., $\bs{\Gamma}^\sta\vecle \bs{\Gamma}$, 
then $\overline{p}_{\bs{\Gamma}} \equiv \overline{G}_{\bs{\Gamma}}(T)$ is a valid $p$-value, i.e., 
$\pr( \overline{p}_{\bs{\Gamma}} \le \alpha ) \le \alpha$ for $\alpha\in (0,1)$, 
where $\overline{G}_{\bs{\Gamma}}(c) = \pr(\overline{T}(\bs{\Gamma}) \ge c)$ 
and $\overline{T}(\bs{\Gamma})$ is defined 
as in Theorem \ref{thm:pair_vector_gamma}.
\end{corollary}
In practice, we can approximate 
$\overline{G}_{\bs{\Gamma}}(\cdot)$ 
using the Monte Carlo method, 
by randomly drawing assignments from the sensitivity model $\mathcal{M}(\bs{\Gamma}, \bs{u})$ 
with hidden confounding $u_{i1} = 1 - u_{i2} = \I(q_{i1} \ge q_{i2})$ for all $i$ and calculating the corresponding values of the test statistic. 

\begin{rmk}
Recently, \citet{heng2020sharpening} considered sensitivity analysis with pair-specific bounds on hidden biases, taking into account interactions between observed and
unobserved confounders. 
Their sensitivity analysis can be viewed as a special case of Theorem \ref{thm:pair_vector_gamma} and Corollary \ref{cor:pair_vector_gamma}, but with carefully specified structure on the upper bound $\bs{\Gamma}$. 
\end{rmk}
\subsection{Sensitivity analysis for quantiles of hidden biases}\label{sec:sen_quantile_pair}

We now consider 
sensitivity analysis assuming that the true hidden bias at rank $k$ is bounded by a constant $\Gamma_0 \ge 1$, i.e., $\Gamma_{(k)}^\sta \le \Gamma_0$. 
This implies that the $k$ matched pairs with the smallest hidden biases, denoted by $\mathcal{I}^\sta_k$, have their $\Gamma_i^\sta$'s at most $\Gamma_0$, 
i.e., $\Gamma_i^\sta \le \Gamma_0$ for $i \in \mathcal{I}^\sta_k$.
There is no restriction on the $I-k$ matched sets with the largest $I-k$ hidden biases, 
and $\Gamma_i^\sta$ for $i\notin \mathcal{I}^\sta_k$ is allowed to be arbitrarily large. 
For any set $\mathcal{I} \subset \{1, 2, \ldots, I\}$,
define 
$\bs{1}_{\mathcal{I}}$ as the indicator vector whose $i$th coordinate is $\I\{i \in \mathcal{I}\}$, 
and $\infty \cdot \bs{1}_{\mathcal{I}}$ as the vector whose $i$th coordinate is infinite if $i \in \mathcal{I}$ and 0 otherwise. 
Define $\mathcal{I}^\complement = \{1, 2,\ldots, I\}\setminus \mathcal{I}$ as the complement of the set $\mathcal{I}$, and  $\bs{1}_{\mathcal{I}}^{\complement} = \bs{1}_{\mathcal{I}^{\complement}} = \bs{1} - \bs{1}_{\mathcal{I}}$ for notational simplicity. 
We can verify that $\Gamma_{(k)}^\sta \le \Gamma_0$ if and only if $\bs{\Gamma}^\sta \vecle \Gamma_0 \cdot \bs{1}_{\mathcal{I}^\sta_k} + \infty \cdot \bs{1}_{\mathcal{I}^\sta_k}^\complement$,
which is further equivalent to $\bs{Z} \sim \mathcal{M}(\Gamma_0 \cdot \bs{1}_{\mathcal{I}} + \infty \cdot \bs{1}_{\mathcal{I}}^\complement, \bs{u})$ for some $\bs{u} \in [0,1]^N$ and some subset $\mathcal{I}\subset \{1,2,\ldots, I\}$ with cardinality $k$.   

To ensure the validity of the test for the null $H_0$ under the sensitivity model, 
we need to find the upper bound of the $p$-value (or the tail probability of $T$) over all possible treatment assignment mechanisms satisfying $\Gamma_{(k)}^\sta \le \Gamma_0$. 
From Theorem \ref{thm:pair_vector_gamma} and the discussion before, 
if $\mathcal{I}_k^\sta$ is known, 
then $T$ can be stochastically bounded 
by $\overline{T}(\Gamma_0 \cdot \bs{1}_{\mathcal{I}^\sta_k} + \infty \cdot \bs{1}_{\mathcal{I}_k^{\sta}}^\complement)$, defined in the same way as in Theorem  \ref{thm:pair_vector_gamma}. 
However, the set $\mathcal{I}^\sta_k$ is generally unknown  and can take $I$ choose $k$ possible values, for which the brute-force enumeration can be NP-hard. 
Fortunately, we can find a closed-form upper bound for the tail probability of $\overline{T}(\Gamma_0 \cdot \bs{1}_{\mathcal{I}} + \infty \cdot \bs{1}_{\mathcal{I}}^\complement)$ over all possible subset $\mathcal{I} \subset \{1, 2, \ldots, I\}$ with cardinality $k$. 
Moreover, this upper bound is sharp.
\begin{theorem}\label{thm:pair_quantile_gamma}
Assume that Fisher's null $H_0$ holds.
For $1\le k \le I$, 
if 
the true hidden bias at rank $k$
is bounded by $\Gamma_0 \in [1, \infty]$, i.e., $\Gamma_{(k)}^\sta \le \Gamma_0$, 
then $T$ 
is stochastically smaller than or equal to $\overline{T}(\Gamma_0; k) = \overline{T}(\Gamma_0 \cdot \bs{1}_{\mathcal{I}_k} + \infty \cdot \bs{1}_{\mathcal{I}_k}^\complement)$, 
i.e., 
$\pr(T\ge c) \le \pr( \overline{T}(\Gamma_0; k) \ge c )$ for $c\in \mathbb{R}$, 
where $\mathcal{I}_{k}$ consists of indices of matched pairs with the $k$ smallest values of $|q_{i1}-q_{i2}|$, 
and $\overline{T}(\Gamma_0 \cdot \bs{1}_{\mathcal{I}_k} + \infty \cdot \bs{1}_{\mathcal{I}_k}^\complement)$ is defined 
as in Theorem \ref{thm:pair_vector_gamma}. 
Specifically, 
$\overline{T}(\Gamma_0; k)$ follows the distribution of 
$T = t(\bs{Z}, \bs{Y}(0))$ in \eqref{eq:test_stat} when 
$\bs{Z}
\sim 
\mathcal{M}(\Gamma_0 \cdot \bs{1}_{\mathcal{I}_k} + \infty \cdot \bs{1}_{\mathcal{I}_k}^\complement, \bs{u})$ with 
$u_{i1} = 1 - u_{i2} = \I(q_{i1} \ge q_{i2})$ for all $i$. 
\end{theorem}
Theorem \ref{thm:pair_quantile_gamma} is 
also implied by \citet[][Theorem 3]{Rosenbaum1987},
since the test statistic in \eqref{eq:test_stat} 
for a matched pair study satisfies the decreasing reflection property introduced there. 
From Theorem \ref{thm:pair_quantile_gamma}, the upper bound of the tail probability of $T$ under $\Gamma_{(k)}^\sta \le \Gamma_0$ is achieved when we make the hidden biases infinitely large for the $I-k$ matched pairs with the largest absolute differences between $q_{i1}$ and $q_{i2}$, and make the hidden biases $\Gamma_0$ for the remaining $k$ pairs. 
This is intuitive because
larger hidden biases for pairs with larger absolute differences 
can provide more increase for the tail probability of $T$.

Theorem \ref{thm:pair_quantile_gamma} immediately provides a valid $p$-value under the null $H_0$ 
with true hidden biases satisfying 
$\Gamma_{(k)}^\sta \le \Gamma_0$, based on which we can test the significance of treatment effects under various upper bounds on the true hidden bias at any rank. 
\begin{corollary}\label{cor:pval_pair_quantile}
If Fisher's null $H_0$ holds   
and the bias at rank $k$ is bounded by $\Gamma_0$, i.e., $\Gamma_{(k)}^\sta \le \Gamma_0$, 
then 
$
\overline{p}_{\Gamma_0; k} \equiv \overline{G}_{\Gamma_0; k}(T)
$
is a valid $p$-value, i.e., 
$\pr( \overline{p}_{\Gamma_0; k} \le \alpha ) \le \alpha$ for $\alpha\in (0,1)$, 
where $\overline{G}_{\Gamma_0; k}(c) = \pr(\overline{T}(\Gamma_0; k) \ge c)$ for all $c\in \mathbb{R}$ and $\overline{T}(\Gamma_0; k)$ is defined 
as in Theorem \ref{thm:pair_quantile_gamma}.
\end{corollary}
\begin{rmk}
We emphasize that
analyses from Theorem \ref{thm:pair_quantile_gamma} and Corollary \ref{cor:pval_pair_quantile}
do not point to a particular subset of matched pairs. 
If we are interested in sensitivity analysis that involves constraints on hidden biases for a particular subset of matched pairs, we can apply Theorem \ref{thm:pair_vector_gamma} and Corollary \ref{cor:pair_vector_gamma} with pair-specific bounds on hidden biases.
\end{rmk}

\subsection{Simultaneous sensitivity analysis for all $\Gamma_{(k)}^\sta$'s and $I^\sta(\Gamma_0)$'s}\label{sec:simultaneous_pair}

From 
Section \ref{sec:equiv_sen}, 
we can equivalently understand 
sensitivity analysis as constructing confidence sets for the true hidden bias $\bs{\Gamma}^\sta$ or a certain function of it, e.g., its maximum or quantiles.  
Assuming that Fisher's null $H_0$ holds,  
$\overline{p}_{\Gamma_0; k}$ in Corollary \ref{cor:pval_pair_quantile} is indeed a valid $p$-value for testing the null hypothesis $\Gamma^\sta_{(k)} \le \Gamma_0 \Longleftrightarrow I^\sta(\Gamma_0) \le I - k$ regarding the true hidden bias. 
{\lxr The standard}
test inversion
can then provide confidence sets for $\Gamma^\sta_{(k)}$ and $I^\sta(\Gamma_0)$. 
Moreover, we can verify that the $p$-value $\overline{p}_{\Gamma_0;k}$ is increasing in $\Gamma_0$ and decreasing in $k$, 
which simplifies the forms of the confidence sets. 
We summarize the results below. 
For simplicity, we define $\Gamma_{(0)}^\sta = 1$ and $\overline{p}_{\Gamma_0;0}= 1$, which are intuitive given that $I^\sta(\Gamma_0) \le I$ always holds. 
\begin{corollary}\label{cor:pair_ci_gamma3}
Assume that Fisher's null $H_0$ holds. 
(i) For any $1\le k \le I$ and $\alpha\in (0,1)$, 
a $1-\alpha$ confidence set for $\Gamma_{(k)}^\sta$ is 
$\{\Gamma_0 \in [1, \infty]: \overline p_{\Gamma_0;k} > \alpha\}$, which is equivalently $(\hat{\Gamma}_{(k)}, \infty]$ or $[\hat{\Gamma}_{(k)}, \infty]$ with $\hat{\Gamma}_{(k)} \equiv  \inf \{\Gamma_0\ge 1: \overline p_{\Gamma_0;k} > \alpha\}$.
(ii)
For any $\Gamma_0\in \mathbb{R}$ and $\alpha\in (0,1)$, 
a $1-\alpha$ confidence set for $I^\sta(\Gamma_0)$ is 
$\{I-k: \overline p_{\Gamma_0;k} > \alpha, 0 \le k \le I\} = \{
\hat{I}(\Gamma_0), \hat{I}(\Gamma_0)+1, \ldots, I
\}$, 
with $\hat{I}(\Gamma_0) \equiv I - \sup \{k: \overline p_{\Gamma_0;k} > \alpha, 0\le k \le I\}$.
\end{corollary}
In Corollary \ref{cor:pair_ci_gamma3}, the lower confidence limit $\hat{\Gamma}_{(I)}$ is the same as the threshold reported in the conventional sensitivity analysis, 
which can be 
sensitive to
extreme hidden biases. 
The additional results from Corollary \ref{cor:pair_ci_gamma3}
on quantiles of hidden biases as well as numbers (or proportions) of pairs with hidden biases exceeding any threshold can provide us more robust causal conclusions. 
Moreover, 
as demonstrated below, 
the sensitivity analysis for all $\Gamma^\sta_{(k)}$'s and $I^\sta(\Gamma_0)$'s are simultaneously valid. 
Therefore, the additional sensitivity analysis for $\Gamma^\sta_{(k)}$'s with $k<I$ and $I^\sta(\Gamma_0)$'s is actually a free lunch added to the conventional sensitivity analysis.

{\lxr
Below we explain the simultaneous validity of the confidence intervals in Corollary \ref{cor:pair_ci_gamma3}. Assume that Fisher's null hypothesis $H_0$ holds, and suppose that the true hidden bias $\Gamma_{(k)}^\star$ lies outside the confidence interval $\{\Gamma_0 \in [1, \infty]: \overline p_{\Gamma_0;k} > \alpha\}$ for some $1 \le k \le I$. This would imply that $\overline p_{\Gamma_{(k)}^\star;k} \le \alpha$ for some $1 \le k \le I$. By the construction of the $p$-value from Theorem \ref{thm:pair_quantile_gamma} and Corollary \ref{cor:pval_pair_quantile}, regardless of the value of $k$, $\alpha \ge \overline p_{\Gamma_{(k)}^\star;k}$ must be greater than or equal to the true tail probability of the test statistic $T$ evaluated at its observed value. 
This occurs with probability at most $\alpha$, due to the validity of usual randomization test. 
We summarize the results in the theorem below.}

\begin{theorem}\label{thm:pair_joint_ci}
Assume Fisher's null $H_0$ holds. 
For any $\alpha \in (0,1)$, the intersection of all $1-\alpha$ confidence sets for the quantiles of true hidden biases $\Gamma_{(k)}^\sta$'s, viewed as a confidence set for $\bs{\Gamma}^\sta$, is the same as that of all $1-\alpha$ confidence sets for the 
$I^\sta(\Gamma_0)$'s.  
Moreover, it is indeed a $1-\alpha$ confidence set with probability at least $1-\alpha$ to cover the true hidden bias $\bs{\Gamma}^\sta$.
\end{theorem}

From Theorem \ref{thm:pair_joint_ci}, in practice, we suggest to conduct sensitivity analysis for all quantiles of hidden biases simultaneously, in addition to the conventional analysis. 
These additional analyses provide more robust and powerful sensitivity analysis for hidden biases, 
without any reduction on the confidence level. 
Specifically, to conclude that a study's conclusion is sensitive to hidden confounding,  
we need to believe that the quantiles of true hidden biases $\Gamma_{(k)}^\sta$'s are uniformly bounded from below by our derived thresholds $\hat{\Gamma}_{(k)}$'s; 
this can greatly strengthen causal evidences compared to the conventional analysis on the maximum hidden bias.
Moreover, the inference results on all $\Gamma_{(k)}^\sta$'s as well as $I^\sta(\Gamma_0)$'s can be easily visualized and interpreted. 
In particular, we can draw confidence intervals for all $\Gamma_{(k)}^\sta$'s in the same figure, and simply count the number of intervals that does not contain any given threshold $\Gamma_0$ to get the lower confidence limit for $I^\sta(\Gamma_0)$. 
We will illustrate this in detail 
in the later application. 

\begin{rmk}
    {\lxr Below we give more intuition for the gain from Theorem \ref{thm:pair_joint_ci}. 
    When performing sensitivity analysis for multiple quantiles of hidden biases, 
    we essentially impose the constraint that, for some $1 \leq k \leq I$, the hidden bias at rank $k$ is bounded by a threshold $\Gamma_{k}$; noting that the threshold itself can vary with $k$ as well. 
    Such a constraint on unmeasured confounding is weaker than usual constraints on a single quantile, say the maximum, and can be robust to extreme hidden confounding. 
    Once $\max_{1\le k \le I} \overline{p}_{\Gamma_k;k} \le \alpha$, we will then have stronger evidence for treatment effects.}
\end{rmk}

\begin{rmk}
    {\add
    From the alternative understanding of sensitivity analysis explained in Section \ref{sec:equiv_sen}, 
    we can equivalently view the $p$-value $\overline{p}_{\bs{\Gamma}}$ in Corollary \ref{cor:pair_vector_gamma} as a valid $p$-value for testing whether the true hidden biases satisfy $\bs{\Gamma}^\sta\vecle \bs{\Gamma}$, assuming that the treatment has no effect. 
    By 
    {\lxr the standard}
    test inversion, 
    for any $\alpha\in (0,1)$, 
    the set of $\bs{\Gamma}\in [1, \infty]^I$ such that  $\overline{p}_{\bs{\Gamma}}>\alpha$ is then a confidence set for $\bs{\Gamma}^\sta$. 
    However, such a test inversion is generally computationally infeasible. 
    Besides, the resulting $I$-dimensional confidence set may not be easy to understand and interpret. 
    The confidence set we constructed in Theorem \ref{thm:pair_joint_ci} is actually a superset of this computationally infeasible one.
    Importantly, 
    the set in Theorem \ref{thm:pair_joint_ci} 
    can be efficiently computed; see the supplementary material for detailed discussion about its computational complexity. 
    Moreover, it can be easily visualized and interpreted;  
    see Section \ref{sec:application} for details.}
\end{rmk}

\begin{rmk}\label{rmk:average_bias}
Theorem \ref{thm:pair_joint_ci} can also 
provide 
lower confidence limits for average hidden biases of form $\bar{\Gamma}^\sta_g = g^{-1}( I^{-1} \sum_{i=1}^I g(\Gamma_i^\sta) ) 
$ for any 
increasing 
function $g$. 
For example, \citet{HS2017} and \citet{FH2017} 
considered 
$g(\Gamma_i^\sta) = \Gamma_i^\sta/(1+\Gamma_i^\sta)$. 
From Theorem \ref{thm:pair_joint_ci}, 
under $H_0$, 
a $1-\alpha$ lower confidence limit for $\bar{\Gamma}^\sta_g$ is 
$\hat{\Gamma}_g = g^{-1}( I^{-1} \sum_{k=1}^I g(\hat{\Gamma}_{(k)}) )$. 
Moreover, 
such a lower confidence limit is simultaneously valid over arbitrary choices of the increasing function $g$. 
\end{rmk}

\section{Sensitivity Analysis for Matched sets}\label{sec:set}

We now 
study 
general matched sets with $n_i \ge 2$ for all $i$. 
Unlike matched pairs, deriving the exact upper bound for the tail probability of 
$T$ 
under a sensitivity model becomes 
generally 
difficult, both theoretically and computationally. 
We will instead seek asymptotic upper bounds that can lead to asymptotic valid $p$-values, utilizing 
ideas in \citet{Rosenbaum2000sep}. 
We first consider sensitivity models with set-specific bounds on hidden biases, then extend 
to sensitivity models bounding only a certain quantile of hidden biases, 
and 
finally 
demonstrate the simultaneous validity of sensitivity analyses for all quantiles. 
{\add 
For conciseness, we will omit the regularity conditions for asymptotics and relegate them to the supplementary material. 
}

\subsection{Sensitivity analysis with set-specific bounds on hidden biases}\label{sec:sen_gamma_vec}

For any $\bs{\Gamma} = (\Gamma_1,\Gamma_2, \ldots, \Gamma_I)^\top \in [1, \infty]^I$, 
we want to conduct valid test for Fisher's null $H_0$ under the sensitivity model with hidden biases at most $\bs{\Gamma}$. 
We achieve this by seeking an asymptotic upper bound for the tail probability of $T$ in \eqref{eq:test_stat} over all possible assignment mechanisms satisfying $\bs{\Gamma}^\sta \vecle \bs{\Gamma}$. 
Intuitively, 
the test statistic $T=\sum_{i=1}^I T_i$, as a summation of $I$ independent random variables, is asymptotically Gaussian with mean $\E(T) = \sum_{i=1}^I \E(T_i)$ and variance $\Var(T) = \sum_{i=1}^I \Var(T_i)$ as $I\rightarrow \infty$.  
For any $c$ no less than the maximum possible value of $\E(T)$, 
the Gaussian approximation of 
$\pr( T \ge c )$ is increasing in both mean $\E(T)$ and variance $\Var(T)$. 
Moreover, because 
the mean $\E(T)$ generally increases much faster (usually at the order of $I$) compared to the standard deviation $\sqrt{\Var(T)}$ (usually at the order of $\sqrt{I}$ ), 
it is more important to maximize the mean. 
Specifically, 
we should first maximize 
$\E(T)$, and then maximize 
$\Var(T)$ given that the mean has been maximized: 
\begin{equation}\label{eq:max_mean_var_vector}
\mu(\bs{\Gamma}) \equiv
\max_{\bs{u}\in [0,1]^N} \E_{\bs{Z} \sim \mathcal{M}(\bs{\Gamma}, \bs{u})}\{t(\bs{Z}, \bs{Y}(0))\}, 
\quad 
\sigma^2(\bs{\Gamma}) \equiv 
\max_{\bs{u}\in \mathcal{U}} \Var_{\bs{Z} \sim \mathcal{M}(\bs{\Gamma}, \bs{u})}\{t(\bs{Z}, \bs{Y}(0))\}, 
\end{equation}
where $\mathcal{U} = \{\bs{u}\in [0,1]^N: \E_{\bs{Z} \sim \mathcal{M}(\bs{\Gamma}, \bs{u})}\{t(\bs{Z}, \bs{Y}(0))\} = \mu(\bs{\Gamma})\}$ contains the maximizer 
for the mean maximization.
By the mutual independence of assignments across matched sets, it suffices to consider 
each matched set $i$ separately. 
That is, 
we want to maximize $\E(T_i)$ as well as $\Var(T_i)$ given that $\E(T_i)$ is maximized, under the constraint that 
$\Gamma_i^\sta \le \Gamma_i$. 
From \citet{Rosenbaum2000sep}, 
such an optimization has an almost closed-form solution, {\add requiring only searching over $n_i-1$ possible configurations of the hidden confounding; see the supplementary material for details.
Let $\mu_{i}(\Gamma_i)$ and $v_{i}^2(\Gamma_i)$ denote the resulting maximum mean and variance of $T_i$. 
We can then efficiently compute \eqref{eq:max_mean_var_vector} by 
$\mu(\bs{\Gamma}) = \sum_{i=1}^I \mu_i(\Gamma_i)$ and 
$\sigma^2(\bs{\Gamma}) = \sum_{i=1}^I v_i^2(\Gamma_i)$.}

The theorem below 
summarizes the 
asymptotic upper bound for the tail probability of
$T$ under 
$H_0$ and the sensitivity model with hidden biases at most $\bs{\Gamma}$.  
Let $\Phi(\cdot)$ denote
the standard Gaussian distribution function. 

\begin{theorem}\label{thm:set_vector_gamma}
Assume Fisher's null $H_0$ 
holds. 
If 
the true hidden bias $\bs{\Gamma}^\sta$
for all sets is bounded by $\bs{\Gamma} \in [1, \infty]^I$, i.e., $\bs{\Gamma}^\sta\vecle \bs{\Gamma}$, 
and 
certain regularity conditions 
hold, 
then for any $c \ge 0$, 
$$
\limsup_{I \rightarrow \infty}
\pr\big\{ T \ge \mu(\bs{\Gamma}) + c \sigma(\bs{\Gamma}) \big\} \le 
\pr\big\{
\tilde{T}(\bs{\Gamma}) \ge \mu(\bs{\Gamma}) + c \sigma(\bs{\Gamma})
\big\} = 1 - \Phi(c), 
$$
where $\tilde{T}(\bs{\Gamma}) \sim \mathcal{N}( \mu(\bs{\Gamma}), \sigma^2(\bs{\Gamma}) )$ with mean and variance defined as in \eqref{eq:max_mean_var_vector}. 
\end{theorem}

Theorem \ref{thm:set_vector_gamma} generalizes \citet{Rosenbaum2000sep}, who considered the special case where $\bs{\Gamma} = \Gamma_0 \bs{1}$ for some $\Gamma_0\ge 1$, 
{\add and imposes weaker regularity conditions; see the supplementary material for details. Moreover, 
analogous to Remark \ref{rmk:pair_gamma_infinite}, 
Theorem \ref{thm:set_vector_gamma} allows some elements of $\bs{\Gamma}$ to be infinite, which is important  
for the later study on  
quantiles of hidden biases.}

From Theorem \ref{thm:set_vector_gamma}, we can then obtain an asymptotic valid $p$-value 
for testing $H_0$ under the constraint that $\bs{\Gamma}^\sta\vecle \bs{\Gamma}$; see the supplementary material for details.

\subsection{Sensitivity analysis for quantiles of hidden biases}\label{sec:sen_quan_set}

We now consider sensitivity analysis assuming that the hidden bias at rank $k$ is bounded by a certain $\Gamma_0 \ge 1$, i.e., $\Gamma^\sta_{(k)} \le \Gamma_0$. 
Below we 
intuitively explain
how to find an asymptotic upper bound for the tail probability of $T$ under a sensitivity model with hidden bias at rank $k$ bounded by $\Gamma_0$. 
First, by the same logic as in Section \ref{sec:sen_quantile_pair}, 
it suffices to consider sensitivity models with hidden biases at most 
$\Gamma_0 \cdot \bs{1}_{\mathcal{I}} + \infty \cdot \bs{1}_{\mathcal{I}}^\complement$, for all possible subsets  $\mathcal{I}\subset \{1,2,\ldots, I\}$ with cardinality $k$.  
In particular, we can use Theorem \ref{thm:set_vector_gamma} to construct an asymptotic upper bound for $T$ at any given $\mathcal{I}$, 
and then search over all possible $\mathcal{I}$; the latter can be challenging. 

Second, as in Section \ref{sec:sen_gamma_vec}, we use Gaussian approximation to find the optimal $\mathcal{I}$. 
For each $1\le i\le I$, we sort the $q_{ij}$'s for matched set $i$ in an increasing way: $q_{i(1)} \le q_{i(2)}\le \ldots \le q_{i(n_i)}$.
Define
\begin{align}\label{eq:mu_gamma0_I}
\mu(\Gamma_0; \mathcal{I}) & = \mu\left( \Gamma_0 \cdot \bs{1}_{\mathcal{I}} + \infty \cdot \bs{1}_{\mathcal{I}}^\complement \right)
= \sum_{i \in \mathcal{I}} \mu_i({\Gamma}_0) + 
\sum_{i \notin \mathcal{I}} \mu_i(\infty)
= \sum_{i \in \mathcal{I}} \mu_i({\Gamma}_0) + \sum_{i \notin \mathcal{I}} q_{i(n_i)}, 
\\
\label{eq:sigma_gamma0_I}
    \sigma^2(\Gamma_0; \mathcal{I}) & = 
    \sigma^2\left( \Gamma_0 \cdot \bs{1}_{\mathcal{I}} + \infty \cdot \bs{1}_{\mathcal{I}}^\complement \right)
    = \sum_{i \in \mathcal{I}} v_i^2({\Gamma}_0) + 
    \sum_{i \notin \mathcal{I}} v_i^2(\infty) 
    = \sum_{i \in \mathcal{I}} v_i^2({\Gamma}_0), 
\end{align}
where $\mu(\cdot)$, $\sigma^2(\cdot)$, $\mu_i(\cdot)$ and $v^2_i(\cdot)$ are defined as in Section \ref{sec:sen_gamma_vec}. 
Note that when the hidden bias for a matched set $i$ is allowed to be infinity, the maximum mean of $T_i$ is achieved when $T_i$ is constant  $q_{i(n_i)}$. 
This implies that $\mu_i(\infty) = q_{i(n_i)}$ and $v_i^2(\infty) = 0$, and 
illustrates the last equalities in \eqref{eq:mu_gamma0_I} and \eqref{eq:sigma_gamma0_I}. 
{\add 
We then consider the following maximization 
over the set $\mathcal{I}$ with cardinality $k$:}
\begin{align}
\label{eq:mu_sigma_quantile_k_simp}
    \mu(\Gamma_0; k) & = 
    \max_{\mathcal{I}: |\mathcal{I}| = k}
    \mu(\Gamma_0; \mathcal{I}), 
    \qquad 
    \sigma^2(\Gamma_0; k)
    = 
    \max_{\mathcal{I}: |\mathcal{I}| = k, \mu(\Gamma_0; \mathcal{I}) = \mu(\Gamma_0; k) } 
    \sigma^2(\Gamma_0; \mathcal{I}). 
\end{align}
Note that \eqref{eq:mu_sigma_quantile_k_simp} essentially gives the maximum mean and variance (given that mean is maximized) of $T$ over all possible assignment mechanisms satisfying $\Gamma^\sta_{(k)} \le \Gamma_0$. 
Fortunately, the optimal $\mathcal{I}$ for \eqref{eq:mu_sigma_quantile_k_simp} has a closed-form solution. 
Specifically, we rank the $I$ matched sets based on the values of $\mu_i(\infty) - \mu_i(\Gamma_0)$ for all $i$.  
Moreover, 
we rank 
ties 
based on their values of $v_i^2(\Gamma_0)$, 
where 
set with smaller $v_i^2(\Gamma_0)$ has a larger rank.
Let $\mathcal{I}_{\Gamma_0; k}$ denote the set of indices of matched sets with the $k$ smallest ranks.  
We can verify that $\mathcal{I}_{\Gamma_0; k}$ is the maximizer for \eqref{eq:mu_sigma_quantile_k_simp}.

The theorem below shows the asymptotic validity of the above Gaussian upper bound. 

\begin{theorem}\label{thm:set_quantile_gamma}
Assume Fisher's null $H_0$ holds. 
If 
the true hidden bias at rank $k$
is bounded by $\Gamma_0 \in [1, \infty)$, i.e., $\Gamma_{(k)}^\sta \le \Gamma_0$, 
and
certain regularity conditions 
hold, 
then for any $c \ge 0$, 
\begin{align*}
\limsup_{I\rightarrow \infty} 
\pr\left\{
T \ge \mu(\Gamma_0; k) + c \sigma(\Gamma_0; k)
\right\}
& \le 
\pr\left\{
\tilde{T}(\Gamma_0; k) \ge 
\mu(\Gamma_0; k) + c \sigma(\Gamma_0; k)
\right\}
= 
1 - \Phi(c),
\end{align*}
where $\tilde{T}(\Gamma_0; k) \sim \mathcal{N}( \mu(\Gamma_0; k), \sigma^2(\Gamma_0; k) )$ with mean and variance 
defined as in \eqref{eq:mu_sigma_quantile_k_simp}.  
\end{theorem}

From Theorem \ref{thm:set_quantile_gamma}, we can immediately obtain an asymptotic valid $p$-value for testing $H_0$ under the 
the constraint that 
$\Gamma_{(k)}^\sta \le \Gamma_0$; see the supplementary material for details.

\subsection{Simultaneous sensitivity analysis for all $\Gamma_{(k)}^\sta$'s and $I^\sta(\Gamma_0)$'s}\label{sec:simulta_set}

From Section \ref{sec:equiv_sen} and analogous to Section \ref{sec:simultaneous_pair}, we can equivalently understand sensitivity analyses as inferring true hidden biases assuming Fisher's null, and we can utilize the asymptotic valid $p$-value from 
Section \ref{sec:sen_quan_set}
to construct large-sample confidence sets for the quantiles of hidden biases $\Gamma_{(k)}$'s as well as
$I^\sta(\Gamma_0)$'s. 
Importantly, these confidence sets will be simultaneously valid, and the intuition is similar to that in Section \ref{sec:simultaneous_pair}. 
For conciseness, 
we relegate the details to the supplementary material.

\section{Effect of Smoking on Blood Cadmium and Lead Levels}\label{sec:application}

We apply and illustrate our method by evaluating the effect of smoking on the amount of cadmium 
and lead 
in the blood, using the 2005–2006 National Health and Nutrition Examination Survey. 
We use the optimal matching 
in \citet{YuRosenbaum2019}, 
taking into account 
gender, age, race, education level, household income level and body mass index, 
and construct 
a matched comparison with 512 matched sets, each of which contains one daily smoker and two nonsmokers. 
For details on matching, 
see the supplementary material.

We now apply our method to conduct sensitivity analysis for quantiles of hidden biases. 
We use the difference-in-means as the test statistic; as shown in the supplementary material, such a sensitivity analysis is not only valid for Fisher's null of no effect, but also valid for the bounded null that smoking would not increase the level of cadmium or lead.  
In the following, we will 
focus on interpretation based on the bounded null hypothesis. 

Figure \ref{fig:sen_lead}(a) and (b) show the $95\%$ lower confidence limits for all quantiles of hidden biases with cadmium and lead as the outcome, respectively, 
assuming that 
smoking would not increase the blood cadmium or lead level. 
From Figure \ref{fig:sen_lead}(a), the causal effect of smoking on cadmium is quite insensitive to unmeasured confounding. 
Similar to Section \ref{sec:motivate}, 
assuming 
that smoking would not increase the blood cadmium level, then we are $95\%$ confident that the 1st, 52nd, 154th, 256th and 359th largest hidden biases must be greater than {\DW $82.44, 72.52, 46.90, 26.88$, and
$11.66$}, respectively. 
Equivalently, there is at least one, $10\%$, $30\%$, $50\%$ and $70\%$ of matched sets having hidden biases greater than {\DW $82.44, 72.52, 46.90, 26.88$, and
$11.66$}, respectively. 
From Remark \ref{rmk:average_bias}, 
the $95\%$ lower confidence limits for the average hidden biases $\bar{\Gamma}_g^\sta$ with $g(x)$ equal to $x, \log(x)$ and $x/(1+x)$ are, respectively, 32.18, 17.73 and 8.33. 
If we doubt any of these confidence statements, then 
the bounded null
is not likely to hold, suggesting that  
smoking 
may increase blood cadmium levels
for some units. 
Note that these confidence statements require a huge amount of hidden biases 
in a 
substantial
proportion of matched sets, which are not likely 
after matching for the observed confounding. 
We believe that daily smoking has a significant non-zero causal effect on the blood cadmium amount, 
which 
is insensitive to unmeasured confounding. 
\begin{figure}[htb]
    \centering
    \begin{subfigure}[htbp]{0.5\textwidth}
    \centering
        \includegraphics[width=.65\textwidth]{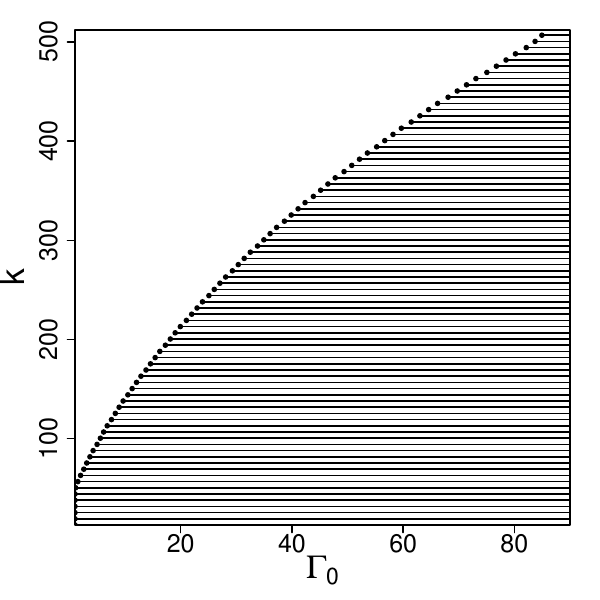}
        \caption{\DW Cadmium}
    \end{subfigure}%
    \begin{subfigure}[htbp]{0.5\textwidth}
    \centering
        \includegraphics[width=.65\textwidth]{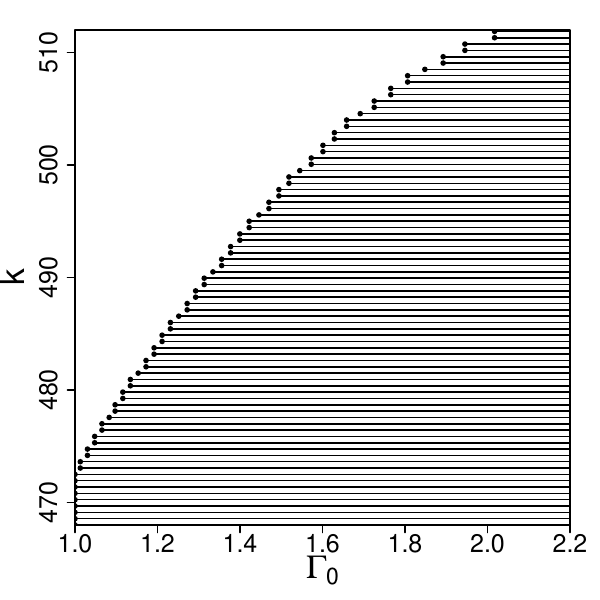}
        \caption{\DW Lead}
    \end{subfigure}
    \caption{$95\%$ lower confidence limits for the hidden biases from the 512 matched sets, assuming smoking has no effect on cadmium (for (a)) or lead (for (b)). }
    \label{fig:sen_lead}
\end{figure}

From Figure \ref{fig:sen_lead}(b), the effect of smoking on lead is moderately insensitive to unmeasured confounding, although much less robust than that for cadmium. 
Specifically, 
assuming that smoking would not increase the lead level, 
the 1st and 26th
largest hidden biases must be greater than {\DW$2.01$ and $1.25$%
}, respectively,
implying that at least one and $5\%$ 
of matched sets have hidden biases greater than {\DW$2.01$ and $1.25$}. 
If we think unmeasured confounding of such strength could not happen, 
then the causal effect of smoking on lead should be significant, 
with the conclusion being 
insensitive to moderate amount of unmeasured confounding.

\section{Conclusion and Discussion}\label{sec:discussion}

We generalize conventional sensitivity analysis for the maximum hidden bias to that for quantiles of hidden biases. 
We demonstrate that the proposed sensitivity analysis for all quantiles of hidden biases is simultaneously valid and is thus a free lunch added to the conventional analysis. 
Furthermore, 
in the supplementary material, 
we show that the proposed sensitivity analysis (including the conventional one) also works for the more general bounded null hypotheses on treatment effects, as long as the test statistic enjoys certain properties \citep{li2020quantile}.

The proposed approach extends straightforwardly to arbitrary sharp and bounded null hypotheses; see the supplementary material for details. 
By inverting tests for a sequence of null hypotheses with various constant effects, we can obtain interval estimates for constant treatment effects under sensitivity models with bounds on quantiles of hidden biases. 
Such confidence intervals, compared to conventional ones with bounds on the maximum hidden bias, can provide 
more robust causal conclusions. 
Analogously, by inverting tests for a sequence of bounded null hypotheses, we can obtain interval estimates for the maximum (or minimum) individual effects under sensitivity models with bounds on quantiles of hidden biases.
It will be interesting to further extend this framework to conduct robust sensitivity analysis for more general weak null hypotheses about, e.g., average effects or effect quantiles \citep{Fogarty:2019aa, li2020quantile, SL22quantile, Chen24, chen2024enhanced}. We leave them for future study.

Determining plausible amounts of hidden biases in a study is an important while challenging problem for most sensitivity analysis approaches. This task typically requires domain knowledge and careful consideration of adjusted pretreatment covariates. \citet{Amplification09} has proposed a two-parameter model for interpreting hidden biases, where a magnitude of $1.5$ corresponds to an unobserved confounding that doubles the odds of treatment and increases the odds of a positive pair difference in responses by a factor of $4$; see \citet[][Chapter 9]{rosenbaum2017observation} for more details. 
In addition, empirical studies can also provide some reference for assessing the sensitivity of a study's conclusion. 
For example, 
as discussed in \citet{rosenbaum2005sensitivity}, 
the study on
the effects of diethylstilbestrol \citep{HUP71} becomes sensitive when $\Gamma_{(I)}^\star$ exceeds $7$, while the study on the effects of coffee \citep{Coffee1973} becomes sensitive when $\Gamma_{(I)}^\star$ exceeds $1.3$. 
Further research is needed to investigate our proposed sensitivity analyses, particularly with respect to specific empirical studies and their general interpretation.

{\lxr 

In the paper we focused on Rosenbaum's sensitivity analysis for matched observational studies. 
There has been growing interest on the marginal sensitivity analysis model for inverse probability weighting estimation of treatment effects \citep{zhaosenIPW2019, Dorn2023, Tan2024}. 
The marginal sensitivity analysis model also imposes a uniform bound on the strength of unmeasured confounding across all units, and thus faces the same critique discussed in Section \ref{sec:match_Rosen_sen}. 
To address this limitation, \citet{zhang2022} recently proposed a sensitivity analysis based on  the average strength of unmeasured confounding. 
It will be interesting to extend our sensitivity analysis for quantiles of unmeasured confounding strength to this context, and we leave it for future study. 
}

\section*{Supplementary Material}
The supplementary material contains additional details for the asymptotic sensitivity analysis, 
the proofs for all theorems, corollaries and propositions, 
and a broader justification of the proposed sensitivity analysis (including the conventional one) for testing bounded null hypotheses.
It also contains discussion on computational complexity of the proposed algorithms and simulation studies for the validity and power of the proposed approach. 
An R package for implementing the proposed sensitivity analysis is available 
at \url{https://github.com/dongxiao-wu/sensitivityq}.

\section*{Acknowledgement}
We thank the Editor, the Associate Editor, and the reviewers for insightful and constructive
comments.

\section*{Funding}

X. L. is partly supported by the National Science Foundation under grant DMS-2400961.

\bibliographystyle{plainnat}
\bibliography{sen}

\newpage
\setcounter{equation}{0}
\setcounter{section}{0}
\setcounter{figure}{0}
\setcounter{example}{0}
\setcounter{proposition}{0}
\setcounter{corollary}{0}
\setcounter{theorem}{0}
\setcounter{lemma}{0}
\setcounter{table}{0}
\setcounter{condition}{0}

\renewcommand {\theproposition} {A\arabic{proposition}}
\renewcommand {\theexample} {A\arabic{example}}
\renewcommand {\thefigure} {A\arabic{figure}}
\renewcommand {\thetable} {A\arabic{table}}
\renewcommand {\theequation} {A\arabic{equation}}
\renewcommand {\thelemma} {A\arabic{lemma}}
\renewcommand {\thesection} {A\arabic{section}}
\renewcommand {\thetheorem} {A\arabic{theorem}}
\renewcommand {\thecorollary} {A\arabic{corollary}}
\renewcommand {\thecondition} {A\arabic{condition}}
\renewcommand {\thermk} {A\arabic{rmk}}

\renewcommand {\thepage} {A\arabic{page}}
\setcounter{page}{1}

\begin{center}
	\bf \LARGE 
	Supplementary Material 
\end{center}

Appendix \ref{sec:add_sen_pair} gives further details for the simultaneous sensitivity analysis in Section \ref{sec:simultaneous_pair}.

Appendix \ref{sec:add_sen_set} gives additional details for the proposed sensitivity analysis 
in Section \ref{sec:set}.

Appendix \ref{sec:bound_null} 
extends the proposed sensitivity analysis to 
bounded null
hypotheses. 

Appendix \ref{sec:complexity} discusses the computational complexity of the proposed sensitivity analysis. 

Appendix \ref{app:pair_vector} contains the proofs of Theorem \ref{thm:pair_vector_gamma} and Corollary \ref{thm:pair_vector_gamma}. 

Appendix \ref{app:pair_quantile} contains the proofs of Theorem \ref{thm:pair_quantile_gamma} and Corollary \ref{cor:pval_pair_quantile}. 

Appendix \ref{app:pair_ci} contains the proofs of Corollary \ref{cor:pair_ci_gamma3} and Theorem \ref{thm:pair_joint_ci}. 

Appendix \ref{app:set_vector} contains the proofs of Theorem \ref{thm:set_vector_gamma} and Corollary \ref{cor:pval_set_vector_gamma}. 

Appendix \ref{app:set_quantile} contains the proofs of Theorem \ref{thm:set_quantile_gamma} and Corollary \ref{cor:sen_pval_set}. 

Appendix \ref{app:set_ci} contains the proofs of Corollary \ref{cor:set_ci_gamma} and Theorem \ref{thm:set_joint_ci} in Appendix \ref{sec:simul_set_app}. 

{\add 
Appendix \ref{sec:bound_null_pair_proof} contains the proofs of Theorems \ref{thm:bounded_null_pair} and \ref{thm:bound_joint_ci} in Appendix \ref{sec:bound_pair}. 

Appendix \ref{sec:bound_null_set_proof} contains the proofs of Theorems \ref{thm:bounded_null_sets_effect_diff_increasing} and \ref{thm:bound_joint_ci_set} in Appendix \ref{sec:bound_set}. 

Appendix \ref{sec:tech_match} contains the technical details for Remark \ref{rmk:match}. 
}

Appendix \ref{sec:simulation} contains simulation studies.

Appendix \ref{app:match} contains the details of matching for the smoking dataset.

\section{Additional details for sensitivity analysis of matched pairs}\label{sec:add_sen_pair}

{\lxr
We give more details for the simultaneous sensitivity analysis of matched pairs in Section \ref{sec:simultaneous_pair}. 
Specifically, we show that the confidence sets for $\Gamma^\sta_{(k)}$'s and $I^\sta(\Gamma_0)$'s can be equivalently viewed as $I$-dimensional confidence sets for the true hidden bias vector $\bs{\Gamma}^\sta$, 
and their intersections over all $k$ and $\Gamma_0$, respectively, are the same. 
These supplement Theorem \ref{thm:pair_joint_ci} in the main paper.

Define 
$
\mathcal{S}_{\Gamma_0;k} = \{\bs{\Gamma}\in [1, \infty]^I: \Gamma_{(k)} \le \Gamma_0 \} = \{\bs{\Gamma}\in \mathbb{R}^I: \I(\Gamma_i > \Gamma_0 ) \le I-k \}$
as the set of possible values of $\bs{\Gamma}^\sta$ under the constraint $\Gamma^\sta_{(k)} \le \Gamma_0 \Longleftrightarrow I^\sta(\Gamma_0) \le I - k$, 
where $\Gamma_{(k)}$ denotes the coordinate of $\bs{\Gamma}$ at rank $k$. 
For $1\le k \le I$ and $\alpha\in (0,1)$,  we can verify that the $1-\alpha$ confidence set for $\Gamma_{(k)}^\sta$ has the following equivalent form as a confidence set for $\bs{\Gamma}^\sta$:
\begin{align*}
\Gamma_{(k)}^\sta \in 
\{\Gamma_0\ge 1: \overline{p}_{\Gamma_0;k} > \alpha\}
\Longleftrightarrow 
\bs{\Gamma}^\sta \in 
\bigcap_{\Gamma_0: \overline{p}_{\Gamma_0;k} \le \alpha} \mathcal{S}_{\Gamma_0;k}^\complement, 
\end{align*}
where $\mathcal{S}_{\Gamma_0;k}^\complement = [1, \infty]^I \setminus \mathcal{S}_{\Gamma_0;k}$.
Analogously, for any $\Gamma_0 \ge 1$ and $\alpha\in (0,1)$, 
the $1-\alpha$ confidence set for $I^\sta(\Gamma_0)$ has the following equivalent form as a confidence set for $\bs{\Gamma}^\sta$: 
\begin{align*}
I^\sta(\Gamma_0) \in \{I-k: \overline{p}_{\Gamma_0;k} > \alpha, 0\le k \le I\} 
\Longleftrightarrow 
\bs{\Gamma}^\sta \in 
\bigcap_{k: \overline{p}_{\Gamma_0;k} \le \alpha} \mathcal{S}_{\Gamma_0;k}^\complement
\end{align*}
It is then immediate that, for any $\alpha\in (0,1)$, the intersection of all confidence sets for $\Gamma_{(k)}^\sta$'s over $1\le k \le K$ is the same as that of all confidence sets for $I^\sta(\Gamma_0)$'s over all $\Gamma_0 \ge 1$, 
and it has the following equivalent forms: 
\begin{align}\label{eq:pair_joint_set}
\bigcap_{k=1}^n 
\bigcap_{\Gamma_0: \overline{p}_{\Gamma_0;k} \le \alpha} \mathcal{S}_{\Gamma_0;k}^\complement 
= 
\bigcap_{\Gamma_0 \ge 1}
\bigcap_{k: \overline{p}_{\Gamma_0;k} \le \alpha} \mathcal{S}_{\Gamma_0;k}^\complement
= 
\bigcap_{\Gamma_0, k: \overline{p}_{\Gamma_0;k} \le \alpha} \mathcal{S}_{\Gamma_0;k}^\complement. 
\end{align}
Importantly, 
as shown in Theorem \ref{thm:pair_joint_ci}, 
the set in \eqref{eq:pair_joint_set} is indeed a $1-\alpha$ confidence set for $\bs{\Gamma}^\sta$, indicating the simultaneous validity of the confidence sets in Corollary \ref{cor:pair_ci_gamma3}.

}

\section{Additional details for sensitivity analysis of matched sets}\label{sec:add_sen_set}

As discussed in Section \ref{sec:set}, finite-sample exact sensitivity analysis becomes challenging for matched sets with $n_i \ge 2$, and we instead invoke finite population asymptotic analysis \citep{lidingclt2016, ShiLi2024}. Specifically, we will embed the $I$ matched sets into a sequence of matched studies with increasing number of sets, and consider asymptotically valid inference as the number of sets $I$ goes to infinity. 
Throughout this section, we give additional details for Section \ref{sec:set}, including regularity conditions guaranteeing the validity of the asymptotic sensitivity analyses.

\subsection{Sensitivity analysis with set-specific bounds on hidden biases}

\subsubsection{Choice of the test statistic}\label{sec:choice_test_stat_appendix}

\citet{Rosenbaum2007m} suggested to use m-statistic with 
$q_{ij} = n_i^{-1} \sum_{l=1}^{n_i} \psi((Y_{ij}(0)-Y_{il}(0))/s)$, 
where 
$\psi(y) = \text{sign}(y) \cdot \kappa \cdot \min[ 1, \max\{0, 
(|y|-\iota)/(\kappa - \iota) 
\} ]$ 
is an odd function 
with truncation at $\kappa$ and inner trimming at $\iota$, 
$s$ is 
the median 
of all the absolute pairwise differences $|Y_{ij}(0) - Y_{ik}(0)|$'s for all $i$ and $j \ne k$. 
When there is no truncation and inner trimming, i.e., $\kappa= \infty$ and $\iota=0$, the scaling term $s$ no longer matters, 
and the m-statistic reduces to a certain weighted difference-in-means statistic; see \eqref{eq:diff_in_means}. 

\subsubsection{Additional details for computation}
We discuss how to compute $\mu_i(\Gamma_i)$ and $v_i^2(\Gamma_i)$ for each matched set $i$ given that the hidden bias for this set is bounded by $\Gamma_i$, using the algorithm proposed by \citet{Rosenbaum2000sep}. 
We sort the $q_{ij}$'s for matched set $i$ in an increasing way: $q_{i(1)} \le q_{i(2)}\le \ldots \le q_{i(n_i)}$, 
and, for $1\le a \le n_i-1$,  define 
\begin{align*}
    \mu_{ia}(\Gamma_i) = 
    \frac{\sum_{j=1}^a q_{i(j)} + \Gamma_i \sum_{j=a+1}^{n_i} q_{i(j)}}{a + \Gamma_i (n_i - a)}, 
    \quad 
    v_{ia}^2(\Gamma_i) = 
    \frac{\sum_{j=1}^a q^2_{i(j)} + \Gamma_i \sum_{j=a+1}^{n_i} q^2_{i(j)}}{a + \Gamma_i(n_i - a)} - \{\mu_{ia}(\Gamma_i)\}^2. 
\end{align*}
Then under the constraint that $\Gamma_i^\sta \le \Gamma_i$, 
the maximum mean of $T_i$ and its maximum variance given its mean is maximized have the following forms: 
\begin{align}\label{eq:max_mean_var_set_i}
    \mu_{i}(\Gamma_i) = 
    \max_{a\in \{1, 2, \ldots, n_i-1\}} \mu_{ia}(\Gamma_i), 
    \quad 
    v_{i}^2(\Gamma_i) = 
    \max_{a: u_{ia}(\Gamma_i) = u_{i}(\Gamma_i)}v_{ia}^2(\Gamma_i), 
    \qquad (1 \le i \le I). 
\end{align}

\subsubsection{Details for regularity conditions}\label{sec:detail_set_specific_bound_regularity}

We introduce the regularity conditions for justifying the asymptotic upper bound of $T$ using the Gaussian distribution with mean $\mu(\bs{\Gamma})$ and variance $\sigma^2(\bs{\Gamma})$, 
under the sensitivity model with hidden biases at most $\bs{\Gamma}$, 
as studied in Section \ref{sec:sen_gamma_vec} of the main paper . 
We consider first the asymptotic Gaussianity for the true distribution of $T$. 
Recall that each set $i$ has $n_i$ units, and the corresponding test statistic $T_i$ takes value in $\{q_{i1}, \ldots, q_{in_i}\}$. 
Let $R_i = \max_{j}q_{ij} - \min_{j}q_{ij}$ be the range of $T_i$. 
The following regularity condition, motivated by the Lindeberg's condition, 
ensures the asymptotic Gaussianity of the true distribution of 
$T$. 
\thmadj
\begin{condition}\label{cond:clt}
As $I\rightarrow \infty$, 
$(\max_{1\le i\le I}R_i^2) /\{ \sum_{i=1}^I R_i^2/(n_i\Gamma_i^\sta)^3 \} \rightarrow 0$.
\end{condition}
\thmadj
We consider then the asymptotic upper bound for the tail probability of $T$. 
For each matched set $i$, 
let $\mu_i$ and $v_i^2$ denote the true mean and variance of $T_i$. 
Recall that $\mu_i(\Gamma_i)$ and $v_i^2(\Gamma_i)$ in \eqref{eq:max_mean_var_set_i} are the maximum mean and variance (given mean is maximized) of $T_i$ when the hidden bias for set $i$ is bounded by $\Gamma_i$.   
For any $\bs{\Gamma} = (\Gamma_1, \Gamma_2, \ldots, \Gamma_I)\in [0,1]^I$, 
we introduce $\mathcal{A}_{\bs{\Gamma}} = \{i:v_i^2>v_i^2(\Gamma_i), 1\le i \le I\}$
to be the set of indices for matched sets whose true variances are greater than the maximized variances, 
and 
${\Delta}_{\mu}(\bs{\Gamma}) = |\mathcal{A}_{\bs{\Gamma}}|^{-1} \sum_{i \in \mathcal{A}_{\bs{\Gamma}}}\{\mu_i(\Gamma_i)-\mu_i\}$ 
and 
${\Delta}_{v^2}(\bs{\Gamma}) = |\mathcal{A}_{\bs{\Gamma}}|^{-1} \sum_{i \in \mathcal{A}_{\bs{\Gamma}}}\{v_i^2 - v_i^2(\Gamma_i)\}$ 
to denote the average differences between true and maximized means and variances of the $T_i$'s for matched sets in $\mathcal{A}_{\bs{\Gamma}}$. 
By definition, if the true hidden biases are indeed bounded by $\bs{\Gamma}$, then $\mu_i$ must be strictly less than $\mu_i(\Gamma_i)$ for all $i\in \mathcal{A}_{\bs{\Gamma}}$; 
otherwise, $\mu_i = \mu_i(\Gamma_i)$ and $v_i^2>v_i^2(\Gamma_i)$ for some $i$, violating the construction in \eqref{eq:max_mean_var_set_i}. 
Consequently, when $\bs{\Gamma}^\sta \vecle \bs{\Gamma}$ and $\mathcal{A}_{\bs{\Gamma}}$ is nonempty, both ${\Delta}_{\mu}(\bs{\Gamma})$ and ${\Delta}_{v^2}(\bs{\Gamma})$ are positive. 
The following condition assumes that ${\Delta}_{\mu}(\bs{\Gamma})$ is not too small compared to ${\Delta}_{v^2}(\bs{\Gamma})$ as $I\rightarrow \infty$, which essentially rigorizes our intuition that maximizing mean is more important than maximizing variance as discussed 
in Section \ref{sec:sen_gamma_vec}. 
For descriptive simplicity, we define the ratio ${\Delta}_{\mu}(\bs{\Gamma})/{\Delta}_{v^2}(\bs{\Gamma})$ to be infinity
when $\mathcal{A}_{\bs{\Gamma}}$
is an empty set. 
\thmadj
\begin{condition}\label{cond:sen_Gamma_vec_asymp}
As $I \rightarrow \infty$, 
$
{\Delta}_{\mu}(\bs{\Gamma})/{\Delta}_{v^2}(\bs{\Gamma}) \cdot \sqrt{\sum_{i=1}^I R_i^2/(n_i\Gamma_i^\sta)^3} \rightarrow \infty. 
$
\end{condition}
\thmadj

\subsubsection{Formal statements of the results}

The regularity conditions in Theorem \ref{thm:set_vector_gamma} are 
Conditions \ref{cond:clt} and \ref{cond:sen_Gamma_vec_asymp}. 
The following corollary gives a formal statement for the asymptotic $p$-value for testing Fisher's null under the sensitivity model with set-specific bounds on hidden biases. 

\thmadj
\begin{corollary}\label{cor:pval_set_vector_gamma}
If Fisher's null $H_0$ in \eqref{eq:H0} holds, the vector $\bs{\Gamma}^\sta$ of true hidden biases for all sets is bounded by $\bs{\Gamma} \in [1, \infty]^I$, i.e., $\bs{\Gamma}^\sta \vecle \bs{\Gamma}$, 
and  
Conditions \ref{cond:clt} and \ref{cond:sen_Gamma_vec_asymp} 
hold, 
then 
$$
\tilde{p}_{\bs{\Gamma}} \equiv \tilde{G}_{\bs{\Gamma}}(T) \cdot \I\{ T\ge \mu(\bs{\Gamma}) \} + \I\{T < \mu(\bs{\Gamma})\}
= 
1 - \Phi\left( \frac{T - \mu(\bs{\Gamma})}{\sigma(\bs{\Gamma})} \right) \cdot \I\{ T\ge \mu(\bs{\Gamma}) \} 
$$ 
is an asymptotic valid $p$-value, i.e.,  
$\limsup_{I\rightarrow \infty}\pr( \tilde{p}_{\bs{\Gamma}} \le \alpha)\le\alpha$ for $\alpha\in (0,1)$, 
where 
$\tilde{G}_{\bs{\Gamma}}(c) = \pr\{\tilde{T}(\bs{\Gamma}) \ge c\}$ with $\tilde{T}(\bs{\Gamma})$ defined as in Theorem \ref{thm:set_vector_gamma},
and 
$\mu(\bs{\Gamma})$ and $\sigma^2(\bs{\Gamma})$ are defined 
in \eqref{eq:max_mean_var_vector}.  
\end{corollary}
\thmadj

\subsubsection{Additional remarks}\label{sec:rmk_set_gamma_vec}

First, Theorem \ref{thm:set_vector_gamma} generalizes \citet{Rosenbaum2000sep} and imposes weaker regularity conditions. 
In \citet[][Proposition 1]{Rosenbaum2000sep}, the invoked assumptions imply that ${\Delta}_{\mu}(\bs{\Gamma}) \ge \delta$ and ${\Delta}_{v^2}(\bs{\Gamma}) \le \overline{v}^2$ for some finite positive constant $\delta$ and $\overline{v}^2$ that do not change with $I$; in this case, Condition \ref{cond:sen_Gamma_vec_asymp} reduces to 
$\sum_{i=1}^I R_i^2/(n_i\Gamma_i^\sta)^3 \rightarrow \infty$ as $I\rightarrow \infty$, which is almost guaranteed by Condition \ref{cond:clt}.

Second,
recall that in our finite population asymptotics we have a sequence of matched studies with increasing sizes, under which the true hidden bias $\bs{\Gamma}^\sta$ must vary with $I$ as well.  
In both Condition \ref{cond:sen_Gamma_vec_asymp} and Theorem \ref{thm:set_vector_gamma}, we also allow $\bs{\Gamma}$ to vary with $I$ in our asymptotics as $I \rightarrow \infty$. 
For simplicity, 
we make such dependence implicit.

\subsection{Sensitivity analysis for quantiles of hidden biases}

\subsubsection{Additional details for computation}

Below we give a practical way to implement the ranking of matched sets described in Section \ref{sec:sen_quan_set}, which is used to find the set $\mathcal{I}_{\Gamma_0; k}$ and thus solve the maximization in \eqref{eq:mu_sigma_quantile_k_simp}. 
Specifically, we can first re-order all sets such that $v_i^2(\Gamma_0)$ decreases in $i$, 
and then rank them based on the value of $\mu_i(\infty) - \mu_i(\Gamma_0)$ using the ``first'' method; see, e.g., the R function \textsf{rank} for details of the ``first'' method \citep{Rteam}.

\subsubsection{Details for regularity conditions}

We introduce the regularity conditions for ensuring the asymptotic validity of the derived Gaussian upper bound for the tail probability of $T$ in Section \ref{sec:sen_quan_set}. 

We first invoke Condition \ref{cond:clt} for 
the asymptotic Gaussian approximation for the true distribution of the test statistic $T$. 
We then consider the asymptotic upper bound for the tail probability of $T$ under the sensitivity model $\Gamma_{(k)}^\sta \le \Gamma_0$, assuming that $\mathcal{I}_k^\sta$, the set of indices of the $k$ matched sets with the smallest hidden biases, is known. 
Note that $\Gamma_{(k)}^\sta \le \Gamma_0$ is equivalent to that $\bs{\Gamma}^\sta \vecle \Gamma_0 \cdot \bs{1}_{\mathcal{I}_k^\sta} + \infty \cdot \bs{1}_{\mathcal{I}_k^\sta}^\complement$. 
We invoke the following condition, which essentially assumes that Condition \ref{cond:sen_Gamma_vec_asymp} holds for $\bs{\Gamma} = \Gamma_0 \cdot \bs{1}_{\mathcal{I}_k^\sta} + \infty \cdot \bs{1}_{\mathcal{I}_k^\sta}^\complement$. 
\thmadj
\begin{condition}\label{cond:quantile_set_true_Ik}
As $I \rightarrow \infty$, 
$
{\Delta}_{\mu}(\Gamma_0 \cdot \bs{1}_{\mathcal{I}_k^\sta} + \infty \cdot \bs{1}_{\mathcal{I}_k^\sta}^\complement)/{\Delta}_{v^2}(\Gamma_0 \cdot \bs{1}_{\mathcal{I}_k^\sta} + \infty \cdot \bs{1}_{\mathcal{I}_k^\sta}^\complement) \cdot \sqrt{\sum_{i=1}^I R_i^2/(n_i\Gamma_i^\sta)^3}  \rightarrow \infty, 
$
where 
${\Delta}_{\mu}(\cdot)$ and ${\Delta}_{v^2}(\cdot)$ are defined the same as in Appendix \ref{sec:detail_set_specific_bound_regularity}. 
\end{condition}
\thmadj
Finally, we consider the asymptotic validity of our set $\mathcal{I}_{\Gamma_0;k}$ obtained from the optimization in \eqref{eq:mu_sigma_quantile_k_simp}. 
In particular, we want to make sure that the asymptotic upper bound for the tail probability of $T$ pretending that $\mathcal{I}_{\Gamma_0;k}$ contains the matched sets with the $k$ smallest hidden biases is no less than that given the truth $\mathcal{I}_k^\sta$. 
For any $\Gamma_0\in [1, \infty)$, 
by our construction, for any $i \in \mathcal{I}_{\Gamma_0;k} \setminus \mathcal{I}_k^\sta$ and $j \in \mathcal{I}_k^\sta \setminus \mathcal{I}_{\Gamma_0;k}$, we must have 
$\mu_i(\infty) - \mu_i(\Gamma_0) \le \mu_j(\infty) - \mu_j(\Gamma_0)$, and moreover, 
$v_i(\Gamma_0) \ge v_j(\Gamma_0) $ if $\mu_i(\infty) - \mu_i(\Gamma_0) = \mu_j(\infty) - \mu_j(\Gamma_0)$. 
Define 
$
\mathcal{B}_k(\Gamma_0) = \big\{(i,j):  v_i^2(\Gamma_0) < v_j^2(\Gamma_0),  i \in \mathcal{I}_{\Gamma_0;k} \setminus \mathcal{I}_k^\sta, j \in \mathcal{I}_k^\sta \setminus \mathcal{I}_{\Gamma_0;k} \big\}. 
$
Then for any $(i,j) \in \mathcal{B}_k(\Gamma_0)$, 
$\mu_i(\infty) - \mu_i(\Gamma_0)$ must be strictly less than $\mu_j(\infty) - \mu_j(\Gamma_0)$. 
We further introduce the following to denote the average differences over all pairs in $\mathcal{B}_k(\Gamma_0)$: 
\begin{align*}
    {\Delta}_{\mu}(\Gamma_0; k) & = 
    \frac{1}{|\mathcal{B}_k(\Gamma_0)|} \sum_{(i,j)\in \mathcal{B}_k(\Gamma)}
    \left[ \{ \mu_j(\infty) - \mu_j(\Gamma_0) \} -  \{ \mu_i(\infty) - \mu_i(\Gamma_0) \} \right], \\
    {\Delta}_{v^2}(\Gamma_0; k) & = 
    \frac{1}{|\mathcal{B}_k(\Gamma_0)|} \sum_{(i,j)\in \mathcal{B}_k(\Gamma_0)}
    \big\{ v^2_j(\Gamma_0) - v^2_i(\Gamma_0) \big\}. 
\end{align*}
By the discussion before,
${\Delta}_{\mu}(\Gamma_0; k)$ and ${\Delta}_{v^2}(\Gamma_0; k)$ are positive if $\mathcal{B}_k(\Gamma_0)$ is non-empty; 
otherwise we define 
their ratio as $\infty$. 
The 
condition below 
rigorizes the intuition 
that maximizing mean is more important than maximizing variance, similar 
in spirit 
to 
Conditions \ref{cond:sen_Gamma_vec_asymp} and \ref{cond:quantile_set_true_Ik}.%
\thmadj
\begin{condition}\label{cond:sen_Gamma_quantile_asymp}
As $I \rightarrow \infty$, 
$
\Delta_{\mu}(\Gamma_0; k) / \Delta_{v^2}(\Gamma_0; k) \cdot \sqrt{\sum_{i\in \mathcal{I}^\sta_k} R_i^2/(n_i\Gamma_0)^3} \rightarrow \infty. 
$
\end{condition}
\thmadj

\subsubsection{Formal statements of the results}

The regularity conditions in Theorem \ref{thm:set_quantile_gamma} are Conditions 
\ref{cond:clt}, \ref{cond:quantile_set_true_Ik} and \ref{cond:sen_Gamma_quantile_asymp}. 
The following corollary gives a formal statement for the asymptotic $p$-value for testing Fisher's null under the sensitivity model with bounds on quantiles of hidden biases. 

\thmadj
\begin{corollary}\label{cor:sen_pval_set}
If Fisher's null $H_0$ in \eqref{eq:H0} holds, the true hidden bias at rank $k$ is bounded by $\Gamma_0 \in [1, \infty)$, i.e., $\Gamma_{(k)}^\sta \le \Gamma_0$, 
and Conditions  
\ref{cond:clt}, \ref{cond:quantile_set_true_Ik} and \ref{cond:sen_Gamma_quantile_asymp}
hold, then 
\begin{align*}
\tilde{p}_{\Gamma_0; k}
& = 
\tilde{G}_{\Gamma_0; k} (T) \cdot \I\{T \ge \mu(\Gamma_0; k) \} +  \I\{T < \mu(\Gamma_0; k) \}
= 
1 - \Phi\left(
\frac{T - \mu(\Gamma_0; k)}{\sigma(\Gamma_0; k)}
\right) \cdot 
\I\{T \ge \mu(\Gamma_0; k)\}
\end{align*}
is an asymptotic valid $p$-value, i.e., 
$\limsup_{I \rightarrow \infty}\pr(\tilde{p}_{\Gamma_0; k} \le \alpha) \le \alpha$ for $\alpha \in (0,1)$, where $\tilde{G}_{\Gamma_0; k} (c) = \pr\{\tilde{T}(\Gamma_0; k) \ge c\}$ with $\tilde{T}(\Gamma_0; k)$ defined 
as 
in Theorem \ref{thm:set_quantile_gamma}, 
and 
$\mu(\Gamma_0; k)$ and $\sigma^2(\Gamma_0; k)$ are defined 
in \eqref{eq:mu_sigma_quantile_k_simp}. 
\end{corollary}
\thmadj

\subsubsection{Additional remarks}\label{sec:rmk_gamma_quantile}

Similar to the discussion in 
{\lxr Appendix}
\ref{sec:rmk_set_gamma_vec}, 
in Conditions \ref{cond:quantile_set_true_Ik} and \ref{cond:sen_Gamma_quantile_asymp} and Theorem \ref{thm:set_quantile_gamma}, both $\Gamma_0$ and $k$ are allowed to vary with $I$, and the asymptotic conditions or results there concern the limiting behavior when $I$ goes to infinite. 
Below we intuitively argue the rationality of the three regularity conditions in Theorem \ref{thm:set_quantile_gamma} under various specifications for $\Gamma_0$ and $k$. 
Condition \ref{cond:clt} depends only on the true hidden biases and is unaffected by $\Gamma_0$ or $k$, and imposes a mild regularity condition. 
When $\Gamma_{(k)}^\sta \le \Gamma_0$, 
as illustrated before, 
both 
$\Delta_{\mu}(\Gamma_0 \cdot \bs{1}_{\mathcal{I}_k^\sta} + \infty \cdot \bs{1}_{\mathcal{I}_k^\sta}^\complement)/\Delta_{v^2}(\Gamma_0 \cdot \bs{1}_{\mathcal{I}_k^\sta} + \infty \cdot \bs{1}_{\mathcal{I}_k^\sta}^\complement)$ 
and 
$\Delta_{\mu}(\Gamma_0; k) / \Delta_{v^2}(\Gamma_0; k)$ are the ratios between two positive averages (or $0/0$ defined to be infinity). 
Thus, it is likely that they can be bounded below by a certain positive constant. 
Assuming this is true, then Condition \ref{cond:quantile_set_true_Ik} and \ref{cond:sen_Gamma_quantile_asymp} reduces to $\sum_{i\in \mathcal{I}_k^\sta} R_i^2/(n_i\Gamma_0)^3 \rightarrow \infty$, 
which is likely to holds when $k\rightarrow \infty$ as $I \rightarrow \infty$. In practice, we are often interested in $k \approx \beta I$ for some $\beta\in (0,1)$, which corresponds to the $\beta$th quantile of the empirical distribution of the hidden biases across all 
matched 
sets.

\subsection{Simultaneous sensitivity analysis for all quantiles of hidden biases 
}\label{sec:simul_set_app}

Below we give the details for Section \ref{sec:simulta_set}. Specifically, 
assuming Fisher's null $H_0$ holds, 
the $p$-value in Corollary \ref{cor:sen_pval_set} is equivalently an asymptotic valid $p$-value for testing the null hypothesis $\Gamma_{(k)}^\sta \le \Gamma_0 \Leftrightarrow I^\sta(\Gamma_0) \le I-k$
on the true hidden biases, 
and inverting the test for various values of $\Gamma_0$ and $k$ can then provide us confidence sets for $\Gamma_{(k)}^\sta$ and $I^\sta(\Gamma_0)$, respectively. 
Moreover, as demonstrated shortly, these confidence sets are simultaneously valid, without the need of any adjustments due to multiple analyses.

We first invoke the following condition, which essentially assumes that Conditions \ref{cond:clt}, \ref{cond:quantile_set_true_Ik} and \ref{cond:sen_Gamma_quantile_asymp} invoked in Theorem \ref{thm:set_quantile_gamma} and Corollary \ref{cor:sen_pval_set} hold simultaneously for all $k$ and $\Gamma_0$ within certain ranges. 

\begin{condition}\label{cond:set_simultaneous} 
As $I \rightarrow \infty$, 
for the sequences of $\{\underline{k}_I: I=1,2, \ldots\}$ and $\{\overline{\Gamma}_I: I = 1, 2, \ldots \}$,
\begin{enumerate}[label=(\roman*), topsep=1ex,itemsep=-0.3ex,partopsep=1ex,parsep=1ex]
\item $(\max_{1\le i\le I}R_i^2) /\{ \sum_{i=1}^I R_i^2/(n_i\Gamma_i^\sta)^3 \} \rightarrow 0$;
\item
$\inf_{\underline{k}_I \le k\le I, \Gamma_0 \in [\Gamma_{(k)}^\sta, \overline{\Gamma}_I]}
\Delta_{\mu}(\Gamma_0 \cdot \bs{1}_{\mathcal{I}_k^\sta} + \infty \cdot \bs{1}_{\mathcal{I}_k^\sta}^\complement)/\Delta_{v^2}(\Gamma_0 \cdot \bs{1}_{\mathcal{I}_k^\sta} + \infty \cdot \bs{1}_{\mathcal{I}_k^\sta}^\complement) \cdot \sqrt{\sum_{i=1}^I R_i^2/(n_i\Gamma_i^\sta)^3}  \rightarrow \infty; 
$
\item
$
\inf_{\underline{k}_I \le k \le I, \Gamma_0 \in [\Gamma_{(k)}^\sta, \overline{\Gamma}_I]}
\Delta_{\mu}(\Gamma_0; k) / \Delta_{v^2}(\Gamma_0; k) \cdot \sqrt{\sum_{i\in \mathcal{I}_k^\sta} R_i^2/(n_i\Gamma_0)^3} \rightarrow \infty. 
$
\end{enumerate}
\end{condition}

\begin{rmk}\label{rmk:reg_cond_quantile}
Below we give some intuition for Condition \ref{cond:set_simultaneous}. 
As discussed in 
{\lxr Appendix}
\ref{sec:rmk_gamma_quantile}, 
Condition \ref{cond:set_simultaneous} (i) does not depend on the value of $\Gamma_0$ or $k$, and generally holds under mild conditions on the true hidden biases and the $q_{ij}$'s. Condition \ref{cond:set_simultaneous}(ii) and (iii) depend on $\Gamma_0$ and $k$. 
Moreover, we expect both ratios $\Delta_{\mu}(\Gamma_0 \cdot \bs{1}_{\mathcal{I}_k^\sta} + \infty \cdot \bs{1}_{\mathcal{I}_k^\sta}^\complement)/\Delta_{v^2}(\Gamma_0 \cdot \bs{1}_{\mathcal{I}_k^\sta} + \infty \cdot \bs{1}_{\mathcal{I}_k^\sta}^\complement)$ and $\Delta_{\mu}(\Gamma_0; k) / \Delta_{v^2}(\Gamma_0; k)$ to be of a positive constant order and relatively robust to a wide range of $\Gamma_0$ and $k$. 
Then, intuitively, Condition \ref{cond:set_simultaneous}(ii) essentially requires that $\sum_{i=1}^I R_i^2/(n_i\Gamma_i^\sta)^3 \rightarrow \infty$ as $I\rightarrow \infty$, which is almost implied by Condition \ref{cond:set_simultaneous}(i). 
Moreover, Condition \ref{cond:set_simultaneous}(iii) essentially requires that $\sum_{i\in \mathcal{I}_{\underline{k}_I}^\sta} R_i^2/(n_i\overline{\Gamma}_I)^3 \rightarrow \infty$ as $I\rightarrow \infty$, which is likely to hold when $\lim_{I\rightarrow \infty} \underline{k}_I = \infty$ and $\overline{\Gamma}_I$ is bounded. 
For example, we can choose $\underline{k}_I \approx \beta I$ and $\overline{\Gamma}_I$ to be a certain positive constant. 
In practice, the choice of $\underline{k}_I$ and $\overline{\Gamma}_I$ generally does not affect our inferential results. 
This is because, when $k$ is small or $\Gamma_0$ is large, we allow a large amount of hidden biases in the treatment assignment, under which we 
do not expect the $p$-values to be significant. 
This issue 
will be further discussed throughout the remaining of this subsection. 
\end{rmk}

To ensure the validity of the $p$-value for all values of $(\Gamma_0, k) \in [1, \infty] \times \{1,2,\ldots, I\}$, 
Condition \ref{cond:set_simultaneous} motivates us to consider the following modified $p$-value: 
\begin{align}\label{eq:modify_pval}
\vardbtilde{p}_{\Gamma_0, k} = 
\tilde{p}_{\Gamma_0;k} \cdot \I(\Gamma_0 \le \overline{\Gamma}_I \text{ and } k \ge \underline{k}_I) 
+ 
1 \cdot \I(\Gamma_0 > \overline{\Gamma}_I \text{ or } k < \underline{k}_I), 
\quad 
(\forall \Gamma_0, k)
\end{align}
with $\underline{k}_I$ and $\overline{\Gamma}_I$ being the sequences in Condition \ref{cond:set_simultaneous}. Specifically, we set the $p$-value to be the conservative 1 
whenever $\Gamma_0$ is greater than 
$\overline{\Gamma}_I$ or $k$ is less than $\overline{k}_I$. 
Similar to the discussion at the end of Remark \ref{rmk:reg_cond_quantile}, 
the set of possible values of $\bs{\Gamma}^\sta$ under 
the null hypothesis $\Gamma_{(k)}^\sta \le \Gamma_0$ increases as $\Gamma_0$ increases and $k$ decreases. Thus, the larger the $\Gamma_0$ and the smaller the $k$, the more hidden biases are allowed in the matched study, under which we expect the $p$-value $\tilde{p}_{\Gamma_0;k}$ to be insignificant (i.e., greater than the significant level of interest). Thus, setting the $p$-value to be 1 for small $k$ and large $\Gamma_0$ as in 
\eqref{eq:modify_pval} generally does not affect our inference results. 

The following corollary then constructs  
asymptotic valid 
confidence sets for quantiles of hidden biases 
as well as 
numbers of sets with hidden biases 
greater than 
any threshold. 

\begin{corollary}\label{cor:set_ci_gamma}
Assume that Fisher's null $H_0$ and Condition \ref{cond:set_simultaneous} hold. 
(i)
For any $1 \le k \le I$ and $\alpha\in (0,1)$, 
an asymptotic $1-\alpha$ confidence set for $\Gamma_{(k)}^\sta$ is $\{\Gamma_0\ge 1: \vardbtilde{p}_{\Gamma_0;k} > \alpha\} \subset [\tilde{\Gamma}_{(k)}, \infty]$, where $\tilde{\Gamma}_{(k)} \equiv  \inf \{\Gamma_0\ge 1: \vardbtilde{p}_{\Gamma_0;k} > \alpha\}$.
(ii)
For any $\Gamma_0\in \mathbb{R}$ and $\alpha\in (0,1)$, 
an asymptotic $1-\alpha$ confidence set for $I^\sta(\Gamma_0)$ is 
$\{I-k: \vardbtilde{p}_{\Gamma_0;k} > \alpha, 0 \le k \le I\} \subset \{
\tilde{I}(\Gamma_0), 
\ldots, I
\}$, 
where 
$\tilde{I}(\Gamma_0) = I - \sup \{k: \vardbtilde{p}_{\Gamma_0;k} > \alpha, 0\le k \le I\}$. 
\end{corollary}

In Corollary \ref{cor:set_ci_gamma}, for any $1 \le k \le I$, 
to construct confidence set for $\Gamma_{(k)}^\sta$, 
we usually report the lower confidence limit $\tilde{\Gamma}_{(k)}$, the infimum of $\Gamma_0$ such that $\vardbtilde{p}_{\Gamma_0;k} > \alpha$, instead of searching over all possible $\Gamma_0$. 
Analogously, for any $\Gamma_0 \in [1, \infty]$, we usually report the lower confidence limit for $I^\sta(\Gamma_0)$. Moreover, the reported lower confidence limits using either the modified $p$-value or the original $p$-value are very similar, since these detected thresholds $\tilde{\Gamma}_{(k)}$ and $\tilde{I}(\Gamma_0)$ are likely to fall in $[1, \overline{\Gamma}_I]$ and $[0, I-\underline{k}_I]$. 

The asymptotic confidence sets in Corollary \ref{cor:set_ci_gamma} are indeed simultaneously valid. 
Similar to 
Section \ref{sec:simultaneous_pair}, the confidence sets for $\Gamma_{(k)}$'s and $I^\sta(\Gamma_0)$'s 
can be equivalently viewed as confidence sets for the true hidden bias $\bs{\Gamma}^\sta$: 
\begin{align*}
\Gamma_{(k)}^\sta \in 
\{\Gamma_0\ge 1: \vardbtilde{p}_{\Gamma_0;k} > \alpha\}
& \Longleftrightarrow 
\bs{\Gamma}^\sta \in 
\bigcap_{\Gamma_0: \vardbtilde{p}_{\Gamma_0;k} \le \alpha} \mathcal{S}_{\Gamma_0;k}^\complement, & \qquad (1\le k \le I)
\\
I^\sta(\Gamma_0) \in \{I-k: \vardbtilde{p}_{\Gamma_0;k} > \alpha, 0 \le k \le I\}
& \Longleftrightarrow 
\bs{\Gamma}^\sta \in 
\bigcap_{k \ge \underline{k}_I: \vardbtilde{p}_{\Gamma_0;k} \le \alpha} \mathcal{S}_{\Gamma_0;k}^\complement, & \qquad \Gamma_0 \in [1, \infty],
\end{align*}
where 
$
\mathcal{S}_{\Gamma_0;k}$
is the set of all possible $\bs{\Gamma}^\sta$ under the constraint $\Gamma^\sta_{(k)} \le \Gamma_0 \Leftrightarrow I^\sta(\Gamma_0) \le I - k$, 
defined
as in 
Appendix \ref{sec:add_sen_pair}. 
We can verify that the intersection of all confidence sets for $\Gamma_{(k)}^\sta$'s over all $k$ is the same as that for  $I^\sta(\Gamma_0)$'s over all $\Gamma_0$, and it has the following equivalent forms:
\begin{align}\label{eq:set_joint_set}
\bigcap_{k=1}^I 
\bigcap_{\Gamma_0: \vardbtilde{p}_{\Gamma_0;k} \le \alpha} \mathcal{S}_{\Gamma_0;k}^\complement 
= 
\bigcap_{\Gamma_0 \ge 1}
\bigcap_{k: \vardbtilde{p}_{\Gamma_0;k} \le \alpha} \mathcal{S}_{\Gamma_0;k}^\complement
= 
\bigcap_{\Gamma_0, k: \vardbtilde{p}_{\Gamma_0;k} \le \alpha} \mathcal{S}_{\Gamma_0;k}^\complement. 
\end{align}

\begin{theorem}\label{thm:set_joint_ci}
Assume that Fisher's null $H_0$ and Condition \ref{cond:set_simultaneous} hold. 
For any $\alpha \in (0,1)$, the intersection of all $1-\alpha$ confidence sets for the quantiles of hidden biases $\Gamma_{(k)}^\sta$'s, viewed as a confidence set for $\bs{\Gamma}^\sta$, is the same as that of all confidence sets for the numbers $I^\sta(\Gamma_0)$'s of matched sets with hidden biases greater than any threshold, 
and it has equivalent forms as in \eqref{eq:set_joint_set}. 
Moreover, it is indeed an asymptotic $1-\alpha$ confidence set with probability at least $1-\alpha$ to cover the true hidden bias vector $\bs{\Gamma}^\sta$, i.e., 
$
\liminf_{I \rightarrow \infty} \pr( \bs{\Gamma}^\sta \in \bigcap_{\Gamma_0, k: \vardbtilde{p}_{\Gamma_0;k} \le \alpha} \mathcal{S}_{\Gamma_0;k}^\complement ) \ge 1-\alpha. 
$
\end{theorem}

Because the confidence set for $\Gamma_{(k)}^\sta$ must be a subset of $[\tilde{\Gamma}_{(k)}, \infty]$ as shown in Corollary \ref{cor:set_ci_gamma}, 
the intersection of all confidence intervals $[\tilde{\Gamma}_{(k)}, \infty]$'s must cover the set in \eqref{eq:set_joint_set}, and thus, from Theorem \ref{thm:set_joint_ci}, it must also be an asymptotic valid confidence set for $\bs{\Gamma}^\sta$. 
In practice, these confidence intervals $[\tilde{\Gamma}_{(k)}, \infty]$'s can be easily visualized.  
Moreover, 
if we are interested in the number of sets $I^\sta(\Gamma_0)$ with biases greater than $\Gamma_0$, 
we can simply count the number of the intervals $[\tilde{\Gamma}_{(k)}, \infty]$'s that does not contain $\Gamma_0$.

{\add 
 \section{Sensitivity Analysis for Bounded Null Hypotheses}\label{sec:bound_null}

The discussion 
in Sections \ref{sec:pair} and \ref{sec:set} in the main paper 
considers mainly Fisher's sharp null hypothesis of no treatment effect for every unit; 
the methods extend straightforwardly to arbitrary sharp null hypothesis that specifies all individual treatment effects, as discussed at the end of this section. 
Below we will extend the methods to allow non-sharp or equivalently composite null hypotheses for treatment effects. 
In particular, we consider the bounded null hypotheses, which, instead of specifying all individual treatment effects as sharp null hypotheses, specify the lower (or upper) bounds of all individual treatment effects. 
Compared to the sharp nulls, 
the bounded nulls are much more likely to hold in practice. 
Without loss of generality, we consider the following  bounded null hypothesis assuming that the treatment has a non-positive effects for all individuals:
\begin{align}\label{eq:bound_null}
	\overline{H}_{0}: 
	Y_{ij}(1)\le Y_{ij}(0) \text{ for }  1\le j\le n_i \text{ and } 1\le i \le I 
	\text{ or equivalently } \bs{Y}(1) \vecle \bs{Y}(0);
\end{align}
As discussed at the end of this section, 
we can straightforwardly extend it to general bounded null hypotheses. 
Below we will demonstrate that the proposed approaches in Sections \ref{sec:pair} and \ref{sec:set} for Fisher's sharp null can also work for the bounded null in \eqref{eq:bound_null}, 
as long as the test statistic $t(\cdot, \cdot)$ in \eqref{eq:test_stat} enjoys certain properties.

Following \citet{li2020quantile}, 
we introduce the following two properties for test statistics, which will be used throughout this section. 
Let $\circ$ denote the element-wise multiplication. 

\begin{definition}\label{def:effect_increase}
	A statistic $t(\bs{z}, \bs{y})$ is said to be effect increasing, if $t(\bs z,\bs y+\bs z\circ\bs\eta+(\bs 1-\bs z)\circ\bs\xi)\ge t(\bs z,\bs y)$ for any $\bs z\in\mathcal Z$ and $\bs y,\bs\eta,\bs\xi\in\mathbb R^N$ with $\bs\eta\vecge\bs 0\vecge\bs\xi$.
\end{definition}

\begin{definition}\label{def:dif_increase}
	A statistic $t(\bs{z}, \bs{y})$ is said to be differential increasing, if $t(\bs z,\bs y+\bs a\circ\bs\eta)-t(\bs z,\bs y)\le t(\bs a,\bs y+\bs a\circ\bs\eta)-t(\bs a,\bs y)$ for any $\bs z,\bs a\in\mathcal Z$ and $\bs y,\bs\eta\in\mathbb R^N$ with $\bs\eta\vecge\bs 0$. 
\end{definition}

Intuitively, for an effect increasing statistic, 
given any treatment assignment, 
the value of the statistic is 
{\lxr increasing}
in treated outcomes and is 
{\lxr decreasing}
in control outcomes. 
For a differential increasing statistic,  
if we add non-negative values to the outcomes of a subset of units, then the change of the statistic is maximized when this subset of units receive treatment. 
Both properties can hold for many commonly used test statistics \citep{li2020quantile}. 
For example, 
in the context of the matched studies, 
{\lxr when each matched set $i$ has $1$ treated unit and $n_i-1$ control units,} 
the usual difference-in-means statistic 
\begin{align}\label{eq:diff_in_means}
	t(\bs{z}, \bs{y}) = \sum_{i=1}^I  w_i \Big\{ \sum_{j=1}^{n_i}z_{ij}y_{ij} - 
	\frac{1}{n_i-1} \sum_{j=1}^{n_i}(1-z_{ij})y_{ij}
	\Big\}
\end{align}
is both effect increasing and differential increasing, 
where $w_i$'s can be any constant weights that can depend on, say, the set sizes $n_i$'s. 

\subsection{Sensitivity analysis for matched pairs}\label{sec:bound_pair}

We first consider the matched pair studies, and are interested in testing the bounded null hypothesis $\overline{H}_0$ in \eqref{eq:bound_null} under sensitivity models with constraints on quantiles of hidden biases.  
As demonstrated in the theorem below, 
the test proposed in Corollary \ref{cor:pval_pair_quantile} for Fisher's null $H_0$ of no treatment effect can also be valid for the bounded null hypothesis $\overline{H}_0$, 
as long as the test statistic satisfies certain properties. 
In other words, we provide a broader justification for the sensitivity analysis based on quantiles of hidden biases discussed in Section \ref{sec:pair}.  

\begin{theorem}\label{thm:bounded_null_pair}
Consider a matched pair study, i.e., $n_i=2$ for all $i$. 
If the bounded null hypothesis $\overline{H}_0$ in \eqref{eq:bound_null} holds,  
the bias at rank $k$ is bounded by $\Gamma_0$, i.e., $\Gamma^\sta_{(k)}\le\Gamma_0$, and the test statistic $t(\cdot, \cdot)$ in \eqref{eq:test_stat} is effect increasing or differential increasing, then $\bar{p}_{\Gamma_0;k}$ in Corollary \ref{cor:pval_pair_quantile} (computed pretending that Fisher's null $H_0$ in \eqref{eq:H0} holds) is still a valid $p$-value, i.e.,
$\pr( \overline{p}_{\Gamma_0; k} \le \alpha ) \le \alpha$ for any $\alpha\in (0,1)$.
\end{theorem}

By the equivalent understanding of sensitivity analysis discussed in Section \ref{sec:equiv_sen} and similar to the discussion in Section \ref{sec:simultaneous_pair}, 
utilizing Theorem \ref{thm:bounded_null_pair} and by test inversion, 
we can construct confidence sets for the quantiles of hidden biases $\Gamma^\sta_{(k)}$'s, as well as the number of matched pairs with hidden biases exceeding any threshold $I^\sta(\Gamma_0)$'s, assuming that the bounded null hypothesis holds. 
These confidence sets have the same form as discussed in Corollary \ref{cor:pair_ci_gamma3}, and moreover, they will also be simultaneously valid without the need of any adjustment due to multiple analyses as in Theorem \ref{thm:pair_joint_ci}. 
This indicates that the proposed sensitivity analysis for the bounded null hypothesis is also simultaneously valid over all quantiles of hidden biases. 
We summarize the results in the following theorem. 

\begin{theorem}\label{thm:bound_joint_ci}
If the test statistic in \eqref{eq:test_stat} is effect increasing or differential increasing, then 
Corollary \ref{cor:pair_ci_gamma3} and Theorem \ref{thm:pair_joint_ci} still hold when we assume that the bounded null hypothesis $\overline{H}_0$ in \eqref{eq:bound_null} (beyond Fisher's sharp null $H_0$ in \eqref{eq:H0}) holds. 
\end{theorem}

\subsection{Sensitivity analysis for matched sets}\label{sec:bound_set}

We then consider the general matched studies with $n_i \ge 2$ for all $i$, 
and again are interested in testing the bounded null hypothesis $\overline{H}_0$ in \eqref{eq:bound_null} under sensitivity models with constraints on quantiles of hidden biases. 
Similarly, we can demonstrate that the test proposed in Corollary \ref{cor:sen_pval_set} for Fisher's sharp null $H_0$ is also asymptotically valid for the bounded null $\overline{H}_0$. 
We introduce the following conditions analogous to those in Section \ref{sec:set}, but with the true potential outcomes replaced by the imputed one based on Fisher's sharp null of no effect. 

\begin{condition}\label{134in_prob}
With the true control potential outcome vector $\bs{Y}(0)$ replaced by the imputed control potential outcome vector $\bs{Y}$ based on the sharp null ${\lxr H_{0}}$ (or equivalently pretending $\bs{Y}$ as the true control potential outcome vector), 
Conditions \ref{cond:clt}, \ref{cond:quantile_set_true_Ik} and \ref{cond:sen_Gamma_quantile_asymp} hold in probability, i.e., the {\lxr convergences} in these three conditions hold in probability as $I\rightarrow \infty$. 
\end{condition}

The following theorem establishes the broader justification for the sensitivity analysis of general matched sets discussed in Section \ref{sec:sen_quan_set}. 

\begin{theorem}\label{thm:bounded_null_sets_effect_diff_increasing}
Consider a general matched study with set sizes $n_i \ge 2$ for all $i$. 
If the bounded null hypothesis $\overline{H}_0$ in \eqref{eq:bound_null} holds, 
the bias at rank $k$ is bounded by $\Gamma_0$, i.e., $\Gamma^\sta_{(k)}\le \Gamma_0$, 
Condition \ref{134in_prob} holds, 
and the test statistic $t(\cdot, \cdot)$ in \eqref{eq:test_stat} is effect increasing or differential increasing, 
then $\tilde p_{\Gamma_0;k}$ in Corollary \ref{cor:sen_pval_set} (computed pretending that Fisher's null $H_0$ in \eqref{eq:H0} holds) is still an asymptotic valid $p$-value, i.e., $\limsup_{I\to\infty}\pr( \tilde{p}_{\Gamma_0;k} \le \alpha ) \le \alpha$ for $\alpha\in (0,1)$.
\end{theorem}

Analogous to Section \ref{sec:simulta_set}, 
assuming that the bounded null hypothesis $\overline{H}_0$ holds, 
we can use Theorem \ref{thm:bounded_null_sets_effect_diff_increasing} to construct confidence sets for quantiles of hidden biases $\Gamma_{(k)}^\sta$'s, as well as the numbers of matched sets with hidden biases greater than any threshold $I^\sta(\Gamma_0)$'s. 
Moreover, 
under certain regularity conditions, 
these confidence sets will be asymptotically simultaneously valid, indicating that the sensitivity analysis in Theorem \ref{thm:bounded_null_sets_effect_diff_increasing} is asymptotically simultaneously valid over all quantiles of hidden biases. 
We summarize the results in the theorem below. 
Before that, we introduce the following condition, and define the modified $p$-value $\dbtilde{p}_{\Gamma_0, k}$ in the same way as \eqref{eq:modify_pval}. 

\begin{condition}\label{cond:bound_joint_ci_prob}
With the true control potential outcome vector $\bs{Y}(0)$ replaced by the imputed control potential outcome vector $\bs{Y}$ based on the sharp null ${\lxr H_{0}}$ (or equivalently pretending $\bs{Y}$ as the true control potential outcome vector), 
Condition \ref{cond:set_simultaneous} holds in probability, i.e., the {\lxr convergences} in {\lxr (i)--(iii)} hold in probability as $I \rightarrow \infty$. 
\end{condition}

\begin{theorem}\label{thm:bound_joint_ci_set}
If the test statistic in \eqref{eq:test_stat} is effect increasing or differential increasing, then 
Corollary \ref{cor:set_ci_gamma} and Theorem \ref{thm:set_joint_ci} still hold when we assume that the bounded null hypothesis $\overline{H}_0$ in \eqref{eq:bound_null} (beyond Fisher's sharp null $H_0$ in \eqref{eq:H0}) and Condition \ref{cond:bound_joint_ci_prob} hold.  
\end{theorem}
}

\begin{rmk}
Throughout the paper and the supplementary material, 
we focus mainly on Fisher's sharp null of no effect and the bounded null of non-positive effects.  
The proposed approach can be straightforwardly used for arbitrary sharp and bounded null hypotheses. 
Specifically, with $\bs{\delta}$ being the hypothesize effects or (upper or lower) bounds of the effects, 
we can consider transformed potential outcomes $(\tilde{\bs{Y}}(1), \tilde{\bs{Y}}(0)) = (\bs{Y}(1) - \bs{\delta}, \bs{Y}(0))$ or $(\tilde{\bs{Y}}(1), \tilde{\bs{Y}}(0)) = (-\bs{Y}(1) + \bs{\delta}, -\bs{Y}(0))$ with observed outcomes {\lxr $\tilde{\bs{Y}} = \bs{Y}- \bs{Z}\circ \bs{\delta}$ or $\tilde{\bs{Y}} = -\bs{Y} + \bs{Z}\circ \bs{\delta}$}, respectively. 
Importantly, for the transformed potential outcomes, either Fisher's sharp null of no effects or the bounded null of non-positive effects will hold. 
\end{rmk}

\section{Computational complexity}\label{sec:complexity}

We now study the computational complexity of our approaches 
in Sections \ref{sec:pair} and \ref{sec:set} and Appendices 
{\lxr \ref{sec:add_sen_pair}--\ref{sec:bound_null}.}
Note that the procedure of sensitivity analysis for bounded null hypotheses is the same as that for sharp null hypotheses, except that it requires the test statistic to satisfy certain properties. 
Therefore, in the following, we focus only on sensitivity analysis for sharp null hypotheses, {\lxr particularly Fisher’s null of no effect, without loss of generality.}

\subsection{Sensitivity analysis for matched pairs}\label{sec:alg_pair}

We consider first the finite-sample exact sensitivity analysis for matched pair studies. 
Algorithm \ref{alg:pair} shows the pseudocode for implementing the approach in Section \ref{sec:pair}. 
Moreover, we also list the time complexity for each step in Algorithm \ref{alg:pair}. 
Below we give some brief explanation. 
\begin{enumerate}[label=(\roman*), topsep=1ex,itemsep=-0.3ex,partopsep=1ex,parsep=1ex]
\item 
First, we consider the calculation of $q_{ij}$'s in \eqref{eq:test_stat}. 
\citet{Rosenbaum2007m} suggested 
{\lxr using}
m-statistic with 
$q_{ij} = n_i^{-1} \sum_{l=1}^{n_i} \psi((Y_{ij}(0)-Y_{il}(0))/s)$, 
where $\psi(y) = \text{sign}(y) \cdot \kappa \cdot \min[ 1, \max\{0, 
(|y|-\iota)/(\kappa - \iota) 
\} ]$ 
is an odd function with truncation at $\kappa$ and inner trimming at $\iota$, 
$s$ is 
the median 
of all the absolute pairwise differences $|Y_{ij}(0) - Y_{ik}(0)|$'s for all $i$ and $j \ne k$. 
In the case of matched pairs,  
the complexity for calculating all the within-pair differences of the control potential outcomes is $O(I)$. 
The complexity for finding the median $s$, or any other quantile, of the $I$ absolute pairwise differences is $O(I)$; see, e.g, \citet{selection}. 
Consequently, the total complexity for calculating all the $q_{ij}$'s is $O(I)$. 
It is not difficult to see that the complexity for the difference-in-means statistic {\lxr in \eqref{eq:diff_in_means}} is also $O(I)$. 

\item Second, we consider the calculation of $\mathcal{I}_k$ in Theorem \ref{thm:pair_quantile_gamma}. 
By definition, it suffices to find the $k$ matched pairs with the 
{\lxr smallest}
values of the $|q_{i1}-q_{i2}|$'s. 
This can be done once we sort all the $I$ matched pairs based on their $|q_{i1}-q_{i2}|$'s, whose complexity is $O(I \log I)$. 

\item Third, the complexity for calculating the realized value of the test statistic is $O(I)$. 

\item Fourth, we consider the calculation of the matrix $\bs{B}_{\simu}$ consisting of the Monte Carlo samples. 
For matched pair $i\in \mathcal{I}_k$, we need to  sample from a distribution $M$ times. 
{\lxr For matched pair $i\notin \mathcal{I}_k$}, we just sample from the point mass at $\max\{q_{i1}, q_{i2}\}$. The complexity of this step is $O(IM)$. 

\item Fifth, the complexity for calculating the column sums of $\bs{B}_{\simu}$ is $O(IM)$. This then provides $M$ Monte Carlo samples from the distribution whose tail probability is the upper bound of $\pr(T\ge t)$ under our sensitivity analysis. 

\item Finally, we compare the realized value of the test statistic to the $M$ Monte Carlo samples to get the $p$-value, which has complexity $O(M)$. 
\end{enumerate}
\begin{algorithm}[ht]
    \caption{Calculation of $p$-values for testing the sharp null of no causal effect in matched pair studies,  under the sensitivity analysis with constraint $\Gamma_{(k)}\le \Gamma_0$ for some given $\Gamma_0\ge 1$ and $1\le k\le I$ 
    }\label{alg:pair}
    \textbf{Input}\\
    \hspace*{\algorithmicindent} $y_{ij}$ for $1\le i\le I, j=1,2$: observed outcomes \\
    \hspace*{\algorithmicindent} $z_{ij}$ for $1\le i\le I, j=1,2$: observed treatment assignment \\
    \hspace*{\algorithmicindent} $(\Gamma_0,k)$: sensitivity analysis constraint of $\Gamma^\star_{(k)}\le\Gamma_0$\\
    \hspace*{\algorithmicindent} $M$: number of Monte Carlo samples for approximating $\overline{G}_{\Gamma_0; k}(\cdot)$ in Corollary \ref{cor:pval_pair_quantile}
    \\
    \textbf{Output} \\ 
    \hspace*{\algorithmicindent} $\bar p_{\Gamma_0;k}$ 
    \\
    \textbf{Procedure} 
    \begin{algorithmic}
        \STATE $q_{ij}$ for $1\le i\le I, j=1,2 \gets$ choice of the test statistic $T$ in \eqref{eq:test_stat}   
        \hfill [{\it complexity: often $O(I)$}]
        \STATE $\mathcal I_k \gets$ indices of pairs  with the {\lxr $k$ smallest $|q_{i1}-q_{i2}|$'s} 
        \hfill
        [{\it complexity: $O(I\log I)$}]
        \STATE $T_{\obs} \gets \sum_{i=1}^I\sum_{j=1}^2z_{ij}q_{ij}$ 
        \hfill [{\it complexity: $O(I)$}]
        \STATE $\bs {B}_{\simu}\gets$ an $I\times M$ matrix
        whose $i$th row contains simulated statistics for the $i$th matched pair, with detailed calculation described below: \hfill [{\it complexity: $O(MI)$}]
        \\
        \FOR{$1\le i\le I$}
            \IF{$i$ in $\mathcal I_k$}
            \STATE the $i$th row of $\bs {B}_{\simu}\gets$ 
            draw from a distribution with probability $1/(1+\Gamma_0)$ at $\min\{q_{i1},q_{i2}\}$ and probability $\Gamma_0/(1+\Gamma_0)$ at $\max\{q_{i1},q_{i2}\}$, independently $M$ times
            \ELSIF{$i$ not in $\mathcal I_k$}
            \STATE the $i$th row of $\bs {B}_{\simu}\gets$ a vector of length $M$ with all elements being $\max\{ q_{i1},q_{i2} \}$
            \ENDIF
        \ENDFOR
        \STATE $\bs {T}_{\simu}\gets$ the vector consisting of column sums of $\bs {B}_{\simu}$
        \hfill 
        [{\it complexity: $O(MI)$}]
        \STATE $\bar p_{\Gamma_0;k} \gets$ the proportion of elements of $\bs{T}_{\simu}$ that are no less than $T_{\obs}$
        \hfill 
        [{\it complexity: $O(M)$}]
        \STATE 
        [{\bf Total complexity: $O(I\log I + I M)$}]
    \end{algorithmic}
\end{algorithm}

From the above, the total complexity for calculating the $p$-value $\bar p_{\Gamma_0;k}$ at given $k$ and $\Gamma_0$ is $O(I\log I + I M)$. 

We then consider the complexity for calculating the lower confidence limit $\hat{\Gamma}_{(k)}$ for $\Gamma^\star_{(k)}$ at a given $k$. 
Note that the $p$-value $\bar p_{\Gamma_0;k}$ is monotone in $\Gamma_0$, and we are using binary search to find $\hat{\Gamma}_{(k)}$, which is intuitively the 
{\lxr smallest}
$\Gamma_0$ such that $\bar p_{\Gamma_0;k}$ is not significant. 
Suppose that we have an upper bound $\bar\Gamma_k$ for $\hat{\Gamma}_{(k)}$, i.e., the $p$-value $\bar p_{\bar\Gamma_k;k}$ is greater than the significance level, then the number of $p$-values we need to compute in order to get $\bar p_{\Gamma_0;k}$ 
{\lxr with precision up to $\epsilon$}
is of order $O( \log (\bar\Gamma_k/\epsilon) )$. 
Consequently, the total complexity for getting $\hat{\Gamma}_{(k)}$ is of the order 
$O(I\log I+IM \log (\bar\Gamma_k/\epsilon) )$, noticing that steps (i)--(iii) discussed before for computing $p$-values need only to be done once. 

Finally, we consider the complexity for calculating the lower confidence limits $\hat{\Gamma}_{(k)}$'s for all the $\Gamma^\star_{(k)}$'s. 
Note that steps (i) and 
{\lxr (iii)}
in the $p$-value calculation need only to be done once, and that step (ii) can be done once we sort all the matched pairs based on their $|q_{i1}-q_{i2}|$'s, which has complexity $O(I\log I)$. 
Moreover, suppose that $\bar\Gamma_I$ denotes an upper bound of $\hat{\Gamma}_{(I)}$. 
By the monotonicity property of the $p$-value, 
$\bar p_{\bar\Gamma_I;k}\ge \bar p_{\bar\Gamma_I;I}$ is greater than the significance level for all $k$.  
Thus, $\bar\Gamma_I$ is also an upper bound of all the $\hat{\Gamma}_{(k)}$'s.  
From the above, the total complexity for calculating the $\hat{\Gamma}_{(k)}$'s  for all $1\le k\le I$ is then
$O(I\log I+ I^2 M \log (\bar\Gamma_I/\epsilon) ) = O(I^2 M \log (\bar\Gamma_I/\epsilon))$. 
Despite the theoretical analysis here assumes a common upper bound $\bar\Gamma_I$ for all the $\hat{\Gamma}_{(k)}$s, in practice, 
when searching for $\hat{\Gamma}_{(k)}$, we only need to search on the interval $[1, \hat{\Gamma}_{(k+1)}]$, for any $1\le k \le I-1$. 
That is, we can find $\hat{\Gamma}_{(I)}$ first, and then sequentially obtain $\hat{\Gamma}_{(I-1)}$, $\hat{\Gamma}_{(I-2)}$, $\ldots$, and $\hat{\Gamma}_{(1)}$, with decreasing interval ranges for the binary search.

\subsection{Sensitivity analysis for matched sets}

\begin{algorithm}[htbp]
    \caption{
    Calculation of $p$-values for testing the sharp null of no causal effect in matched set studies,  under the sensitivity analysis with constraint $\Gamma_{(k)}\le \Gamma_0$ for some given $\Gamma_0\ge 1$ and $1\le k\le I$ 
    }\label{alg:set}
    \hspace*{\algorithmicindent} \textbf{Input}\\
    \hspace*{\algorithmicindent} $y_{ij}$ for $1\le i\le I, 1\le j \le n_i$: observed outcomes \\
    \hspace*{\algorithmicindent} $z_{ij}$ for $1\le i\le I, 1\le j\le n_i$: observed treatment assignment \\
    \hspace*{\algorithmicindent} $(\Gamma_0,k)$: sensitivity analysis constraint of $\Gamma^\star_{(k)}\le\Gamma_0$\\
    \hspace*{\algorithmicindent} \textbf{Output} \\ 
    \hspace*{\algorithmicindent} $\tilde p_{\Gamma_0;k}$\\
    \textbf{Procedure} 
    \begin{algorithmic}
        \STATE $q_{ij}$ for $1\le i\le I, 1\le j \le n_i \gets$ choice of test statistic $T$ 
        \hfill [{\it complexity: often $O(\sum_{i=1}^In_i^2)$}]
        \\
        \STATE $T_{\obs} \gets \sum_{i=1}^I\sum_{j=1}^{n_i}z_{ij}q_{ij}$ 
        \hfill [{\it complexity: $O(N)$}]\\ 
        \FOR{$1\le i\le I$}
        \STATE 
        sort $q_{ij}$'s within set $i$: $q_{i(1)}
        \le q_{i(2)}
        \le\cdots\le q_{i(n_i)}$ 
        \hfill
        [{\it complexity: $O(n_i\log(n_i))$}]\\
        \STATE 
        $Q_1 = q_{i(1)}$, 
        and 
        $Q_j = Q_{j-1}+q_{i(j)}$ for $1\le j \le n_i \gets$ cumulative summations of the $q_{i(j)}$s\\ 
        \hfill
        [{\it complexity: $O(n_i)$}]\\
        \STATE 
        $S_1 = q_{i(1)}^2$, 
        and 
        $S_j = S_{j-1}+q_{i(j)}^2$ for $1\le j \le n_i \gets$ cumulative summations of the $q_{i(j)}^2$s\\ 
        \hfill
        [{\it complexity: $O(n_i)$}]\\
        \ENDFOR{
            \hfill 
            [{\it complexity for the for loop: $O(\sum_{i=1}^I n_i\log(n_i))$}]}
        \FOR{$1\le i\le I$}
            \FOR{$1\le a\le n_i-1$}
            \STATE $\mu_{ia}(\Gamma_0)\gets
    \frac{Q_a + \Gamma_0 (Q_{n_i} - Q_a)}{a + \Gamma_0 (n_i - a)}$
            \STATE $v_{ia}^2(\Gamma_0) \gets 
    \frac{S_a + \Gamma_0 (S_{n_i}-S_a)}{a + \Gamma_0(n_i - a)} - \{\mu_{ia}(\Gamma_0)\}^2$
            \ENDFOR{
            \hfill 
            [{\it complexity for the inner for loop: $O(n_i)$}]}
            \STATE $\mu_{i}(\Gamma_0) \gets 
    \max_{a\in \{1, 2, \ldots, n_i-1\}} \mu_{ia}(\Gamma_0)$
            \STATE $v_{i}^2(\Gamma_0) \gets 
    \max_{a: u_{ia}(\Gamma_0) = u_{i}(\Gamma_0)}v_{ia}^2(\Gamma_0)$
            \STATE $\mu_{i}(\infty) \gets q_{i(n_i)}$
            \hfill [{\it complexity for these four steps: $O(n_i)$}]
        \ENDFOR{
        \hfill
        [{\it total complexity for the for loop: $O(N)$}]}
        \STATE sort $\mu_{i}(\infty)-\mu_{i}(\Gamma_0)$ for all $i$, where ties are sorted based on their values of $v_{i}^2(\Gamma_0)$s
        \\ 
        \hfill
        [{\it complexity: $O(I\log I)$}]
        \STATE $\mathcal I_{\Gamma_0;k} \gets$ set of indices of matched sets with the $k$ smallest ranks of the $\mu_{i}(\infty)-\mu_{i}(\Gamma_0)$'s \\
        \hfill
        [{\it complexity: $O(I\log I)$}]
        \STATE $\mu(\Gamma_0;k) \gets \sum_{i \in \mathcal{I}_{\Gamma_0;k}} \mu_i({\Gamma}) + 
\sum_{i \notin \mathcal{I}_{\Gamma_0;k}} \mu_i(\infty)$
        \hfill
        [{\it complexity: $O(I)$}]
        \STATE $\sigma^2(\Gamma_0;k) \gets \sum_{i \in \mathcal{I}_{\Gamma_0; k}} v_i^2({\Gamma}_0)$
        \hfill
        [{\it complexity: $O(I)$}]
        \STATE $\tilde p_{\Gamma_0;k}\gets 1-\Phi\left(
    \frac{T - \mu(\Gamma_0; k)}{\sigma(\Gamma_0; k)}
    \right) \cdot 
    \mathbf{1}\{T \ge \mu(\Gamma_0; k)\}$
    \hfill
    [{\it complexity: $O(1)$}]
    \STATE 
    [{\bf Total complexity: $O(I\log I + \sum_{i=1}^In_i^2)$}]
    \end{algorithmic}
\end{algorithm}

We then consider the large-sample sensitivity analysis for general matched set studies when the number of sets $I$ is large. 
Algorithm \ref{alg:set} shows the pseudo code for implementing the approach in Section \ref{sec:set}. 
Again, we list the time complexity for each step in Algorithm \ref{alg:set}. 
Below we give some brief explanation. 
\begin{enumerate}[label=(\roman*), topsep=1ex,itemsep=-0.3ex,partopsep=1ex,parsep=1ex]
\item 
First, we consider the calculation of the $q_{ij}$'s 
and the realized value of the test statistic $T_\obs$. 
Consider again the $m$-statistic suggested by \citet{Rosenbaum2007m}. 
For each set $i$, there are $\binom{n_i}{2}$ pairwise differences. Following the same logic as that in Appendix \ref{sec:alg_pair}, the complexity for getting all the $q_{ij}$'s is $O(\sum_{i=1}^I n_i^2)$. 
The complexity for getting $T_\obs$ is then $O(N)$, recalling that $N=\sum_{i=1}^I n_i$ is the total number of units. 

\item Second, we consider the calculation of the $q_{i(j)}$'s, $Q_j$'s and $S_j$'s.
The complexity for getting $q_{i(j)}$'s is $O(\sum_{i=1}^I n_i \log(n_i))$, since we need to sort $n_i$ units within each matched set $i$.
The complexity for getting $Q_j$'s and $S_j$'s is $O(\sum_{i=1}^I n_i) = O(N)$, which follows immediately from their expressions. 

\item 
Third, we consider the calculation of $\mu_i(\Gamma_0)$, 
$v_i^2(\Gamma_0)$, 
$\mu_i(\infty)$
and 
$v_i^2(\infty)$, for all matched set $1\le i \le I$. 
Note that, for any set $i$, 
to find $\mu_i(\Gamma_0)$ and 
$v_i^2(\Gamma_0)$,
we only need to search over $n_i-1$ possible configurations of hidden confounding, 
and the mean and variance calculation for each configuration requires only $O(1)$ complexity. 
This is because we have already calculated the cumulative summations of the $q_{i(j)}$'s and $q_{i(j)}^2$'s beforehand. 
Besides, $\mu_i(\infty)$
and 
$v_i^2(\infty)$ are simply $q_{i(n_i)}$ and $0$, whose complexity is $O(1)$ for each set $i$. 
Therefore, the total complexity is $O(\sum_{i=1}^I n_i) = O(N)$.

\item Finally, we consider the calculation of $\mathcal{I}_{\Gamma_0; k}$, $\mu(\Gamma_0; k)$, $\sigma^2(\Gamma_0; k)$ and the $p$-value $\tilde{p}_{\Gamma_0; k}$. 
To get $\mathcal{I}_{\Gamma_0; k}$, it suffices to rank all the $I$ matched sets based on the $\mu_{i}(\infty)-\mu_{i}(\Gamma_0)$s, where we break ties based on the $v_{i}^2(\Gamma_0)$'s. Thus, the complexity for getting $\mathcal{I}_{\Gamma_0; k}$ is $O(I \log I)$.  
The complexities for getting $\mu(\Gamma_0; k)$, $\sigma^2(\Gamma_0; k)$ and the $p$-value $\tilde{p}_{\Gamma_0; k}$ are, respectively, $O(I), O(I)$ and $O(1)$, following immediately from their expressions. 
\end{enumerate}

From the above, the total complexity for calculating the $p$-value $\tilde p_{\Gamma_0;k}$ at given $k$ and $\Gamma_0$ is $O(I\log I + \sum_{i=1}^I n_i^2)$. 

We then consider the complexity for calculating the lower confidence limit $\tilde{\Gamma}_{(k)}$ for $\Gamma^\star_{(k)}$ at a given $k$. 
A subtlety here is that the monotonicity of the large-sample $p$-value is not easy to verify and may not be true, 
but we expect it to hold in most of the 
{\lxr cases},
especially when the asymptotic approximation works well. 
In practice, we will still use the binary search to find $\tilde{\Gamma}_{(k)}$. 
Furthermore, as long as the searched $\tilde{\Gamma}_{(k)}$ satisfies that the $p$-value $\tilde{p}_{\tilde{\Gamma}_{(k)}; k}$ is less than or equal to the significance level of interest, the resulting lower confidence limit will be valid. 
This 
{\lxr can be seen from}
the form of the confidence set in \eqref{eq:set_joint_set}. 
Thus, we will still use binary search to find $\tilde{\Gamma}_{(k)}$.
The number of $p$-values we need to calculate is then of order $O( \log (\bar\Gamma_k/\epsilon) )$, assuming that we have an upper bound $\bar\Gamma_k$ for $\tilde{\Gamma}_{(k)}$ {\lxr and the target precision is $\epsilon$}. 
Consequently, the total complexity for getting $\tilde{\Gamma}_{(k)}$ is of the order 
$O(\sum_{i=1}^I n_i^2 + (I\log I+N)\log (\bar\Gamma_k/\epsilon) )$, noticing that steps (i) and (ii) discussed before for computing $p$-values need only to be done once. 

Finally, we consider the complexity for calculating the lower confidence limits $\tilde{\Gamma}_{(k)}$'s for all the $\Gamma^\star_{(k)}$'s. 
Note that steps (i) and (ii) in the $p$-value calculation need only to be done once. 
Suppose that $\bar\Gamma_I$ denotes an upper bound of $\tilde{\Gamma}_{(I)}$, 
which will also be used as an upper bound for $\tilde{\Gamma}_{(k)}$'s for $k<I$. 
From the above, the total complexity for calculating the $\tilde{\Gamma}_{(k)}$'s  for all $1\le k\le I$ is then
$O(\sum_{i=1}^I n_i^2 + (I\log I+N)I\log (\bar\Gamma_k/\epsilon) )$. 
Similar to the discussion for matched pair case in Appendix \ref{sec:alg_pair}, 
in practice, 
when searching for $\tilde{\Gamma}_{(k)}$, we can search only on the interval $[1, \tilde{\Gamma}_{(k+1)}]$, for any $1\le k \le I-1$. 

\begin{rmk}\label{rmk:asym_complex}
  We consider a special but realistic case where all matched sets have bounded sizes. 
    In this case, the total complexity for getting lower confidence limits for all the the $\Gamma^\star_{(k)}$'s reduces to 
    $O(I^2 \log I\log (\bar\Gamma_k/\epsilon) )$.   
\end{rmk}

\subsection{Numerical experiments}

We conduct some numerical experiments to see the actual run time of our proposed method. 
We consider matched studies with $I=100,500,1000$ and $5000$ matched sets of equal sizes $n_1 = n_2 = \ldots = n_I \equiv n$ for $n=2, 3, 4, 5$ and $6$, respectively, 
and will conduct the asymptotic sensitivity analysis in Section \ref{sec:set}. 
We generate the control potential outcomes independently from the Gaussian distribution with mean zero and variance $0.5$, and assume that the true treatment effects are constant 2. 
We generate treatment assignments as in a matched completely randomized experiment. This is the ``favorable'' situation considered by \citet{rosenbaum2010book} for studying the power of sensitivity analyses; see Appendix \ref{sec:simulation}. 
For each combination of $(I, n)$ with $I\le 1000$, we simulate $10$ datasets, calculate the corresponding $95\%$ lower confidence limits $\tilde\Gamma_{(k)}$ for all $1\le k \le I$, and report the average running time in minutes in Table \ref{tab:runtime_1000}. 
To find these lower confidence limits, we use the \textsf{uniroot} function in R with convergence tolerance set at $10^{-4}$. 
The average running time takes less than half a minute when $I=100$, about 6 minutes when $I=500$, and 
about 25 minutes when $I=1000$. 
We do see the significant increase of running time as $I$ increases, which is roughly consistent with the $O(I^2 \log I )$ complexity discussed in Remark \ref{rmk:asym_complex}, assuming that the number of times needed to search for each confidence limit is of constant order.

\begin{table}[htb]
    \centering
    \small
    \caption{Average running time in minutes for calculating the lower confidence limits $\tilde\Gamma_{(k)}$ for all $1\le k\le I$. 
    }
    \label{tab:runtime_1000}
    \begin{tabular}{cccccccc}
    \toprule
     & $n=2$ & $n=3$ & $n=4$ & $n=5$ & $n=6$\\
    \midrule
    $I=100$
    & 0.2523364& 0.2574613 &0.2705687 &0.2776757& 0.2894428\\[5pt]
    $I=500$
    & 5.786684& 6.102555 &6.102555& 6.102555& 6.102555  \\[5pt]
    $I=1000$
    & 22.75612 &24.28257& 24.89604& 26.48302& 26.80329  \\
    \bottomrule
\end{tabular}
\end{table}

Moreover, for $I=5000$ and $2 \le n \le 6$, we generate $\tilde\Gamma_{(k)}$'s for $k=0.05I \times j$, with $j=1,2,\ldots, 20$. 
That is, instead of all $5000$ quantiles, we now construct only $5\%, 10\%, 
\ldots, 100\%$ quantiles of hidden biases. 
{\lxr For each set size $n$, we simulate $10$ datasets.}
The average running time  is shown in Table \ref{tab:runtime_5000}, which is less than 3 minutes in all cases. 
Note that these $20$ quantiles of hidden biases are in some sense representative for all quantiles, and can be computed rather quickly. In practice, we can also try a finer grid of quantiles, such as $1\%, 2\%, \ldots, 100\%$ quantiles. 
We can compute the confidence limits for all $5000$ quantiles of hidden biases, although it will take a much longer time; in practice, this can be sped up by parallelizing the computation for these quantiles. 
Nevertheless, 
when we visualize the lower confidence limits as in Figure \ref{fig:sen_lead}, the plot using a fine grid of quantiles, say 100 quantiles ranging from $1\%$ to $100\%$,  will be similar to that using all $5000$ quantiles of hidden biases. 
Importantly, in this way, we now only need to find lower confidence limits for much fewer, say 100, quantiles of hidden biases, which will not increase with $I$. 

\begin{table}[htb]
    \centering
    \small
    \caption{Average runtime for generating $\tilde\Gamma_{(k)}$'s for $k=0.05I \times j$, with $j=1,2,\ldots, 20$.
    }
    \label{tab:runtime_5000}
    \begin{tabular}{cccccccc}
    \toprule
     & $n=2$ & $n=3$ & $n=4$ & $n=5$ & $n=6$\\
    \midrule
    $I=5000$
    & 1.922614 &2.191779& 2.353492& 2.492236& 2.657214\\
    \bottomrule
\end{tabular}
\end{table}

\section{Proofs of Theorem \ref{thm:pair_vector_gamma} and Corollary \ref{thm:pair_vector_gamma}}\label{app:pair_vector}

To prove Theorem \ref{thm:pair_vector_gamma}, we need the following three lemmas. 
\begin{lemma}\label{lemma:Zq_sum_max}
Let $q_1$ and $q_2$ be two constants, 
and 
$Z_1 \equiv 1 - Z_2$ 
and 
$\overline{Z}_1 \equiv 1-\overline{Z}_2$
be two binary random variables with success probabilities 
\begin{align*}
    \pr(Z_1 = 1) = \frac{\exp(\overline{\gamma} u_1)}{\exp(\overline{\gamma} u_1) + \exp(\overline{\gamma} u_2)}, 
    \quad 
    \pr(\overline{Z}_1 = 1) = 
    \frac{\I(q_1\ge q_2) \exp(\overline{\gamma}) + \I(q_1 < q_2)}{1+\exp(\overline{\gamma})}, 
\end{align*}
where $u_1$ and $u_2$ are two constants in $[0,1]$. 
Then 
$
    \pr(Z_1q_1 + Z_2 q_2 \ge c) 
    \le 
    \pr(\overline{Z}_1q_1 + \overline{Z}_2 q_2 \ge c)
$
for any $c\in \mathbb{R}$. 
\end{lemma}
\begin{proof}[Proof of Lemma \ref{lemma:Zq_sum_max}]
If $q_1 = q_2$, then both $Z_1q_1 + Z_2 q_2$ and $\overline{Z}_1q_1 + \overline{Z}_2 q_2$ are constant $q_1$, and Lemma \ref{lemma:Zq_sum_max} holds obviously. 
Below we consider only the case in which $q_1 \ne q_2$. 
By symmetry and without loss of generality, we assume $q_1 > q_2$.
Then by definition, 
\begin{align*}
    Z_1q_1+Z_2q_2=
    \begin{cases}
    q_1 & \text{with probability }\frac{\exp(\overline{\gamma} u_1)}{\exp(\overline{\gamma} u_1) + \exp(\overline{\gamma} u_2)},\\
    q_2 & \text{with probability }\frac{\exp(\overline{\gamma} u_2)}{\exp(\overline{\gamma} u_1) + \exp(\overline{\gamma} u_2)}.
    \end{cases}
\end{align*}
and 
\begin{align*}
    \overline Z_1q_1+\overline Z_2q_2=
    \begin{cases}
    q_1 & \text{with probability }\frac{\exp(\overline{\gamma})}{1+\exp(\overline{\gamma})}, \\
    q_2 & \text{with probability }\frac{1}{1+\exp(\overline{\gamma})}.
    \end{cases}
\end{align*}
For $c\le q_2$, 
$
\pr(Z_1q_1+Z_2q_2\ge c)=\pr(\overline Z_1q_1+\overline Z_2q_2\ge c)=1.
$
For $q_2< c\le q_1$, 
\begin{align*}
\pr(Z_1q_1+Z_2q_2\ge c)&=\pr(Z_1q_1+Z_2q_2=q_1)=\frac{\exp(\overline{\gamma} u_1)}{\exp(\overline{\gamma} u_1) + \exp(\overline{\gamma} u_2)}
\le 
\frac{\exp(\overline{\gamma} )}{\exp(\overline{\gamma}) + 1}
\\
& =\pr(\overline Z_1q_1+\overline Z_2q_2=q_1)
=\pr(\overline Z_1q_1+\overline Z_2q_2\ge c), 
\end{align*}
where the last inequality holds because 
$\exp(\overline{\gamma} u_1)/\{\exp(\overline{\gamma} u_1) + \exp(\overline{\gamma} u_2)\}$ is increasing in $u_1$ and decreasing in $u_2$. 
For $c> q_1$, 
$
\pr(Z_1q_1+Z_2q_2\ge c)=\pr(\overline Z_1q_1+\overline Z_2q_2\ge c)=0.
$
From the above, we have
$
    \pr(Z_1q_1 + Z_2 q_2 \ge c) 
    \le 
    \pr(\overline{Z}_1q_1 + \overline{Z}_2 q_2 \ge c)
$
for any $c\in \mathbb{R}$. 
Therefore, Lemma \ref{lemma:Zq_sum_max} holds. 
\end{proof}

\begin{lemma}\label{lemma:inverse cdf}
Let $X$ and $\overline{X}$ be two random variables satisfying $\pr(X\ge c) \le \pr(\overline{X} \ge c)$ for any $c\in \mathbb{R}$, 
and $F$ and $\overline{F}$ be the distribution functions of $X$ and $\overline{X}$. 
Then 
the quantile functions of $X$ and $\overline{X}$ satisfy 
$F^{-1}(p) \le \overline{F}^{-1}(p)$ for all $p \in (0, 1)$. 
\end{lemma}
\begin{proof}[Proof of Lemma \ref{lemma:inverse cdf}]
From the condition in Lemma \ref{lemma:inverse cdf}, for any $c\in \mathbb{R}$, 
$$
F(c) = 1 - \pr(X>c) = 1 - \lim_{d\rightarrow c+}\pr(X\ge d) \ge 1 - \lim_{d\rightarrow c+}\pr(\overline{X}\ge d) = 1 - \pr(\overline{X}>c) = \overline{F}(c).
$$
This implies that $\{x\in \mathbb{R}: F(x)\ge p\} \supset \{x\in \mathbb{R}: \overline{F}(x)\ge p\}$ for any $p\in (0,1)$. 
By the definition of quantile functions, 
we have $F^{-1}(p)=\inf\{x: F(x)\ge p\} \le \inf\{x: \overline{F}(x)\ge p\} = \overline{F}^{-1}(p)$. 
Therefore, Lemma \ref{lemma:inverse cdf} holds. 
\end{proof}

\begin{lemma}\label{lemma:sum_order}
Let $(X_1, \ldots, X_I)$ and $(\overline{X}_1, \ldots, \overline{X}_I)$ be $2I$ random variables satisfying 
(i) $(X_1, \ldots, X_I)$ are mutually independent, 
(ii) $(\overline{X}_1, \ldots, \overline{X}_I)$ are mutually independent,
and 
(iii) $\pr(X_i \ge c) \le \pr(\overline{X}_i\ge c)$ for all $1\le i\le I$ and  $c \in \mathbb{R}$. 
Then 
$
\pr(\sum_{i=1}^I X_i \ge c ) \le \pr( \sum_{i=1}^I \overline{X}_i \ge c)
$
for all $c\in\mathbb R$. 
\end{lemma}
\begin{proof}[Proof of Lemma \ref{lemma:sum_order}]
Let $F_i$ and $\overline{F}_i$ be the distribution functions of $X_i$ and $\overline{X}_i$, respectively, for $1\le i\le I$, 
and let $(U_1, U_2, \ldots, U_I)$ be $I$ independent and identically distributed (i.i.d.) random variables following the same distribution as $\Unif(0,1)$. 
Then by the property of quantile functions, 
$
(X_1, \ldots, X_I) \sim (F_1^{-1}(U_1), \ldots, F_I^{-1}(U_I))
$
and 
$
(\overline{X}_1, \ldots, \overline{X}_I) \sim (\overline{F}_1^{-1}(U_1), \ldots, \overline{F}_I^{-1}(U_I)). 
$
From Lemma \ref{lemma:inverse cdf}, we have 
$
    \sum_{i=1}^I X_i
    \sim 
    \sum_{i=1}^I F_i^{-1}(U_i) 
    \le 
    \sum_{i=1}^I \overline{F}_i^{-1}(U_i)
    $ $\sim \sum_{i=1}^I \overline{X}_i, 
$
which immediately implies that for any $c\in \mathbb{R}$, 
$$
\pr\left(\sum_{i=1}^I X_i \ge c \right) 
= 
\pr\left( \sum_{i=1}^I F_i^{-1}(U_i)  \ge c\right)
\le 
\pr\left( \sum_{i=1}^I \overline{F}_i^{-1}(U_i)  \ge c\right)
= \pr\left( \sum_{i=1}^I \overline{X}_i \ge c\right). 
$$
Therefore, Lemma \ref{lemma:sum_order} holds. 
\end{proof}

\begin{proof}[\bf Proof of Theorem \ref{thm:pair_vector_gamma}]
As discussed in Section \ref{sec:test_rand_sen}, 
under Fisher's null $H_0$, $q_{ij}$'s are constant for all $i$ and $j$. 
For $1\le i\le I$, we define $\overline{Z}_{i1} \equiv 1 - \overline{Z}_{i2}$ as a binary random variable with the following success probability:
\begin{align*}
    \pr(\overline{Z}_{i1} = 1) = 
    \frac{\I(q_{i1}\ge q_{i2}) \exp(\overline{\gamma}_i) + \I(q_{i1} < q_{i2})}{1+\exp(\overline{\gamma}_i)},
\end{align*}
{\lxr where $(\overline{Z}_{i1}, \overline{Z}_{i2})$'s are mutually independent across all $i$.}
Let $X_i=Z_{i1}q_{i1}+Z_{i2}q_{i2}$ and $\overline X_i=\overline Z_{i1}q_{i1}+\overline Z_{i2}q_{i2}$ for $1\le i \le I$. 
Then by definition, 
$T = \sum_{i=1}^I X_i$ and 
$\overline{T}({\bs{\Gamma}}) \sim \sum_{i=1}^I \overline{X}_i$. 
If the bias $\bs{\Gamma}^\sta$ in all matched pairs are bounded by $\bs{\Gamma}$, i.e., $\bs{\Gamma}^\sta\vecle {\bs{\Gamma}}$, 
then by Lemma \ref{lemma:Zq_sum_max}, $\pr (X_i\ge c)\le\pr(\overline X_i\ge c)$ for all $c\in\mathbb R.$
From Lemma \ref{lemma:sum_order} and by the mutual independence of $(Z_{i1}, Z_{i2})$'s and $(\overline{Z}_{i1}, \overline{Z}_{i2})$'s across all matched pairs, 
we then have for all $c\in \mathbb{R}$, 
$$\pr (T\ge c)=\pr\left(\sum_{i=1}^IX_i\ge c\right)\le\pr\left(\sum_{i=1}^I\overline X_i\ge c\right)=\pr(\overline{T}({\bs{\Gamma}}) \ge c).$$
Therefore, Theorem \ref{thm:pair_vector_gamma} holds. 
\end{proof}

To prove Corollary \ref{cor:pair_vector_gamma}, we need the following lemma. 

\begin{lemma}\label{lemma:survival_dominant_by_uniform}
	For any random variable $W$, define  
	$
	G(w) = \pr(W\geq w). 
	$
	Then $\pr( G(W) \leq \alpha )\leq \alpha$ for all $\alpha\in (0,1)$.
\end{lemma}
\begin{proof}[Proof of Lemma \ref{lemma:survival_dominant_by_uniform}]
Note that $G(w)$ is 
{\lxr decreasing}
in $w$. 
For any $p\in (0,1)$, we 
define
$
G^{-1}(p) = \inf \{w: G(w)\leq p\}. 
$
For any given $\alpha\in (0,1)$, we consider the following two cases. 
If $G(G^{-1}(\alpha))\leq \alpha$, then 
\begin{align*}
\pr(G(W)\leq \alpha) & \leq \pr( W \geq G^{-1}(\alpha))=G(G^{-1}(\alpha))\leq \alpha.
\end{align*}
Otherwise, $G(G^{-1}(\alpha))>\alpha$, and 
\begin{align*}
	\pr(G(W)\leq \alpha) \leq& \pr(W> G^{-1}(\alpha))=\lim_{m\rightarrow \infty} \pr\left(W\geq G^{-1}(\alpha)+\frac{1}{m}\right)
	= \lim_{m\rightarrow \infty} G\left(G^{-1}(\alpha)+\frac{1}{m}\right)\leq \alpha.
\end{align*}
Therefore, Lemma \ref{lemma:survival_dominant_by_uniform} holds. 
\end{proof}

\begin{proof}[\bf Proof of Corollary \ref{cor:pair_vector_gamma}]
Define $G(c) = \pr(T\ge c)$. 
From Theorem \ref{thm:pair_vector_gamma}, we can know that $G(c) \le \overline{G}_{\bs{\Gamma}}(c)$ for all $c\in \mathbb{R}$. 
This then implies that, 
for $\alpha\in (0,1),$
\begin{align*}
    \pr(\overline p_{\bs \Gamma}\le\alpha)&=
    \pr\left\{ \overline G_{\bs{\Gamma}}(T)\le\alpha \right\}
    \le \pr(G(T)\le \alpha) \le \alpha,
\end{align*}
where the last inequality follows from Lemma \ref{lemma:survival_dominant_by_uniform}. 
Therefore, Corollary \ref{cor:pair_vector_gamma} holds. 
\end{proof}

\section{Proofs of Theorem \ref{thm:pair_quantile_gamma} and Corollary \ref{cor:pval_pair_quantile}}\label{app:pair_quantile}

To prove Theorem \ref{thm:pair_quantile_gamma}, we need the following lemma.
\begin{lemma}\label{lemma:two_pair_two_prob}
Let $q_1 \le q_2$ and $r_1 \le r_2$ be four constants satisfying $q_2 - q_1 \le r_2-r_1$. 
Let 
$Z_2 \equiv 1-Z_1$,  $W_2 \equiv 1 - W_1$, 
$\overline{Z}_2 \equiv 1-\overline{Z}_1$,  $\overline{W}_2 \equiv 1 - \overline{W}_1$ 
be four binary random variables with success probabilities
$
\pr(Z_2 = 1) = \overline{p}, \pr(W_2= 1) = \underline{p}, 
$
$
   \pr(\overline{Z}_2 = 1) = \underline{p}
$
and 
$\pr(\overline{W}_2 = 1) = \overline{p}, 
$
for some $0\le \underline{p} \le\overline{p} \le 1$. 
Then we must have, for any $c\in \mathbb{R}$, 
\begin{align*}
    \pr(Z_1 q_1 + Z_2 q_2 + W_1 r_1 + W_2 r_2\ge c)
    \le 
    \pr(\overline{Z}_1 q_1 + \overline{Z}_2 q_2 + \overline{W}_1 r_1 + \overline{W}_2 r_2\ge c). 
\end{align*}
\end{lemma}
\begin{proof}[Proof of Lemma \ref{lemma:two_pair_two_prob}]
Let 
$X = Z_1 q_1 + Z_2 q_2 + W_1 r_1 + W_2 r_2$
and 
$\overline{X} = \overline{Z}_1 q_1 + \overline{Z}_2 q_2 + \overline{W}_1 r_1 + \overline{W}_2 r_2$. 
From the conditions in Lemma \ref{lemma:two_pair_two_prob}, 
$q_1+r_1\le q_2+r_1\le q_1+r_2\le q_2+r_2$. 
Below we first consider the case in which these four quantities are all different. 
By definition, 
\begin{align*}
    X = 
    \begin{cases}
    q_1 + r_1, & \text{with probability } (1-\overline{p}) (1-\underline{p})
    \\
    q_2 + r_1, & \text{with probability } \overline{p} (1-\underline{p})
    \\
    q_1 + r_2, & \text{with probability } (1-\overline{p}) \underline{p}
    \\
    q_2 + r_2, & \text{with probability } \overline{p} \cdot \underline{p}, 
    \end{cases}
\end{align*}
and 
\begin{align*}
    \overline{X} = 
    \begin{cases}
    q_1 + r_1, & \text{with probability } (1-\underline{p}) (1-\overline{p})
    \\
    q_2 + r_1, & \text{with probability } \underline{p} (1-\overline{p})
    \\
    q_1 + r_2, & \text{with probability } (1-\underline{p}) \overline{p}
    \\
    q_2 + r_2, & \text{with probability } \underline{p} \cdot \overline{p}
    \end{cases}
\end{align*}
This implies that 
\begin{align*}
    \pr(X\ge q_1 + r_1) & = 1 = \pr(\overline{X} \ge q_1 + r_1), 
    \\
    \pr(X\ge q_2 + r_1) & = 1 - (1-\overline{p}) (1-\underline{p}) = \pr(\overline{X} \ge q_2 + r_1), 
    \\
    \pr(X\ge q_1 + r_2) & = (1-\overline{p}) \underline{p} + \overline{p} \cdot \underline{p} = \underline{p}  \le \overline{p} =  (1-\underline{p}) \overline{p} + \underline{p} \cdot \overline{p}
    =
    \pr(\overline{X} \ge q_1 + r_2), 
    \\
    \pr(X\ge q_2 + r_2) & = \overline{p} \cdot \underline{p}  = \pr(\overline{X} \ge q_2 + r_2). 
\end{align*}
Consequently, we must have 
$
\pr(X\ge c) \le \pr(\overline{X} \ge c)
$
for all $c$.
By the same logic, we can verify that this also holds when some elements of $\{q_1+r_1, q_2+r_1, q_1+r_2, q_2+r_2\}$ are equal. 
Therefore, Lemma \ref{lemma:two_pair_two_prob} holds. 
\end{proof}
\begin{proof}[\bf Proof of Theorem \ref{thm:pair_quantile_gamma}]
Recall that $\mathcal{I}^\sta_k$ is the set of indices of the matched sets with the $k$ smallest bias $\Gamma_i^\sta$'s. 
Let $\bs{\Gamma} = \Gamma_0 \cdot \bs{1}_{\mathcal{I}^\sta_k} + \infty \cdot \bs{1}_{\mathcal{I}^\sta_k}^\complement$. 
Then $\Gamma_{(k)}^\sta\le \Gamma_0$ is equivalent to $\bs{\Gamma}^\sta \vecle \bs{\Gamma}$. 
From Theorem \ref{thm:pair_vector_gamma}, 
$\pr(T\ge c) \le \pr(\overline{T}({\bs{\Gamma}}) \ge c)$ for all $c\in \mathbb{R}$, 
where $\overline{T}({\bs{\Gamma}})$ is defined the same as in Theorem \ref{thm:pair_vector_gamma}. 
Define $\mathcal{I}_k$ the same as in Theorem \ref{thm:pair_quantile_gamma}, 
and 
$\tilde{\bs{\Gamma}} = \Gamma_0 \cdot \bs{1}_{\mathcal{I}_k} + \infty \cdot \bs{1}_{\mathcal{I}_k}^\complement$. 
We define $\overline{T}({\tilde{\bs{\Gamma}})}$ analogously as in Theorem \ref{thm:pair_vector_gamma}. 

Below we compare the distributions of $\overline{T}({\bs{\Gamma}})$ and $\overline{T}({\tilde{\bs{\Gamma}}})$. 
Let $q_{i1}' = \min\{q_{i1}, q_{i2}\}$ and $q_{i2}' = \max\{q_{i1}, q_{i2}\}$, for $1\le i \le I$. 
By definition, $\overline{T}({\bs{\Gamma}}) \sim \sum_{i=1}^I (\overline{Z}_{i1} q_{i1}' + \overline{Z}_{i2} q_{i2}' )$, 
where $\overline{Z}_{i2} \equiv 1-\overline{Z}_{i1}$ are mutually independent across all $i$ and 
\begin{align*}
    \pr( \overline{Z}_{i2} = 1 \mid \mathcal{Z}) = 
    \begin{cases}
    \Gamma_0/(1+\Gamma_0), & \text{if } i \in \mathcal{I}^\sta_k, \\
    1, & \text{otherwise}. 
    \end{cases}
\end{align*}
Similarly, 
$\overline{T}({\tilde{\bs{\Gamma}}}) \sim \sum_{i=1}^I (\tilde{Z}_{i1} q_{i1}' + \tilde{Z}_{i2} q_{i2}' )$, 
where $\tilde{Z}_{i2} \equiv 1-\tilde{Z}_{i1}$ are mutually independent across all $i$ and 
\begin{align*}
    \pr( \tilde{Z}_{i2} = 1 \mid \mathcal{Z}) = 
    \begin{cases}
    \Gamma_0/(1+\Gamma_0), & \text{if } i \in \mathcal{I}_k, \\
    1, & \text{otherwise}. 
    \end{cases}
\end{align*}
Define $\mathcal{C} = ( \mathcal{I}_k^\complement \cap {\mathcal I_k^\sta}^\complement )\cup ( \mathcal{I}_k \cap {\mathcal I_k^\sta} )$. 
By definition, $\overline{Z}_{i2}\sim  \tilde{Z}_{i2}$ for all $i\in \mathcal{C}$. 
Consider any $i \in \mathcal{I}_k \setminus {\mathcal I_k^\sta}$ and any $j \in {\mathcal I_k^\sta} \setminus \mathcal{I}_k$. 
Then, by the construction of $\mathcal{I}_k$, we must have $q_{i2}'-q_{i1}' \le q_{j2}'-q_{j1}'$. 
From Lemma \ref{lemma:two_pair_two_prob} and the fact that 
$\pr(\overline{Z}_{i2} = 1\mid \mathcal{Z}) = \pr(\tilde{Z}_{j2}=1 \mid \mathcal{Z}) = 1$ and  $\pr(\overline{Z}_{j2} = 1\mid \mathcal{Z}) = \pr(\tilde{Z}_{i2} = 1 \mid \mathcal{Z}) = \Gamma_0/(1+\Gamma_0) \le 1$, we have, for all $c\in \mathbb{R}$, 
\begin{align}\label{eq:order_I_k_J_k}
\pr\left( \overline{Z}_{i1} q_{i1}' + \overline{Z}_{i2} q_{i2}' + \overline{Z}_{j1} q_{j1}' + \overline{Z}_{j2} q_{j2}'  \ge c \mid \mathcal{Z} \right)
\le 
\pr\left( 
\tilde{Z}_{i1} q_{i1}' + \tilde{Z}_{i2} q_{i2}' + \tilde{Z}_{j1} q_{j1}' + \tilde{Z}_{j2} q_{j2}' \ge c \right). 
\end{align}
From Lemma \ref{lemma:sum_order} and the mutual independence of $\overline{Z}_{i2}$'s and $\tilde{Z}_{i2}$'s over all $i$, we then have, for all $c\in \mathbb{R}$, 
\begin{align*}
    & \quad \ \pr\left( \overline{T}({\bs{\Gamma}}) \ge c\right) \\
    & = 
    \pr\left\{ \sum_{i\in \mathcal{C}} (\overline{Z}_{i1} q_{i1}' + \overline{Z}_{i2} q_{i2}' ) 
    + 
    \sum_{i\in \mathcal{I}_k \setminus {\mathcal I_k^\sta}} (\overline{Z}_{i1} q_{i1}' + \overline{Z}_{i2} q_{i2}' ) 
    + 
    \sum_{j \in {\mathcal I_k^\sta} \setminus \mathcal{I}_k} (\overline{Z}_{j1} q_{j1}' + \overline{Z}_{j2} q_{j2}' ) 
    \ge c \mid \mathcal{Z}\right\}
    \\
    & \le 
    \pr\left\{
    \sum_{i\in \mathcal{C}} (\tilde{Z}_{i1} q_{i1}' + \tilde{Z}_{i2} q_{i2}' ) 
    + 
    \sum_{i\in \mathcal{I}_k \setminus {\mathcal I_k^\sta}} (\tilde{Z}_{i1} q_{i1}' + \tilde{Z}_{i2} q_{i2}' ) 
    + 
    \sum_{j \in {\mathcal I_k^\sta} \setminus \mathcal{I}_k} (\tilde{Z}_{j1} q_{j1}' + \tilde{Z}_{j2} q_{j2}' )
    \ge c \mid \mathcal{Z} \right\}
    \\
    & = 
    \pr\left( \overline{T}({\tilde{\bs{\Gamma}}}) \ge c \right),  
\end{align*}
where the last inequality holds because $\mathcal{I}_k \setminus {\mathcal I_k^\sta}$ and ${\mathcal I_k^\sta} \setminus \mathcal{I}_k$ must have the same size and we can pair them in an arbitrary way and use the inequality \eqref{eq:order_I_k_J_k}. 

From the above and Theorem \ref{thm:pair_vector_gamma}, we have, for all $c\in \mathbb{R}$
$$
\pr\left( T\ge c\right) 
\le 
\pr\left( \overline{T}({\bs{\Gamma}}) \ge c \right) 
\le 
\pr\left( \overline{T}(\tilde{\bf\Gamma}) \ge c \right) = \pr\left( \overline{T}({\Gamma_0};k)\ge c\right),$$
where the last 
{\lxr equality}
follows from the definition of $\overline{T}({\Gamma_0};k)$. 
Therefore, Theorem \ref{thm:pair_quantile_gamma} holds. 
\end{proof}

\begin{proof}[\bf Proof of Corollary \ref{cor:pval_pair_quantile}]
Define $G(c) = \pr(T\ge c)$. 
From Theorem \ref{thm:pair_quantile_gamma} and Lemma \ref{lemma:survival_dominant_by_uniform}, we can know that $G(c) \le \overline{G}_{\Gamma;k}(c)$ for all $c\in \mathbb{R}$. 
This then implies that, 
for $\alpha\in (0,1),$
\begin{align*}
    \pr(\overline p_{\Gamma;k}\le\alpha)&=
    \pr\left\{ \overline G_{\Gamma;k}(T)\le\alpha \right\}
    \le \pr(G(T)\le \alpha) \le \alpha.
\end{align*}
Therefore, Corollary \ref{cor:pval_pair_quantile} holds.
\end{proof}

\section{Proofs of Corollary \ref{cor:pair_ci_gamma3} and Theorem \ref{thm:pair_joint_ci}}\label{app:pair_ci}

To prove Corollary \ref{cor:pair_ci_gamma3}, we need the following lemma. 

\begin{lemma}\label{lemma:property_p_gamma_k}
$\overline{p}_{\Gamma_0;k}$ is increasing in $\Gamma_0$ and decreasing in $k$. 
\end{lemma}

\begin{proof}[Proof of Lemma \ref{lemma:property_p_gamma_k}]
Recall that  
$\overline{p}_{\Gamma_0;k} = \overline{G}_{\Gamma_0;k}(T)$, 
where $\overline{G}_{\Gamma_0;k}$ is the tail probability of $\overline{T}(\Gamma_0;k) \sim \overline{T}(\Gamma_0 \cdot \bs{1}_{\mathcal{I}_k} + \infty \cdot \bs{1}_{\mathcal{I}_k}^\complement)$. 
For any $\Gamma_0 \le \Gamma_0'$ and $k \ge k'$,  
by definition, we must have 
$
\Gamma_0 \cdot \bs{1}_{\mathcal{I}_k} + \infty \cdot \bs{1}_{\mathcal{I}_k}^\complement 
\vecle 
\Gamma_0' \cdot \bs{1}_{\mathcal{I}_{k'}} + \infty \cdot \bs{1}_{\mathcal{I}_{k'}}^\complement. 
$
By the same logic as Theorem \ref{thm:pair_vector_gamma}, we can know that 
$\overline{T}(\Gamma_0 \cdot \bs{1}_{\mathcal{I}_k} + \infty \cdot \bs{1}_{\mathcal{I}_k}^\complement)$ is stochastically smaller than or equal to $\overline{T}(\Gamma_0' \cdot \bs{1}_{\mathcal{I}_{k'}} + \infty \cdot \bs{1}_{\mathcal{I}_{k'}}^\complement)$. 
This immediately implies that $\overline{p}_{\Gamma_0;k} \le \overline{p}_{\Gamma_0';k'}$. 
Therefore, $\overline{p}_{\Gamma_0;k}$ is increasing in $\Gamma_0$ and decreasing in $k$. 
i.e., Lemma \ref{lemma:property_p_gamma_k} holds.
\end{proof}

\begin{proof}[\bf Proof of Corollary \ref{cor:pair_ci_gamma3}]
The equivalent form of the confidence sets in Corollary \ref{cor:pair_ci_gamma3} follows immediately from Lemma \ref{lemma:property_p_gamma_k}. 
The validity of the confidence sets in Corollary \ref{cor:pair_ci_gamma3} is a direct corollary of Theorem \ref{thm:pair_joint_ci}, which will be proved immediately. 
\end{proof}

\begin{proof}[\bf Proof of Theorem \ref{thm:pair_joint_ci}]
First, we prove that 
$
\Gamma_{(k)}^\sta \in 
\{\Gamma_0\ge 1: \overline{p}_{\Gamma_0;k} > \alpha\}
\Longrightarrow 
\bs{\Gamma}^\sta \in 
\bigcap_{\Gamma_0: \overline{p}_{\Gamma_0;k} \le \alpha} \mathcal{S}_{\Gamma_0;k}^\complement. 
$
Suppose that $\Gamma_{(k)}^\sta \in 
\{\Gamma_0\ge 1: \overline{p}_{\Gamma_0;k} > \alpha\}$. 
Then $\overline p_{\Gamma_{(k)}^\sta;k} > \alpha$. 
For any $\Gamma_0$ such that $\overline{p}_{\Gamma_0;k} \le \alpha$, from Lemma \ref{lemma:property_p_gamma_k}, we must have $\Gamma_0 < \Gamma_{(k)}^\sta$, which implies that $\bs{\Gamma}^\sta \notin \mathcal{S}_{\Gamma_0;k}$. 
Therefore, we must have $\bs{\Gamma}^\sta \in 
\bigcap_{\Gamma_0: \overline{p}_{\Gamma_0;k} \le \alpha} \mathcal{S}_{\Gamma_0;k}^\complement$. 

Second, we prove that 
$
\bs{\Gamma}^\sta \in 
\bigcap_{\Gamma_0: \overline{p}_{\Gamma_0;k} \le \alpha} \mathcal{S}_{\Gamma_0;k}^\complement
\Longrightarrow 
\Gamma_{(k)}^\sta \in 
\{\Gamma_0\ge 1: \overline{p}_{\Gamma_0;k} > \alpha\}. 
$
We prove this by contradiction. 
Suppose that $\bs{\Gamma}^\sta \in 
\bigcap_{\Gamma_0: \overline{p}_{\Gamma_0;k} \le \alpha} \mathcal{S}_{\Gamma_0;k}^\complement$ and $\Gamma_{(k)}^\sta \notin
\{\Gamma_0\ge 1: \overline{p}_{\Gamma_0;k} > \alpha\}.$ 
The latter implies that $\overline p_{\Gamma_{(k)}^\sta;k}  \le \alpha$, 
and the former then implies that $\bs{\Gamma}^\sta \in \mathcal{S}_{\Gamma_{(k)}^\sta;k}^\complement$. 
However, $\bs{\Gamma}^\sta \in \mathcal{S}_{\Gamma_{(k)}^\sta;k}$ by definition, resulting in a contradiction.

Third, we prove that 
$
I^\sta(\Gamma_0) \in \{I-k: \overline{p}_{\Gamma_0;k} > \alpha, 0\le k \le I\}
\Longrightarrow
\bs{\Gamma}^\sta \in 
\bigcap_{k: \overline{p}_{\Gamma_0;k} \le \alpha} \mathcal{S}_{\Gamma_0;k}^\complement. 
$
Suppose that $I^\sta(\Gamma_0) \in \{I-k: \overline{p}_{\Gamma_0;k} > \alpha, 1\le k \le I\}$. 
Then $\overline{p}_{\Gamma_0;I-I^\sta(\Gamma_0)} > \alpha$. 
For any $k$ such that $\overline{p}_{\Gamma_0;k} \le \alpha$, from Lemma \ref{lemma:property_p_gamma_k}, we must have $k > I-I^\sta(\Gamma_0)$ or equivalently $I^\sta(\Gamma_0) > I-k$, 
which by definition implies that $\bs{\Gamma}^\sta \notin \mathcal{S}_{\Gamma_0;k}$. 
Therefore, we must have $\bs{\Gamma}^\sta \in 
\bigcap_{k: \overline{p}_{\Gamma_0;k} \le \alpha} \mathcal{S}_{\Gamma_0;k}^\complement.$

Fourth, we prove that 
$
\bs{\Gamma}^\sta \in 
\bigcap_{k: \overline{p}_{\Gamma_0;k} \le \alpha} \mathcal{S}_{\Gamma_0;k}^\complement
\Longrightarrow
I^\sta(\Gamma_0) \in \{I-k: \overline{p}_{\Gamma_0;k} > \alpha, 0\le k \le I\}.
$
We prove this by contradiction. 
Suppose that $\bs{\Gamma}^\sta \in 
\bigcap_{k: \overline{p}_{\Gamma_0;k} \le \alpha} \mathcal{S}_{\Gamma_0;k}^\complement$ and $I^\sta(\Gamma_0) \notin \{I-k: \overline{p}_{\Gamma_0;k} > \alpha, 0\le k \le I\}$.
The latter implies that 
$\overline p_{\Gamma_0;I - I^\sta(\Gamma_0)} \le \alpha$, 
and the former then implies that $\bs{\Gamma}^\sta \in  \mathcal{S}_{\Gamma_0;I - I^\sta(\Gamma_0)}^\complement$. 
However, $\bs{\Gamma}^\sta \in  \mathcal{S}_{\Gamma_0;I - I^\sta(\Gamma_0)}$ by definition, resulting in a contradiction. 

Fifth, we prove the validity of the confidence set in \eqref{eq:pair_joint_set}. Note that 
\begin{align*}
\pr\Big( \bs{\Gamma}^\sta \in \bigcap_{\Gamma_0, k: \overline{p}_{\Gamma_0;k} \le \alpha} \mathcal{S}_{\Gamma_0;k}^\complement  \Big) 
& = 
1 - \pr\Big( \bs{\Gamma}^\sta \in \bigcup_{\Gamma_0, k: \overline{p}_{\Gamma_0;k} \le \alpha} \mathcal{S}_{\Gamma_0;k} \Big). 
\end{align*}
If $\bs{\Gamma}^\sta \in \bigcup_{\Gamma_0, k: \overline{p}_{\Gamma_0;k} \le \alpha} \mathcal{S}_{\Gamma_0;k}$, then there must exist $(\Gamma_0, k)$ such that $\overline{p}_{\Gamma_0;k} \le \alpha$ and $\bs{\Gamma}^\sta \in \mathcal{S}_{\Gamma_0;k}$, 
where the latter further implies that $\Gamma_{(k)}^\sta \le \Gamma_0$. 
Define $G(c) = \pr(T\ge c)$ for all $c\in \mathbb{R}$. 
From Theorem \ref{thm:pair_quantile_gamma}, we must have $G(T) \le \overline{p}_{\Gamma_0;k} \le \alpha$. Thus, 
\begin{align*}
\pr\Big( \bs{\Gamma}^\sta \in \bigcap_{\Gamma_0, k: \overline{p}_{\Gamma_0;k} \le \alpha} \mathcal{S}_{\Gamma_0;k}^\complement  \Big) 
& = 
1 - \pr\Big( \bs{\Gamma}^\sta \in \bigcup_{\Gamma_0, k: \overline{p}_{\Gamma_0;k} \le \alpha} \mathcal{S}_{\Gamma_0;k} \Big)
\ge 1 - \pr(G(T)\le \alpha) \ge 1 -\alpha,
\end{align*}
where the last inequality holds due to Lemma \ref{lemma:survival_dominant_by_uniform}. 

From the above, Theorem \ref{thm:pair_joint_ci} holds. 
\end{proof}

\section{Proofs of Theorem \ref{thm:set_vector_gamma} and Corollary \ref{cor:pval_set_vector_gamma}}\label{app:set_vector}
To prove Theorem \ref{thm:set_vector_gamma}, we need the following three lemmas.
\begin{lemma}\label{lemma:lindfellerclt} 
For each $n$, let $X_{n,m},1\le m\le n$, be independent random variables with $\E X_{n,m}=0.$ Suppose
\begin{itemize}
    \item[(i)] $\sum_{m=1}^n \E X_{n,m}^2\to\sigma^2>0$;
    \item[(ii)] for all $\epsilon>0,\lim_{n\to\infty}\sum_{m=1}^n\E(|X_{n,m}|^2;|X_{n,m}|>\epsilon)=0$.
\end{itemize}
Then $S_n=X_{n,1}+\cdots+X_{n,n} \converged \mathcal{N}(0,1)$ as $n\to\infty$.
\end{lemma}
\begin{proof}[Proof of Lemma \ref{lemma:lindfellerclt}]
Lemma \ref{lemma:lindfellerclt} is the Lindeberg-Feller theorem; see, e.g., \citet{durrett2010probability}. 
\end{proof}

\begin{lemma}\label{lemma:match_clt}
If 
Condition \ref{cond:clt} holds, then 
\begin{align*}
    ( T - \E T )/\Var^{1/2}(T) \converged \mathcal{N}(0,1). 
\end{align*}
\end{lemma}
\begin{proof}[Proof of Lemma \ref{lemma:match_clt}]
Recall that $T=\sum_{i=1}^I T_i = \sum_{i=1}^I \sum_{j=1}^{n_i} Z_{ij}q_{ij}$,  $T_i$'s are mutually independent across all matched sets, 
and $\mu_i = \E(T_i)$ and $v^2_i = \Var(T_i)$ are the mean and variance of $T_i$ for $1\le i \le I$. 
By the mutual independence of $T_i$'s, $T$ has mean $\mu = \sum_{i=1}^I \mu_i$ and variance $\sigma^2 = \sum_{i=1}^I v_i^2$. 
Below we prove the asymptotic Gaussianity of $T$ using Lemma \ref{lemma:lindfellerclt}. 

We first bound $|\sum_{j=1}^{n_i} Z_{ij}q_{ij}-\mu_i|$ and $v^2_i$. 
By definition, 
$q_{i(1)} \le \sum_{j=1}^{n_i} Z_{ij}q_{ij} \le q_{i(n_i)}$. 
This implies that 
$q_{i(1)} \le \mu_i \le q_{i(n_i)}$ and thus $|\sum_{j=1}^{n_i} Z_{ij}q_{ij}-\mu_i| \le q_{i(n_i)} - q_{i(1)} = R_i$. 
With slight abuse of notation, let $Z_{i(j)}$ denote the corresponding treatment assignment indicator for $q_{i(j)}$, i.e., 
$T_s = \sum_{j=1}^{n_i} Z_{ij}q_{ij} = \sum_{j=1}^{n_i} Z_{i(j)}q_{i(j)}$, 
and define $e_{i(j)} = \pr(Z_{i(j)}=1\mid \mathcal{Z})$. 

By definition, we can then bound $v^2_i$ as follows:
\begin{align*}
   v^2_i & = \sum_{j=1}^{n_i} e_{i(j)} \big(q_{i(j)} - \mu_i\big)^2  
   \ge 
   e_{i(1)} \big(q_{i(1)} - \mu_i\big)^2 + e_{i(n_i)} \big(q_{i(n_i)} - \mu_i\big)^2.  
\end{align*}
Note that $\mu_i = \sum_{j=1}^{n_i} e_{i(j)} q_{i(j)}$ must satisfy
\begin{align*}
    \big( 1-e_{i(n_i)} \big) q_{i(1)} + e_{i(n_i)} q_{i(n_i)} \le 
    \mu_i \le e_{i(1)} q_{i(1)} + \big\{ 1- e_{i(1)}\big\} q_{i(n_i)}, 
\end{align*}
which implies that 
\begin{align*}
    \mu_i - q_{i(1)} & \ge \big\{ 1-e_{i(n_i)} \big\} q_{i(1)} + e_{i(n_i)} q_{i(n_i)} - q_{i(1)}
    = e_{i(n_i)} R_i, \\
    q_{i(n_i)} - \mu_i & \ge q_{i(n_i)} - e_{i(1)} q_{i(1)} - \big\{ 1- e_{i(1)}\big\} q_{i(n_i)}
    = e_{i(1)} R_i. 
\end{align*}
Thus, we can further bound $v^2_i$ by 
\begin{align*}
     v^2_i &  
   \ge 
   e_{i(1)} \big(q_{i(1)} - \mu_i\big)^2 + e_{i(n_i)} \big(q_{i(n_i)} - \mu_i\big)^2
   \ge e_{i(1)} e_{i(n_i)}^2 R_i^2 + e_{i(n_i)} e_{i(1)}^2 R_i^2 \\
   & = e_{i(1)} e_{i(n_i)} \big\{ e_{i(1)} + e_{i(n_i)} \big\} R_i^2. 
\end{align*}
Because the bias in matched set $i$ is $\Gamma_i^\sta$, the $e_{i(j)}$'s for units in matched set $i$ must be bounded below by $1/\{1+(n_i-1)\Gamma_i^\sta\} \ge 1/(n_i\Gamma_i^\sta)$. 
Consequently, 
\begin{align*}
    v^2_i 
    & \ge e_{i(1)} e_{i(n_i)} \big\{ e_{i(1)} + e_{i(n_i)} \big\} R_i^2
    \ge \frac{2}{(n_i\Gamma^\sta_i)^3} R_i^2. 
\end{align*}

We then prove the asymptotic Gaussianity of $T$. 
Let 
{\lxr $X_i = (T_i - \mu_i)/\sigma$,}
$1\le i \le I$. 
Obviously, $X_i$'s are mutually independent across all $i$, and $\sum_{i=1}^I \E(X_i^2) = 1$. To prove the asymptotic Gaussianity of $T$, it suffices to verify condition (ii) in Lemma \ref{lemma:lindfellerclt}. 
From the proof in the first part, we have 
\begin{align*}
    X_i^2 & = \frac{(T_i - \mu_i)^2}{\lxr \sigma^2}
    = 
    \frac{(T_i - \mu_i)^2}{\sum_{i=1}^I v_i^2}
    \le 
    \frac{R_i^2}{\sum_{i=1}^I 2R_i^2/\{(n_i\Gamma^\sta_i)^3\}}
    = 
    \frac{1}{2}\frac{R_i^2}{\sum_{i=1}^I R_i^2/(n_i\Gamma_i^\sta)^3}. 
\end{align*}
Under Condition \ref{cond:clt}, $\max_{1\le i \le I} X_i^2 \rightarrow 0$, and thus we must have 
$\lim_{I\to\infty}\sum_{i=1}^I\E(|X_i|^2;|X_i|>\epsilon)=0$ for any $\epsilon>0$. 
From Lemma \ref{lemma:lindfellerclt}, we then have 
$$
\frac{T-\E T}{\Var^{1/2}(T)} = \frac{T-\mu}{\sigma} = \sum_{i=1}^I X_i \converged \mathcal{N}(0,1). 
$$
Therefore, Lemma \ref{lemma:match_clt} holds. 
\end{proof}

\begin{lemma}\label{lemma:set_vector_gamma}
Assume Fisher's null $H_0$ in \eqref{eq:H0} holds. 
If $\bs{\Gamma}^\sta\vecle \bs{\Gamma} \in [1, \infty]^I$ 
and Condition \ref{cond:sen_Gamma_vec_asymp} holds, 
then for any $c \ge 0$, 
$\mu(\bs{\Gamma}) + c \sigma(\bs{\Gamma})$ is greater than or equal to $\mu+c\sigma$ when $I$ is sufficiently large.
\end{lemma}

\begin{proof}[Proof of Lemma \ref{lemma:set_vector_gamma}]
By the definition of $\mathcal{A_{\bf\Gamma}}$ and the construction of $\mu_i(\Gamma_i)$'s and $v_i^2(\Gamma_i)$'s,  
$\mu_i<\mu_i(\overline\Gamma_i)$ for all $i\in \mathcal{A}_{\bs\Gamma}$, and 
\begin{align*}
    \mu(\bs{\Gamma}) - \mu 
    = 
    \sum_{i=1}^I\big\{\mu_i( \Gamma_i)-\mu_i \big\}
    \ge
    \sum_{i\in \mathcal{A}_{\bs\Gamma}}\big\{\mu_i( \Gamma_i)-\mu_i\big\}=|\mathcal{A}_{\bs\Gamma}| \cdot \Delta_{\mu}(\bs{\Gamma}).
\end{align*}
Moreover, 
\begin{align*}
\sigma - \sigma(\bs{\Gamma})
& = 
\sqrt{\sum_{i=1}^Iv_i^2}-\sqrt{\sum_{i=1}^Iv_i^2(\Gamma_i)}
= 
\frac{\sum_{i=1}^Iv_i^2 - \sum_{i=1}^Iv_i^2(\Gamma_i)}{\sqrt{\sum_{i=1}^Iv_i^2} + \sqrt{\sum_{i=1}^Iv_i^2(\Gamma_i)}}
\le 
\frac{ \sum_{i\in \mathcal{A}_{\bs\Gamma}} (v_i^2 - v_i^2(\Gamma_i)) }{\sqrt{\sum_{i=1}^Iv_i^2} + \sqrt{\sum_{i=1}^Iv_i^2(\Gamma_i)}}
\\
& 
= 
\frac{|\mathcal{A}_{\bs\Gamma}| \Delta_{v^2}(\bs{\Gamma})}{\sqrt{\sum_{i=1}^Iv_i^2} + \sqrt{\sum_{i=1}^Iv_i^2(\Gamma_i)}}
\le 
\frac{|\mathcal{A}_{\bs\Gamma}| \Delta_{v^2}(\bs{\Gamma})}{\sqrt{\sum_{i=1}^Iv_i^2}}. 
\end{align*}
From the proof of Lemma \ref{lemma:match_clt}, we have 
$
    \sum_{i=1}^Iv_i^2 \ge 2\sum_{i=1}^I {R_i^2}/{(n_i\Gamma^\sta_i)^3}. 
$
These then imply that 
\begin{align}\label{eq:diff_mu_sigma_gamma_vec}
\mu(\bs{\Gamma}) + c \sigma(\bs{\Gamma}) - (\mu + c \sigma)
& = 
\mu(\bs{\Gamma}) - \mu  - c\big\{ \sigma - \sigma(\bs{\Gamma}) \big\}
\ge 
|\mathcal{A}_{\bs\Gamma}| \cdot \Delta_{\mu}(\bs{\Gamma}) - c \frac{|\mathcal{A}_{\bs\Gamma}| \Delta_{v^2}(\bs{\Gamma})}{\sqrt{\sum_{i=1}^Iv_i^2}}
\nonumber
\\
& = 
    \frac{|\mathcal{A}_{\bs\Gamma}| \Delta_{v^2}(\bs{\Gamma})}{\sqrt{\sum_{i=1}^Iv_i^2}}
    \left\{ 
    \frac{\Delta_{\mu}(\bs{\Gamma}) {\sqrt{\sum_{i=1}^Iv_i^2}}}{\Delta_{v^2}(\bs{\Gamma})}
    - c
    \right\} 
    \nonumber
    \\
    & \ge 
    \frac{|\mathcal{A}_{\bs\Gamma}| \Delta_{v^2}(\bs{\Gamma})}{\sqrt{\sum_{i=1}^Iv_i^2}}
    \left\{ 
    \sqrt2\frac{\Delta_{\mu}(\bs{\Gamma}) }{\Delta_{v^2}(\bs{\Gamma})}\sqrt{\sum_{i=1}^I\frac{R_i^2}{(n_i\Gamma^\sta_i)^3}}
    - c
    \right\} . 
\end{align}
Under Condition \ref{cond:sen_Gamma_vec_asymp}, there exists $\underline{I}$ such that for all $I\ge \underline{I}$, we have
$$
\sqrt2\frac{\Delta_{\mu}(\bs{\Gamma}) }{\Delta_{v^2}(\bs{\Gamma})}\sqrt{\sum_{i=1}^I\frac{R_i^2}{(n_i\Gamma^\sta_i)^3}}
    > c .
$$
Thus, 
$\mu(\bs{\Gamma}) + c \sigma(\bs{\Gamma}) - (\mu + c \sigma) \ge 0$ for $I\ge \underline{I}$. 
Therefore, 
$\mu(\bs{\Gamma}) + c \sigma(\bs{\Gamma})$ is greater than or equal to $\mu+c\sigma$ when $I$ is sufficiently large. 
\end{proof}

\begin{proof}[\bf Proof of Theorem \ref{thm:set_vector_gamma}]
We compare the tail probability of $\{ T - \mu(\bs{\Gamma}) \}/{\sigma}(\bs{\Gamma})$ to that of a standard Gaussian random variable. 
For any $c\ge 0$, from Lemma \ref{lemma:set_vector_gamma}, 
$\mu(\bs{\Gamma}) + c \sigma(\bs{\Gamma}) \ge \mu + c \sigma$ when $I$ is sufficiently large. 
This implies that
\begin{align*}
    \limsup_{I\rightarrow \infty}\pr\left( \frac{T - \mu(\bs{\Gamma})}{\sigma(\bs{\Gamma})} \ge c \right)
    & = 
    \limsup_{I\rightarrow \infty} \pr\left( 
    T \ge \mu(\bs{\Gamma}) + c \sigma(\bs{\Gamma})
    \right)
    \le 
    \limsup_{I\rightarrow \infty} \pr\left( 
    T \ge \mu + c \sigma
    \right)
    \\
    & = 
    \limsup_{I\rightarrow \infty} \pr\left( 
    \frac{T-\mu}{\sigma} \ge c 
    \right). 
\end{align*}
From Lemma \ref{lemma:match_clt}, under Condition \ref{cond:clt}, $(T-\mu)/\sigma \converged \mathcal{N}(0,1)$. Thus,  
$
\limsup_{I\rightarrow \infty} \pr\{ (T-\mu)/\sigma \ge c \} = 1 - \Phi(c), 
$
where $\Phi(\cdot)$ is the distribution function of standard Gaussian distribution. 
Therefore, 
$
\limsup_{I \rightarrow \infty}
\pr( \{ T - \mu(\bs{\Gamma}) \}/{\sigma}(\bs{\Gamma}) \ge c ) \le 1 - \Phi(c)
$
for any $c \ge 0$. 

From the above, Theorem \ref{thm:set_vector_gamma} holds. 
\end{proof}

\begin{proof}[\bf Proof of Corollary \ref{cor:pval_set_vector_gamma}]
By definition, 
\begin{align*}
 \pr (\tilde{p}_{\bs{\Gamma}} \le \alpha) 
 & = 
 \pr\left\{
 \Phi\left( \frac{T - \mu(\bs{\Gamma})}{\sigma(\bs{\Gamma})} \right) \cdot \I\{ T\ge \mu(\bs{\Gamma}) \}  \ge 1 - \alpha
 \right\}
 \\
 & = 
 \pr\left\{
 \Phi\left( \frac{T - \mu(\bs{\Gamma})}{\sigma(\bs{\Gamma})} \right) \ge 1 - \alpha, T \ge  \mu(\bs{\Gamma})
 \right\}\\
 & = 
 \pr\left\{
 T \ge \Phi^{-1}( 1 - \alpha ) \sigma(\bs{\Gamma}) + \mu(\bs{\Gamma}), T \ge  \mu(\bs{\Gamma})
 \right\}
 = 
 \pr
 \left\{
 T \ge c \sigma(\bs{\Gamma}) + \mu(\bs{\Gamma})
 \right\}
\end{align*}
where $c = \max\{0, \Phi^{-1}( 1 - \alpha )\} \ge 0$. 
From Theorem \ref{thm:set_vector_gamma}, we then have 
\begin{align*}
\limsup_{I\rightarrow \infty}\pr( \tilde{p}_{\bs{\Gamma}} \le \alpha)
= 
\limsup_{I \rightarrow \infty}
\pr\left\{ T \ge \mu(\bs{\Gamma}) + c \sigma(\bs{\Gamma}) \right\} \le 
1 - \Phi(c) 
\le 1 - \Phi(\Phi^{-1}( 1 - \alpha )) = \alpha. 
\end{align*}
Therefore, Corollary \ref{cor:pval_set_vector_gamma} holds. 
\end{proof}

\section{Proofs of Theorem \ref{thm:set_quantile_gamma} and Corollary \ref{cor:sen_pval_set}}\label{app:set_quantile}

To prove Theorem \ref{thm:set_quantile_gamma}, we need the following lemma. 

\begin{lemma}\label{lemma:set_quantile_gamma}
Assume Fisher's null $H_0$ holds. 
If 
the bias at rank $k$
is bounded by $\Gamma_0 \in [1, \infty)$, i.e., $\Gamma_{(k)}^\sta \le \Gamma_0$, 
and Condition \ref{cond:sen_Gamma_quantile_asymp}
holds, 
then for any $c \ge 0$, 
$\mu(\Gamma_0; \mathcal{I}_{\Gamma_0; k}) + c\sigma(\Gamma_0; \mathcal{I}_{\Gamma_0; k})$ is greater than or equal to $\mu(\Gamma_0; \mathcal{I}_k^\sta) + c \sigma(\Gamma_0; \mathcal{I}_k^\sta)$ when $I$ is sufficiently large, 
where $\mu(\Gamma_0; \mathcal{I}_{\Gamma_0; k})$ and $\mu(\Gamma_0; \mathcal{I}_k^\sta)$ are defined as in \eqref{eq:mu_gamma0_I}, 
and $\sigma(\Gamma_0; \mathcal{I}_{\Gamma_0; k})$ and $\sigma(\Gamma_0; \mathcal{I}_k^\sta)$ are defined as in \eqref{eq:sigma_gamma0_I}. 
\end{lemma}

\begin{proof}[Proof of Lemma \ref{lemma:set_quantile_gamma}]
Recall the definition of $\mathcal{I}_{\Gamma_0; k}$ and $\mathcal{I}_k^\sta$. 
For any $i \in \mathcal{I}_{\Gamma_0; k} \setminus \mathcal{I}_k^\sta$ and $j \in \mathcal{I}_k^\sta \setminus \mathcal{I}_{\Gamma_0; k}$, 
we must have $\mu_i(\infty) - \mu_i(\Gamma_0) \le \mu_j(\infty) - \mu_j(\Gamma_0)$. 
Moreover, if $\mu_i(\infty) - \mu_i(\Gamma_0) = \mu_j(\infty) - \mu_j(\Gamma_0)$, then 
$v_i^2(\Gamma_0) \ge  v_j^2(\Gamma_0)$. 
By the definition of $\mathcal{B}_k(\Gamma_0)$, 
for any $(i,j)\in \mathcal{B}_k(\Gamma_0)$, we must have $\mu_i(\infty) - \mu_i(\Gamma_0) < \mu_j(\infty) - \mu_j(\Gamma_0)$.

By definition, %
when $\mathcal{I}_{\Gamma_0; k} \ne \mathcal{I}_k^\sta$, 
\begin{align*}
    & \quad \ \mu(\Gamma_0; \mathcal{I}_{\Gamma_0; k}) - \mu(\Gamma_0; \mathcal{I}_k^\sta) 
    \\
    & = 
    \sum_{j \in \mathcal{I}_k^\sta \setminus \mathcal{I}_{\Gamma_0; k}}
    \{ \mu_j(\infty) - \mu_j(\Gamma_0) \} - 
    \sum_{i \in \mathcal{I}_{\Gamma_0; k} \setminus \mathcal{I}_k^\sta}
    \{ \mu_i(\infty) - \mu_i(\Gamma_0) \}\\
    & = 
    \frac{1}{|\mathcal{I}_{\Gamma_0; k} \setminus \mathcal{I}_k^\sta|}
    \sum_{(i,j)\in (\mathcal{I}_{\Gamma_0; k} \setminus \mathcal{I}_k^\sta) \times (\mathcal{I}_k^\sta \setminus \mathcal{I}_{\Gamma_0; k})}
    \left[ \{ \mu_j(\infty) - \mu_j(\Gamma_0) \} -  \{ \mu_i(\infty) - \mu_i(\Gamma_0) \} \right]\\
    & \ge 
    \frac{1}{|\mathcal{I}_{\Gamma_0; k} \setminus \mathcal{I}_k^\sta|}
    \sum_{(i,j)\in \mathcal{B}_k(\Gamma_0)}
    \left[ \{ \mu_j(\infty) - \mu_j(\Gamma_0) \} -  \{ \mu_i(\infty) - \mu_i(\Gamma_0) \} \right]\\
    & = 
    \frac{|\mathcal{B}_k(\Gamma_0)|}{|\mathcal{I}_{\Gamma_0; k} \setminus \mathcal{I}_k^\sta|}
    \Delta_{\mu}(\Gamma_0; k),
\end{align*}
and 
\begin{align*}
    \sigma^2(\Gamma_0; \mathcal{I}_k^\sta) - \sigma^2(\Gamma_0; \mathcal{I}_{\Gamma_0; k}) 
    & = 
    \sum_{j \in \mathcal{I}_k^\sta \setminus \mathcal{I}_{\Gamma_0; k}}v^2_j(\Gamma_0) - 
    \sum_{i \in \mathcal{I}_{\Gamma_0; k} \setminus \mathcal{I}_k^\sta}v^2_i(\Gamma_0)\\
    & = 
    \frac{1}{|\mathcal{I}_{\Gamma_0; k} \setminus \mathcal{I}_k^\sta|}
    \sum_{(i,j)\in (\mathcal{I}_{\Gamma_0; k} \setminus \mathcal{I}_k^\sta) \times (\mathcal{I}_k^\sta \setminus \mathcal{I}_{\Gamma_0; k})}
    \left[v^2_j(\Gamma_0)- v^2_i(\Gamma_0) \right]\\
    & \le \frac{1}{|\mathcal{I}_{\Gamma_0; k} \setminus \mathcal{I}_k^\sta|}
    \sum_{(i,j)\in \mathcal{B}_k(\Gamma_0)}
    \left[v^2_j(\Gamma_0)- v^2_i(\Gamma_0) \right]\\
    & = 
    \frac{|\mathcal{B}_k(\Gamma_0)|}{|\mathcal{I}_{\Gamma_0; k} \setminus \mathcal{I}_k^\sta|}
    \Delta_{v^2}(\Gamma_0; k). 
\end{align*}
We can then bound the difference between $\sigma(\Gamma_0, \mathcal{I}_k^\sta)$ and $\sigma(\Gamma_0, \mathcal{I}_{\Gamma_0;k})$ by
\begin{align*}
    \sigma(\Gamma_0; \mathcal{I}_k^\sta) - \sigma(\Gamma_0; \mathcal{I}_{\Gamma_0;k}) 
    &
    = 
    \frac{1}{\sqrt{\sum_{i\in \mathcal{I}^\sta_{k}}v_i^2(\Gamma_0)} + \sqrt{\sum_{i\in \mathcal{I}_{\Gamma_0; k}}v_i^2(\Gamma_0)}} \left\{\sigma^2(\Gamma_0; \mathcal{I}_k^\sta) - \sigma^2(\Gamma_0; \mathcal{I}_{\Gamma_0,k})  \right\}\\
    &\le 
    \frac{1}{\sqrt{\sum_{i\in \mathcal{I}_{\Gamma_0; k}}v_i^2(\Gamma_0)}+\sqrt{\sum_{i\in \mathcal{I}^\sta_{k}}v_i^2(\Gamma_0)}} 
    \frac{|\mathcal{B}_k(\Gamma_0)|}{|\mathcal{I}_{\Gamma_0; k} \setminus \mathcal{I}_k^\sta|}
    \Delta_{v^2}(\Gamma_0; k)\\
    & \le 
    \frac{1}{\sqrt{\sum_{i\in \mathcal{I}^\sta_{k}}v_i^2(\Gamma_0)}} 
    \frac{|\mathcal{B}_k(\Gamma_0)|}{|\mathcal{I}_{\Gamma_0; k} \setminus \mathcal{I}_k^\sta|}
    \Delta_{v^2}(\Gamma_0; k). 
\end{align*}
Consequently, when $\mathcal{I}_{\Gamma_0; k} \ne \mathcal{I}_k^\sta$, for any $c\ge 0$, we have 
\begin{align*}
    & \quad \ \mu(\Gamma_0; \mathcal{I}_{\Gamma_0; k}) + c\sigma(\Gamma_0; \mathcal{I}_{\Gamma_0; k})   - \mu(\Gamma_0; \mathcal{I}_k^\sta) - c \sigma(\Gamma_0; \mathcal{I}_k^\sta)
    \nonumber
    \\
    & = \left\{ \mu(\Gamma_0; \mathcal{I}_{\Gamma_0; k}) - \mu(\Gamma_0; \mathcal{I}_k^\sta) \right\} 
    - c \left\{ \sigma(\Gamma_0; \mathcal{I}_k^\sta) - \sigma(\Gamma_0; \mathcal{I}_{\Gamma_0; k}) \right\}
    \nonumber
    \\
    & \ge 
    \frac{|\mathcal{B}_k(\Gamma_0)|}{|\mathcal{I}_{\Gamma_0; k} \setminus \mathcal{I}_k^\sta|}
    \Delta_{\mu}(\Gamma_0; k) - 
    c \frac{1}{\sqrt{\sum_{i\in \mathcal{I}^\sta_{k}}v_i^2(\Gamma_0)}} 
    \frac{\mathcal{B}_k(\Gamma_0)}{|\mathcal{I}_{\Gamma_0; k} \setminus \mathcal{I}_k^\sta|}
    \Delta_{v^2}(\Gamma_0; k)
    \nonumber
    \\
    & = 
    \frac{|\mathcal{B}_k(\Gamma_0)|}{|\mathcal{I}_{\Gamma_0; k} \setminus \mathcal{I}_k^\sta|}
    \frac{\Delta_{v^2}(\Gamma_0; k)}{\sqrt{\sum_{i\in \mathcal{I}^\sta_{k}}v_i^2(\Gamma_0)}}
    \left\{
    \frac{\Delta_{\mu}(\Gamma_0; k)}{\Delta_{v^2}(\Gamma_0; k)} 
    \sqrt{\sum_{i\in \mathcal{I}^\sta_{k}}v_i^2(\Gamma_0)}
    - 
    c  
    \right\}. 
\end{align*}
From the proof of Lemma \ref{lemma:match_clt}, we have 
\begin{align*}
    \sum_{i\in \mathcal{I}^\sta_{k}}v_i^2(\Gamma_0) \ge 
    2 \sum_{i\in \mathcal{I}^\sta_{k}} 
    \frac{R_i^2}{(n_i\Gamma_0)^3}.  
\end{align*}
Thus, for any $c\ge 0$, if $\mathcal{I}_{\Gamma_0; k} \ne \mathcal{I}_k^\sta$, then 
\begin{align}\label{eq:diff_mu_sigma_gamma_quantile}
    &\quad \mu(\Gamma_0; \mathcal{I}_{\Gamma_0; k}) + c\sigma(\Gamma_0; \mathcal{I}_{\Gamma_0; k})   - \mu(\Gamma_0; \mathcal{I}_k^\sta) - c \sigma(\Gamma_0; \mathcal{I}_k^\sta)
    \nonumber
    \\& 
    \ge
    \frac{|\mathcal{B}_k(\Gamma_0)|}{|\mathcal{I}_{\Gamma_0; k} \setminus \mathcal{I}_k^\sta|}
    \frac{\Delta_{v^2}(\Gamma_0; k)}{\sqrt{\sum_{i\in \mathcal{I}^\sta_{k}}v_i^2(\Gamma_0)}}
    \left\{\sqrt{2}\frac{\Delta_{\mu}(\Gamma_0; k)}{\Delta_{v^2}(\Gamma_0; k)}\sqrt{\sum_{i\in \mathcal{I}_k^\sta}\frac{R_i^2}{(n_i\Gamma_0)^3}} - c \right\}. 
\end{align}
From Condition \ref{cond:sen_Gamma_quantile_asymp}, 
$\mu(\Gamma_0;\mathcal{I}_{\Gamma_0;k}) + c\sigma(\Gamma_0; \mathcal{I}_{\Gamma_0;k})$ is greater than or equal to $\mu(\Gamma_0; \mathcal{I}_k^\sta) + c \sigma(\Gamma_0; \mathcal{I}_k^\sta)$ when $I$ is sufficiently large. 
Therefore, Lemma \ref{lemma:set_quantile_gamma} holds. 
\end{proof}

\begin{proof}[\bf Proof of Theorem \ref{thm:set_quantile_gamma}]
From Condition \ref{cond:sen_Gamma_quantile_asymp} and Lemma \ref{lemma:set_quantile_gamma}, $\mu(\Gamma_0; \mathcal{I}_{\Gamma_0; k}) + c\sigma(\Gamma_0; \mathcal{I}_{\Gamma_0; k})$ is greater than or equal to $\mu(\Gamma_0; \mathcal{I}_k^\sta) + c \sigma(\Gamma_0; \mathcal{I}_k^\sta)$ when $I$ is sufficiently large.
From Condition \ref{cond:quantile_set_true_Ik} and Lemma \ref{lemma:set_vector_gamma}, 
$\mu(\Gamma_0;\mathcal{I}_k^\sta) + c \sigma(\Gamma_0; \mathcal{I}_k^\sta)$ is greater than or equal to $\mu + c \sigma$ when $I$ is sufficiently large. 
These imply that $\mu(\Gamma_0; \mathcal{I}_{\Gamma_0; k}) + c\sigma(\Gamma_0; \mathcal{I}_{\Gamma_0; k})$ is greater than or equal to $\mu + c \sigma$ when $I$ is sufficiently large. 
Consequently, for any $c\ge 0$, when $I$ is sufficiently large,  
\begin{align*}
    \pr \left\{ T \ge \mu(\Gamma_0; k) + c \sigma(\Gamma_0; k) \right\}
    & \le 
    \pr \left( T \ge \mu + c \sigma \right)
    = \pr\left( \frac{T-\mu}{\sigma} \ge c \right).
\end{align*}
From Condition \ref{cond:clt} and Lemma \ref{lemma:match_clt}, $(T-\mu)/\sigma \converged \mathcal{N}(0,1)$ as $I\rightarrow \infty$. 
Therefore, 
\begin{align*}
    \limsup_{I \rightarrow \infty} \pr \left\{ T \ge \mu(\Gamma_0; k) + c \sigma(\Gamma_0; k) \right\}
    & \le 
    \limsup_{I \rightarrow \infty} \pr\left( \frac{T-\mu}{\sigma} \ge c \right) = 
    1 - \Phi(c), 
\end{align*}
i.e., 
Theorem \ref{thm:set_quantile_gamma} holds.
\end{proof}

\begin{proof}[\bf Proof of Corollary \ref{cor:sen_pval_set}]
By definition, 
\begin{align*}
\pr \left( \tilde{p}_{\Gamma_0; k} \le \alpha\right)
& = 
\pr\left\{ 
\Phi\left(
\frac{T - \mu(\Gamma_0; k)}{\sigma(\Gamma_0; k)}
\right) \cdot 
\I\{T \ge \mu(\Gamma_0; k)\} \ge 1-\alpha
\right\}\\
& = 
\pr\left\{
T \ge \mu(\Gamma_0; k) + \Phi^{-1}(1-\alpha)\sigma(\Gamma_0; k), 
T \ge \mu(\Gamma_0; k)
\right\}\\
& = \pr\left\{
T \ge \mu(\Gamma_0; k) + c\sigma(\Gamma_0; k)
\right\},
\end{align*}
where $c = \max\{0, \Phi^{-1}(1-\alpha)\} \ge 0$. 
From Theorem \ref{thm:set_quantile_gamma}, as $I\rightarrow \infty$, we have
\begin{align*}
 \limsup_{I\rightarrow \infty}\pr \left( \tilde{p}_{\Gamma_0; k} \le \alpha\right)
 & =  \limsup_{I\rightarrow \infty} \pr\left\{
T \ge \mu(\Gamma_0; k) + c\sigma(\Gamma_0; k)
\right\}
\le 1 - \Phi(c)
\le 1 - \Phi(\Phi^{-1}(1-\alpha)) 
\\
& = \alpha. 
\end{align*}
Therefore, Corollary \ref{cor:sen_pval_set} holds. 
\end{proof}

\section{Proofs of Corollary \ref{cor:set_ci_gamma} and Theorem \ref{thm:set_joint_ci}}\label{app:set_ci}

Because Corollary \ref{cor:set_ci_gamma} is a direct consequence of Theorem \ref{thm:set_joint_ci}, it suffices to prove Theorem \ref{thm:set_joint_ci}. 
To prove Theorem \ref{thm:set_joint_ci}, we need the following lemma. 

\begin{lemma}\label{lemma:set_joint_ci}
Assume Fisher's null $H_0$ holds. 
If Conditions \ref{cond:set_simultaneous}(ii) and (iii) holds, then for any $c\ge 0$, 
$$\inf_{\underline{k}_I \le k\le I, \Gamma_0 \in [\Gamma_{(k)}^\sta, \overline{\Gamma}_I]} \left\{ \mu(\Gamma_0;k) + c \sigma(\Gamma_0; k) \right\}
\ge \mu + c\sigma
$$
when $I$ is sufficiently large.
\end{lemma}
\begin{proof}[Proof of Lemma \ref{lemma:set_joint_ci}]
For any $c\ge 0$, 
from Conditions \ref{cond:set_simultaneous}(ii) and (iii), 
there exists $\underline{I}$ such that when $I\ge \underline{I}$, 
\begin{align*}
\inf_{\underline{k}_I \le k\le I, \Gamma_0 \in [\Gamma_{(k)}^\sta, \overline{\Gamma}_I]}
\frac{\Delta_{\mu}(\Gamma_0 \cdot \bs{1}_{\mathcal{I}_k^\sta} + \infty \cdot \bs{1}_{\mathcal{I}_k^\sta}^\complement)}{\Delta_{v^2}(\Gamma_0 \cdot \bs{1}_{\mathcal{I}_k^\sta} + \infty \cdot \bs{1}_{\mathcal{I}_k^\sta}^\complement)} \sqrt{\sum_{i=1}^I R_i^2/(n_i\Gamma_i^\sta)^3}  \ge c/\sqrt{2}, 
\\
\inf_{\underline{k}_I \le k \le I, \Gamma_0 \in [\Gamma_{(k)}^\sta, \overline{\Gamma}_I]}
\frac{\Delta_{\mu}(\Gamma_0; k)}{\Delta_{v^2}(\Gamma_0; k)}
\sqrt{\sum_{i\in \mathcal{I}_k^\sta} R_i^2/(n_i\Gamma_0)^3} \ge c/\sqrt{2}. 
\end{align*}
Note that, by definition, 
$
\mu(\Gamma_0; \mathcal{I}_k^\sta) = \mu(\Gamma_0 \cdot \bs{1}_{\mathcal{I}_k^\sta} + \infty \cdot \bs{1}_{\mathcal{I}_k^\sta}^\complement)
$
and 
$
\sigma^2(\Gamma_0; \mathcal{I}_k^\sta) = \sigma^2(\Gamma_0 \cdot \bs{1}_{\mathcal{I}_k^\sta} + \infty \cdot \bs{1}_{\mathcal{I}_k^\sta}^\complement). 
$
From \eqref{eq:diff_mu_sigma_gamma_vec} and \eqref{eq:diff_mu_sigma_gamma_quantile} in the proofs of Lemmas \ref{lemma:set_vector_gamma} and \ref{lemma:set_quantile_gamma},  
we can know that, for any $I\ge \underline{I}$, and any $\underline{k}_I \le k\le I$ and $\Gamma_0 \in [\Gamma_{(k)}^\sta, \overline{\Gamma}_I]$, 
\begin{align*}
& \quad \ \mu(\Gamma_0; \mathcal{I}_{\Gamma_0; k}) + c\sigma(\Gamma_0; \mathcal{I}_{\Gamma_0; k}) - \mu - c \sigma
\\
& \ge 
\left\{ \mu(\Gamma_0; \mathcal{I}_{\Gamma_0; k}) + c\sigma(\Gamma_0; \mathcal{I}_{\Gamma_0; k})   - \mu(\Gamma_0; \mathcal{I}_k^\sta) - c \sigma(\Gamma_0; \mathcal{I}_k^\sta) \right\}
+ 
\left\{ \mu(\Gamma_0; \mathcal{I}_k^\sta) + c \sigma(\Gamma_0; \mathcal{I}_k^\sta)  - \mu - c \sigma \right\}
\\
& \ge 0. 
\end{align*}
This then implies that, for any $I \ge \underline{I}$, 
\begin{align*}
    \inf_{\underline{k}_I \le k\le I, \Gamma_0 \in [\Gamma_{(k)}^\sta, \overline{\Gamma}_I]} \left\{ \mu(\Gamma_0;k) + c \sigma(\Gamma_0; k) \right\}
    \ge \mu + c\sigma.
\end{align*}
Therefore, Lemma \ref{lemma:set_joint_ci} holds.
\end{proof}

\begin{proof}[\bf Proof of Theorem \ref{thm:set_joint_ci}]
By the same logic as the proof of Theorem \ref{thm:pair_joint_ci}, we can show that $\Gamma_{(k)}^\sta \in 
\{\Gamma_0\ge 1: \vardbtilde{p}_{\Gamma_0;k} > \alpha\}
\Longleftrightarrow 
\bs{\Gamma}^\sta \in 
\bigcap_{\Gamma_0: \vardbtilde{p}_{\Gamma_0;k} \le \alpha} \mathcal{S}_{\Gamma_0;k}^\complement$
and 
$
I^\sta(\Gamma_0) \in \{I-k: \vardbtilde{p}_{\Gamma_0;k} > \alpha, 0 \le k \le I\}
\Longleftrightarrow 
\bs{\Gamma}^\sta \in 
\bigcap_{k \ge \underline{k}_I: \tilde{p}_{\Gamma_0;k} \le \alpha} \mathcal{S}_{\Gamma_0;k}^\complement. 
$
Below we prove the asymptotic validity of the confidence set in \eqref{eq:set_joint_set}. 

Note that 
\begin{align*}
\pr\Big( \bs{\Gamma}^\sta \in \bigcap_{\Gamma_0, k: \vardbtilde{p}_{\Gamma_0;k} \le \alpha} \mathcal{S}_{\Gamma_0;k}^\complement  \Big) = 
1 - 
\pr\Big( \bs{\Gamma}^\sta \in \bigcup_{\Gamma_0, k: \vardbtilde{p}_{\Gamma_0;k} \le \alpha} \mathcal{S}_{\Gamma_0;k} \Big). 
\end{align*}
If 
$\bs{\Gamma}^\sta \in \bigcup_{\Gamma_0, k: \vardbtilde{p}_{\Gamma_0;k} \le \alpha} \mathcal{S}_{\Gamma_0;k}$, 
then there must exists $(\Gamma_0', k')$ such that 
$\Gamma_0' \le \overline{\Gamma}_I$, 
$k' \ge \underline{k}_I$, 
$\vardbtilde{p}_{\Gamma_0';k'} = \tilde{p}_{\Gamma_0';k'} \le \alpha$, and 
$\bs{\Gamma}^\sta \in \mathcal{S}_{\Gamma_0';k'}$, 
where the latter implies that $\Gamma_{(k')}^\sta \le \Gamma_0'$. 
By definition, 
\begin{align*}
 \tilde{p}_{\Gamma_0'; k'} \le \alpha
& \Longleftrightarrow
\Phi\left(
\frac{T - \mu(\Gamma_0'; k')}{\sigma(\Gamma_0'; k')}
\right) \cdot 
\I\{T \ge \mu(\Gamma_0'; k')\} \ge 1-\alpha
\\
& \Longleftrightarrow
T \ge \mu(\Gamma_0'; k') + \Phi^{-1}(1-\alpha)\sigma(\Gamma_0'; k') \text{ and } 
T \ge \mu(\Gamma_0'; k')
\\
& \Longleftrightarrow 
T \ge \mu(\Gamma_0'; k') + c\sigma(\Gamma_0'; k'), 
\end{align*}
where $c = \max\{0, \Phi^{-1}(1-\alpha)\} \ge 0$. 
Note that $k'\ge \underline{k}_I$ and $\Gamma_{(k')}^\sta \le \Gamma_0' \le \overline{\Gamma}_I$. 
From Conditions \ref{cond:set_simultaneous}(ii) and (iii) and Lemma \ref{lemma:set_joint_ci}, 
there exists $\underline{I}$ such that when $I\ge \underline{I}$, 
\begin{align*}
\mu(\Gamma_0'; k') + c\sigma(\Gamma_0'; k')
\ge 
\inf_{\underline{k}_I \le k\le I, \Gamma_0 \in [\Gamma_{(k)}^\sta, \overline{\Gamma}_I]} \left\{ \mu(\Gamma_0;k) + c \sigma(\Gamma_0; k) \right\}
\ge \mu + c\sigma. 
\end{align*}
Therefore, 
when $I\ge \underline{I}$, 
if $\bs{\Gamma}^\sta \in \bigcup_{\Gamma_0, k: \vardbtilde{p}_{\Gamma_0;k} \le \alpha} \mathcal{S}_{\Gamma_0;k}$, 
then we must have $T \ge \mu + c\sigma$. 
This implies that, {\lxr for $I\ge \underline{I}$,} 
\begin{align*}
\pr\Big( \bs{\Gamma}^\sta \in \bigcap_{\Gamma_0, k \ge \underline{k}_I: \vardbtilde{p}_{\Gamma_0;k} \le \alpha} \mathcal{S}_{\Gamma_0;k}^\complement  \Big) = 
1 - 
\pr\Big( \bs{\Gamma}^\sta \in \bigcup_{\Gamma_0, k \ge \underline{k}_I: \vardbtilde{p}_{\Gamma_0;k} \le \alpha} \mathcal{S}_{\Gamma_0;k} \Big)
\ge 
1 -  \pr(T \ge \mu + c\sigma). 
\end{align*}
From Lemma \ref{lemma:match_clt}, we have
\begin{align*}
 \liminf_{I\rightarrow \infty}\pr\Big( \bs{\Gamma}^\sta \in \bigcap_{\Gamma_0, k \ge \underline{k}_I: \vardbtilde{p}_{\Gamma_0;k} \le \alpha} \mathcal{S}_{\Gamma_0;k}^\complement  \Big)
 & \ge 
1 -  \limsup_{I\rightarrow\infty}\pr\left( \frac{T-\mu}{\sigma} \ge c\right) 
= \Phi(c)
\ge  \Phi(\Phi^{-1}(1-\alpha)) \\
& \ge 1 - \alpha. 
\end{align*}
Therefore, Theorem \ref{thm:set_joint_ci} holds. 
\end{proof}

{\add 

\section{Proofs of Theorems \ref{thm:bounded_null_pair} and \ref{thm:bound_joint_ci}}\label{sec:bound_null_pair_proof}

We will prove Theorem \ref{thm:bounded_null_pair} in the case of effect increasing and differential increasing statistics separately as follows. 

\begin{proof}[\bf Proof of Theorem \ref{thm:bounded_null_pair}(a)]
We 
consider the case where 
the test statistic $t(\cdot, \cdot)$ in \eqref{eq:test_stat} is effect increasing.
For any treatment assignment $\bs a\in\mathcal Z$, 
let $\bs{Y}(\bs{a}) \equiv \bs{a} \circ \bs{Y}(1) + (\bs{1}-\bs{a}) \circ \bs{Y}(0)$ denote the corresponding observed outcome under assignment $\bs{a}$. 
Let $\tilde{G}(c) \equiv \sum_{\bs a\in\mathcal Z}\pr_0(\bs A=\bs a)\I\{t(\bs a,\bs Y(\bs a))\ge c\}$ denote the tail probability of the true randomization distribution of $t(\bs Z,\bs Y)$, where we use $\pr_0(\cdot)$ to denote the true probability measure for the observed treatment assignment $\bs{Z}$. 
Consequently, 
{\lxr by Lemma \ref{lemma:survival_dominant_by_uniform},}
$\tilde{G}(t(\bs Z,\bs Y))$ must be stochastically larger than or equal to $\Unif(0,1)$.

Below we further introduce several notations. 
For any vector $\bs{y}\in \mathbb{R}^{N}$, $\bs{\Gamma}\in [1, \infty]^I$ and $\bs{u}\in [0,1]^N$, 
let
$G_{\bs\Gamma,\bs u,\bs y}(c)=\sum_{\bs a\in\mathcal Z}\pr_{\bs\Gamma,\bs u}(\bs A=\bs a)\I\{t(\bs a,\bs y)\ge c\}$ denote the tail probability of the distribution of $t(\bs{A}, \bs{y})$ with $\bs{y}$ kept fixed and $\bs{A}$ following the sensitivity model $\mathcal{M}(\bs\Gamma,\bs u)$, 
where we use $\pr_{\bs\Gamma,\bs u}(\cdot)$ to denote the probability measure of the treatment assignment under the sensitivity model $\mathcal{M}(\bs\Gamma,\bs u)$. 
Let $\tilde{G}_{\bs\Gamma,\bs u, \bs{Y}(1), \bs{Y}(0)}(c) = \sum_{\bs a\in\mathcal Z}\pr_{\bs\Gamma,\bs u}(\bs A=\bs a)\I\{t(\bs a,\bs Y(\bs{a}))\ge c\}$ denote the tail probability of the distribution of $t(\bs{A}, \bs{Y}(\bs{A}))$ with $\bs{A}$ following the sensitivity model $\mathcal{M}(\bs\Gamma,\bs u)$.

From \citet[][\lxr Proof of Theorem 1(a)]{li2020quantile}, 
under the bounded null hypothesis $\overline{H}_{0}$ and with the effect increasing statistic $t(\cdot, \cdot)$, we must have 
\begin{align}\label{eq:prop_effect}
    t(\bs{a}, \bs{Y}) \ge t(\bs{a}, \bs{Y}(\bs{a})) \quad \text{ for any }\bs a\in\mathcal Z. 
\end{align}
For completeness, we give a proof of \eqref{eq:prop_effect} below. 
By some algebra, we can write $\bs{Y}$ as 
\begin{align*}
    \bs{Y} 
    & = \bs{Y}(\bs{a}) + \bs{a} \circ \{\bs{Y} - \bs{Y}(1)\} + (\bs{1}-\bs{a}) \circ \{\bs{Y} - \bs{Y}(0)\}
    \\
    & = \bs{Y}(\bs{a}) + \bs{a} \circ (\bs{1}-\bs{Z})\circ \{\bs{Y}(0) - \bs{Y}(1)\} + (\bs{1}-\bs{a}) \circ \bs{Z} \circ \{\bs{Y}(1) - \bs{Y}(0)\}. 
\end{align*}
Note that, under the bounded null $\overline{H}_0$, 
$\bs{Y}(1) - \bs{Y}(0) \vecle \bs{0}$. 
Therefore, the inequality \eqref{eq:prop_effect} follows immediately from the property of the effect increasing statistic in Definition \ref{def:effect_increase}.

From \eqref{eq:prop_effect}, we then have 
\begin{align}\label{eq:prop_effect2}
    G_{\bs\Gamma,\bs u,\bs Y}(c)&=\sum_{\bs a\in\mathcal Z}\pr_{\bs\Gamma,\bs u}(\bs A=\bs a)\I\{t(\bs a,\bs Y)\ge c\}\nonumber\\
    &\ge\sum_{\bs a\in\mathcal Z}\pr_{\bs\Gamma,\bs u}(\bs A=\bs a)\I\{t(\bs a,\bs Y(\bs a))\ge c\} = \tilde{G}_{\bs\Gamma,\bs u, \bs Y(1), \bs Y(0)}(c).
\end{align}
By definition, we then have 
\begin{align}\label{eq:valid_p_pair_effect}
    \bar{p}_{\Gamma_0;k}&=\sup_{(\bs\Gamma,\bs u)\in[1,\infty]^I\times[0,1]^N:\Gamma_{(k)}\le\Gamma_0}G_{\bs\Gamma,\bs u,\bs Y}(t(\bs Z,\bs Y))\nonumber\\
    &\ge\sup_{(\bs\Gamma,\bs u)\in[1,\infty]^I\times[0,1]^N:\Gamma_{(k)}\le\Gamma_0}\tilde{G}_{\bs\Gamma,\bs u, \bs Y(1), \bs Y(0)}(t(\bs Z,\bs Y))
    \ge 
    \tilde{G}(t(\bs Z,\bs Y)),
\end{align}
where the second last inequality follows from \eqref{eq:prop_effect2} and the last inequality follows from the fact that the true hidden bias at rank $k$ is bounded by $\Gamma_0$. 
Because $\tilde{G}(t(\bs Z,\bs Y))$ is stochastically larger than or equal to $\Unif(0,1)$, $\bar{p}_{\Gamma_0;k}$ must also be stochastically larger than or equal to $\Unif(0,1)$, i.e., it is a valid $p$-value. 

From the above, Theorem \ref{thm:bounded_null_pair} holds when the test statistic is effect increasing.  
\end{proof}

\begin{proof}[\bf Proof of Theorem \ref{thm:bounded_null_pair}(b)]
We consider the case where the test statistic $t(\cdot, \cdot)$ in \eqref{eq:test_stat} is differential increasing. 
Define $\pr_0(\cdot)$, $\pr_{\bs{\Gamma}, \bs{u}}(\cdot)$ and $G_{\bs\Gamma,\bs u,\bs y}(\cdot)$ the same as in the proof of Theorem \ref{thm:bounded_null_pair}(a). 
Let $G(c) \equiv \sum_{\bs a\in\mathcal Z}\pr_{0}(\bs A=\bs a)\I\{t(\bs a,\bs Y(0))\ge c\}$ denote the tail probability of the true randomization distribution of $t(\bs Z,\bs Y(0))$. 
Consequently, 
{\lxr by Lemma \ref{lemma:survival_dominant_by_uniform},}
$G(t(\bs Z,\bs Y(0))$ must be stochastically larger than or equal to $\Unif(0,1)$.

From \citet[][\lxr Proof of Theorem 1(a)]{li2020quantile}, 
under the bounded null hypothesis $\overline{H}_{0}$ and with the differential increasing statistic $t(\cdot, \cdot)$, we must have 
\begin{align}\label{eq:diff_increase}
    t(\bs a,\bs Y)-t(\bs Z,\bs Y)\ge t(\bs a,\bs Y(0))-t(\bs Z,\bs Y(0)) \quad \text{ for any }\bs a\in\mathcal Z.
\end{align}
For completeness, we give a proof of \eqref{eq:diff_increase} below. 
By some algebra,
$\bs{Y}(0) - \bs{Y} = \bs{Z}\circ \{\bs{Y}(0) - \bs{Y}(1)\}$. 
Note that $\bs{Y}(0) - \bs{Y}(1) \vecge \bs{0}$ under the bounded null $\overline{H}_0$. By the property of the differential increasing statistic in Definition \ref{def:dif_increase}, we have, for any $\bs{a}\in \mathcal{Z}$,  
\begin{align*}
    t(\bs{Z}, \bs{Y}(0)) - t(\bs{Z}, \bs{Y})
    & = t(\bs{Z}, \bs{Y} + \bs{Z}\circ \{\bs{Y}(0) - \bs{Y}(1)\}) - t(\bs{Z}, \bs{Y})\\
    & \ge t(\bs{a}, \bs{Y} + \bs{Z}\circ \{\bs{Y}(0) - \bs{Y}(1)\}) - t(\bs{a}, \bs{Y}) 
    = t(\bs{a}, \bs{Y}(0)) - t(\bs{a}, \bs{Y}).
\end{align*}
This immediately implies \eqref{eq:diff_increase}.

From \eqref{eq:diff_increase}, we then have 
\begin{align}\label{eq:tail_prob_ineq_diff}
    G_{\bs\Gamma,\bs u, \bs Y}(t(\bs Z,\bs Y))&=\sum_{\bs a\in\mathcal Z}\pr_{\bs\Gamma,\bs u}(\bs A=\bs a)\I\{t(\bs a,\bs Y)\ge t(\bs Z,\bs Y)\}\nonumber\\
    &\ge\sum_{\bs a\in\mathcal Z}\pr_{\bs\Gamma,\bs u}(\bs A=\bs a)\I\{t(\bs a,\bs Y(0))\ge t(\bs Z,\bs Y(0))\}
    =G_{\bs\Gamma,\bs u,\bs Y(0)}(t(\bs Z,\bs Y(0))).
\end{align}
By definition, we then have 
\begin{align}\label{eq:valide_p}
    \bar{p}_{\Gamma_0;k}&=\sup_{(\bs\Gamma,\bs u) \in [1,\infty]^I\times [0,1]^N:\Gamma_{(k)}\le\Gamma_0} G_{\bs\Gamma,\bs u,\bs Y}(t(\bs Z,\bs Y))\nonumber\\
    &\ge \sup_{(\bs\Gamma,\bs u) \in [1,\infty]^I\times [0,1]^N:\Gamma_{(k)}\le\Gamma_0} G_{\bs\Gamma,\bs u,\bs Y(0)}(t(\bs Z,\bs Y(0)))\ge G(t(\bs Z,\bs Y(0))), 
\end{align}
where the second last inequality follows from \eqref{eq:tail_prob_ineq_diff} and the last inequality follows from the fact that the true hidden bias at rank $k$ is bounded by $\Gamma_0$. 
Because $G(t(\bs Z,\bs Y(0)))$ is stochastically larger than or equal to $\Unif(0,1)$, 
$\bar{p}_{\Gamma_0;k}$ must also be stochastically larger than or equal to $\Unif(0,1)$, i.e., it is a valid $p$-value. 

From the above, Theorem \ref{thm:bounded_null_pair} holds when the test statistic is differential increasing. 
\end{proof}

\begin{proof}[\bf Proof of Theorem \ref{thm:bound_joint_ci}]
Assume that the bounded null hypothesis $\overline{H}_0$ in \eqref{eq:bound_null} holds. 
The validity of confidence sets as in Corollary \ref{cor:pair_ci_gamma3} for each quantile of hidden biases (as well as the number of matched pairs with hidden biases greater than any threshold) follows immediately from Theorem \ref{thm:bounded_null_pair}, 
and the equivalent forms of the confidence sets as in Corollary \ref{cor:pair_ci_gamma3} follow by the same logic as there. 
Below we prove the simultaneous validity of these confidence sets.
From the discussion in Section \ref{sec:simultaneous_pair}, it suffices to prove that, for any $\alpha\in (0,1)$, the confidence set in \eqref{eq:pair_joint_set} covers the true hidden bias vector $\bs{\Gamma}^\sta$ with probability at least $1-\alpha$.

Suppose that $\bs{\Gamma}^\sta \in \bigcup_{\Gamma_0, k: \overline{p}_{\Gamma_0;k} \le \alpha} \mathcal{S}_{\Gamma_0;k}$. Then there must exist $(\Gamma_0, k)$ such that $\overline{p}_{\Gamma_0;k} \le \alpha$ and $\bs{\Gamma}^\sta \in \mathcal{S}_{\Gamma_0;k}$, 
which further implies that $\Gamma_{(k)}^\sta \le \Gamma_0$. 
From the proof of Theorem \ref{thm:bounded_null_pair} (in particular \eqref{eq:valid_p_pair_effect} and \eqref{eq:valide_p}), we must have 
\begin{align*}
    \alpha \ge \overline{p}_{\Gamma_0;k} 
    \ge 
    \begin{cases}
    \tilde{G}(t(\bs{Z}, \bs{Y})), & \quad \text{if the test statistic } t(\cdot, \cdot) \text{ is effect increasing}, \\
    G(t(\bs{Z}, \bs{Y}(0))), & \quad \text{if the test statistic } t(\cdot, \cdot) \text{ is differential increasing},
    \end{cases}
\end{align*}
where {\lxr $\tilde{G}(\cdot)$} and $G(\cdot)$ denote the tail probabilities of $t(\bs{Z}, \bs{Y})$ and $t(\bs{Z}, \bs{Y}(0))$, respectively, defined as in the proof of Theorem \ref{thm:bounded_null_pair}.

From the discussion before, we then have 
\begin{align*}
    \pr\Big( \bs{\Gamma}^\sta \in \bigcap_{\Gamma_0, k: \overline{p}_{\Gamma_0;k} \le \alpha} \mathcal{S}_{\Gamma_0;k}^\complement  \Big) 
& = 
1 - \pr\Big( \bs{\Gamma}^\sta \in \bigcup_{\Gamma_0, k: \overline{p}_{\Gamma_0;k} \le \alpha} \mathcal{S}_{\Gamma_0;k} \Big)
\\
& \ge 
\begin{cases}
1 - \pr\{\tilde{G}(t(\bs{Z}, \bs{Y}))\le \alpha\}, & \quad \text{if } t(\cdot, \cdot) \text{ is effect increasing},
\\
1 - \pr\{G(t(\bs{Z}, \bs{Y}(0)))\le \alpha\}, & \quad \text{if } t(\cdot, \cdot) \text{ is differential increasing},
\end{cases}
\\
& \ge 1 -\alpha,
\end{align*}
where the last inequality follows immediately from Lemma \ref{lemma:survival_dominant_by_uniform}.

From the above, Theorem \ref{thm:bound_joint_ci} holds. 
\end{proof}

\section{Proofs of Theorems \ref{thm:bounded_null_sets_effect_diff_increasing} and \ref{thm:bound_joint_ci_set}}\label{sec:bound_null_set_proof}

To prove Theorem \ref{thm:bounded_null_sets_effect_diff_increasing}, we first introduce the following condition.  

\begin{condition}\label{cond:134as}
With true control potential outcome vector $\bs{Y}(0)$ replaced by the imputed control potential outcome vector $\bs{Y}$ based on the sharp null ${\lxr H_{0}}$ (or equivalently pretending $\bs{Y}$ as the true control potential outcome vector), 
Conditions \ref{cond:clt}, \ref{cond:quantile_set_true_Ik} and \ref{cond:sen_Gamma_quantile_asymp} hold almost surely, in the sense that the {\lxr convergences} in these three conditions hold almost surely as $I\rightarrow \infty$. 
\end{condition}

We then need the following two lemmas. 

\begin{lemma}\label{lemma:bounded_null_sets_effect_increasing}
Consider a general matched study with set sizes $n_i \ge 2$ for all $i$. 
If the bounded null hypothesis $\overline{H}_0$ in \eqref{eq:bound_null} holds, 
the bias at rank $k$ is bounded by $\Gamma_0$, i.e., $\Gamma^\sta_{(k)}\le \Gamma_0$, 
Condition \ref{cond:134as} holds, 
and the test statistic $t(\cdot, \cdot)$ in \eqref{eq:test_stat} is effect increasing, 
then $\tilde p_{\Gamma_0;k}$ in Corollary \ref{cor:sen_pval_set} is still an asymptotic valid $p$-value, in the sense that $\limsup_{I\to\infty}\pr( \tilde{p}_{\Gamma_0;k} \le \alpha ) \le \alpha$ for $\alpha\in (0,1)$.
\end{lemma}

\begin{proof}[Proof of Lemma \ref{lemma:bounded_null_sets_effect_increasing}]
Define $\mu,\sigma,\mu(\Gamma_0;k),\sigma(\Gamma_0;k)$ analogously as in Section \ref{sec:sen_quan_set} but with the true control potential outcome vector $\bs{Y}(0)$ replaced by the imputed control potential outcome vector $\bs{Y}$ based on the sharp null ${\lxr H_{0}}$ (or equivalently by pretending $\bs{Y}$ as the true control potential outcome vector).
For notational clarity, we will write them as $\mu_I,\sigma_I,\mu_I(\Gamma_0;k),\sigma_I(\Gamma_0;k)$, using the subscript $I$ to emphasize that they are for the matched data with $I$ matched sets. 
Define $G(\cdot)$ and $G_{\bs{\Gamma}, \bs{u}, \bs{y}}$ and $\tilde G_{\bs\Gamma,\bs u,\bs Y(1),\bs Y(0)}$ the same as in the proof of Theorem \ref{thm:bounded_null_pair}(a), and let $\bs{u}^\star$ denote the true hidden confounding, so that the actual treatment assignment $\bs{Z}$ follows $\mathcal{M}(\bs{\Gamma}^\star, \bs{u}^\star)$.  
Let $F(x) = G_{\bs\Gamma^\star,\bs u^\star, \bs Y}(-x)$ be a distribution function, and $F^{-1}$ be the corresponding quantile function.  
{\lxr Recall}
that $T = t(\bs Z,\bs Y)$. 

First, we prove that $G_{\bs\Gamma^\sta,\bs u^\sta,\bs Y}(T)$ is stochastically larger than or equal to $\Unif(0,1)$. 
Because $t(\cdot, \cdot)$ is effect increasing, 
by the same logic as \eqref{eq:prop_effect2}, we can know that  
$G_{\bs\Gamma^\sta,\bs u^\sta,\bs Y}(T) \ge \tilde{G}_{\bs\Gamma^\sta,\bs u^\sta, \bs Y(1), \bs Y(0)}(T)$. 
Note that $\tilde{G}_{\bs\Gamma^\sta,\bs u^\sta, \bs Y(1), \bs Y(0)}(\cdot)$ is the true tail probability of  $T=t(\bs Z,\bs Y)$. 
From Lemma \ref{lemma:survival_dominant_by_uniform}, $\tilde{G}_{\bs\Gamma^\sta,\bs u^\sta, \bs Y(1), \bs Y(0)}(T)$, as well as  $G_{\bs\Gamma^\sta,\bs u^\sta,\bs Y}(T)$, must be stochastically larger than or equal to $\Unif(0,1)$. 

Second, we prove that, for any $\alpha\in (0,1)$ and $\eta>0$, 
\begin{align}\label{eq:bound_T_u_sigma}
    \pr\left\{
    \frac{T - \mu_I}{\sigma_I} > \Phi^{-1}(1-\alpha) + \eta
    \right\}
    \le 
    \alpha + \pr\left\{
    \left| 
    \frac{F^{-1}(\alpha) + \mu_I}{\sigma_I} - \Phi^{-1}(\alpha)
    \right| 
    > \eta
    \right\}. 
\end{align}
Because $G_{\bs\Gamma^\star,\bs u^\star, \bs Y}(T)$ must be stochastically larger than or equal to $\Unif(0,1)$, by the property of quantile functions, 
we can know that, 
for any $\alpha\in (0,1)$, 
\begin{align*}
    \alpha & \ge \pr\{ 
    G_{\bs\Gamma^\star,\bs u^\star, \bs Y}(T) \le \alpha
    \} 
    = 
    \pr\{ 
    F(-T) \le \alpha
    \} 
    \ge 
    \pr\{ 
    -T < F^{-1}(\alpha)
    \} 
    = 
    \pr\{ 
    T > -F^{-1}(\alpha)
    \}. 
\end{align*}
By some algebra, this further implies that 
\begin{align*}
    \alpha & \ge
    \pr\left\{
    \frac{T - \mu_I}{\sigma_I} > -\frac{F^{-1}(\alpha) + \mu_I}{\sigma_I}
    \right\}
    \\
    & \ge 
    \pr\left\{
    \frac{T - \mu_I}{\sigma_I} > - \frac{F^{-1}(\alpha) + \mu_I}{\sigma_I},
    \ 
    \left| \frac{F^{-1}(\alpha) + \mu_I}{\sigma_I} - \Phi^{-1}(\alpha) \right| \le \eta
    \right\}\\
    & \ge 
    \pr\left\{
    \frac{T - \mu_I}{\sigma_I} > -\Phi^{-1}(\alpha) + \eta,
    \ 
    \left| \frac{F^{-1}(\alpha) + \mu_I}{\sigma_I} - \Phi^{-1}(\alpha) \right| \le \eta
    \right\}
    \\
    & \ge 
    \pr\left\{
   \frac{T - \mu_I}{\sigma_I} > -\Phi^{-1}(\alpha) + \eta
    \right\}
    -
    \pr\left\{
    \left| \frac{F^{-1}(\alpha) + \mu_I}{\sigma_I} - \Phi^{-1}(\alpha) \right| > \eta
    \right\}.
\end{align*}
Note that $-\Phi^{-1}(\alpha) = \Phi^{-1}(1-\alpha)$. This immediately implies \eqref{eq:bound_T_u_sigma}.

Third, we prove that, for any  $\alpha\in (0,1)$ and $\eta>0$, 
as $I \rightarrow \infty$, 
\begin{align}\label{eq:conv_quant}
    \pr\left\{
    \left| \frac{F^{-1}(\alpha) + \mu_I}{\sigma_I} - \Phi^{-1}(\alpha) \right| > \eta
    \right\} \rightarrow 0. 
\end{align}
From Lemma A8, Condition \ref{cond:134as} and the property of quantile functions, 
for any $\alpha\in (0,1)$, 
with probability 1, 
$\{F^{-1}(\alpha) + \mu_I \}/\sigma_I$ converges to $\Phi^{-1}(\alpha)$. 
That is, for any $\alpha\in (0,1)$, 
$\{F^{-1}(\alpha) + \mu_I \}/\sigma_I \convergeas \Phi^{-1}(\alpha)$. This immediately implies \eqref{eq:conv_quant}.

Fourth, we prove that, for any $\alpha\in (0,1)$, 
\begin{align}\label{eq:limsup_T_mu_sigma}
    \limsup_{I\rightarrow \infty}
    \pr\left\{
    \frac{T - \mu_I}{\sigma_I} \ge \Phi^{-1}(1-\alpha)
    \right\} \le \alpha. 
\end{align}
Consider any given $\alpha\in (0,1)$. Let $\alpha'$ be any constant in $(\alpha,1)$, and $\eta = \{ \Phi^{-1}(1-\alpha)-\Phi^{-1}(1-\alpha') \}/2 > 0$. 
From \eqref{eq:bound_T_u_sigma}, we have 
\begin{align*}
    \pr\left\{
    \frac{T - \mu_I}{\sigma_I} \ge \Phi^{-1}(1-\alpha) 
    \right\}
    & = 
    \pr\left\{
    \frac{T - \mu_I}{\sigma_I} \ge \Phi^{-1}(1-\alpha') + 2 \eta
    \right\}
    \\
    & \le 
    \pr\left\{
    \frac{T - \mu_I}{\sigma_I} > \Phi^{-1}(1-\alpha') +  \eta
    \right\}
    \\
    & \le \alpha' + \pr\left\{
    \left| 
    \frac{F^{-1}(\alpha') + \mu_I}{\sigma_I} - \Phi^{-1}(\alpha')
    \right| 
    > \eta
    \right\}. 
\end{align*}
From \eqref{eq:conv_quant}, we then have
$\limsup_{I\rightarrow \infty}\pr\{
    (T - \mu_I)/\sigma_I \ge \Phi^{-1}(1-\alpha) 
\} \le \alpha'$. 
Because $\alpha'$ can take any value in $(\alpha, 1)$, 
we must have \eqref{eq:limsup_T_mu_sigma}.

Fifth, we prove that, for any $c\ge 0$,
\begin{align}\label{eq:asymp_real_larger}
    \pr\{\mathcal \mu_I+c\sigma_I > \mu_I(\Gamma_0;k)+c\sigma_I(\Gamma_0;k)\}\to 0 \quad \text{ as }I\to\infty.
\end{align}
From the proof of Theorem \ref{thm:set_quantile_gamma} and Condition \ref{cond:134as}, for any $c\ge 0$, 
we have, with probability 1, 
$\mu_I+c\sigma_I\le \mu_I(\Gamma_0;k)+c\sigma_I(\Gamma_0;k)$ for sufficiently large $I$. 
Thus, 
$
\pr( \bigcup_{m=1}^\infty\bigcap_{I=m}^\infty \{\mu_I+c\sigma_I\le \mu_I(\Gamma_0;k)+c\sigma_I(\Gamma_0;k)\} ) = 1. 
$
By Fatou's lemma, this further implies that 
\begin{align*}
    1 & = \pr\left( \bigcup_{m=1}^\infty\bigcap_{I=m}^\infty \{\mu_I+c\sigma_I\le \mu_I(\Gamma_0;k)+c\sigma_I(\Gamma_0;k)\} \right)
    \\
    & = \pr\left( \liminf_{I\rightarrow \infty} \{\mu_I+c\sigma_I\le \mu_I(\Gamma_0;k)+c\sigma_I(\Gamma_0;k)\} \right)\\
    & \le \liminf_{I\to\infty} \pr\left\{ \mu_I+c\sigma_I\le \mu_I(\Gamma_0;k)+c\sigma_I(\Gamma_0;k) \right\}. 
\end{align*}
Therefore, \eqref{eq:asymp_real_larger} must hold.

Sixth, we prove that, for any $c\ge 0$, 
\begin{align}\label{eq:limsup_T_mu_sigma_Gamma0_k}
    \limsup_{I\rightarrow \infty}
    \pr\left\{
    T \ge \mu_I(\Gamma_0;k)+ c \sigma_I(\Gamma_0;k)
    \right\} \le 1-\Phi(c). 
\end{align}
Consider any given  $c\ge 0$. 
Note that
\begin{align*}
    & \quad \ \pr\big\{
    T \ge \mu_I(\Gamma_0;k)+ c \sigma_I(\Gamma_0;k)
    \big\}\\
    & \le 
    \pr\big\{
    T \ge \mu_I(\Gamma_0;k)+ c \sigma_I(\Gamma_0;k), \ 
    \mu_I(\Gamma_0;k)+ c \sigma_I(\Gamma_0;k) \ge \mu_I + c\sigma_I
    \big\}
    \\
    & \quad \ + 
    \pr\big\{
    \mu_I(\Gamma_0;k)+ c \sigma_I(\Gamma_0;k) < \mu_I + c\sigma_I
    \big\}\\
    & \le 
    \pr(
    T \ge \mu_I + c\sigma_I
    ) + 
    \pr\{
    \mu_I(\Gamma_0;k)+ c \sigma_I(\Gamma_0;k) < \mu_I + c\sigma_I
    \}.
\end{align*}
From \eqref{eq:limsup_T_mu_sigma} and \eqref{eq:asymp_real_larger}, we can immediately derive \eqref{eq:limsup_T_mu_sigma_Gamma0_k}. 

Finally, we prove the asymptotic validity of the $p$-value $\tilde p_{\Gamma_0;k}$. 
From the proof of Corollary \ref{cor:sen_pval_set}, 
for any $\alpha\in (0,1)$, 
$
    \pr(\tilde p_{\Gamma_0;k} \le \alpha ) = 
    \pr\{
    T \ge \mu_I(\Gamma_0;k)+ c \sigma_I(\Gamma_0;k)
    \}, 
$
where $c = \max\{0, \Phi^{-1}(1-\alpha)\}$. 
From \eqref{eq:limsup_T_mu_sigma_Gamma0_k}, we then have 
\begin{align*}
    \limsup_{I\rightarrow \infty}\pr(\tilde p_{\Gamma_0;k} \le \alpha ) & = 
    \limsup_{I\rightarrow \infty} \pr\{
    T \ge \mu_I(\Gamma_0;k)+ c \sigma_I(\Gamma_0;k)
    \}
    \le 1 - \Phi(c) 
    \le 1 - \Phi(\Phi^{-1}(1-\alpha)) 
    \\
    & = \alpha. 
\end{align*}

From the above, Lemma \ref{lemma:bounded_null_sets_effect_increasing} holds.
\end{proof}

\begin{lemma}\label{lemma:bounded_null_sets_diff}
Consider a general matched study with set sizes $n_i \ge 2$ for all $i$. 
If the bounded null hypothesis $\overline{H}_0$ in \eqref{eq:bound_null} holds, 
the bias at rank $k$ is bounded by $\Gamma_0$, i.e., $\Gamma^\sta_{(k)}\le \Gamma_0$, 
Condition \ref{cond:134as} holds, 
and the test statistic $t(\cdot, \cdot)$ in \eqref{eq:test_stat} is differential increasing, 
then $\tilde p_{\Gamma_0;k}$ in Corollary \ref{cor:sen_pval_set} is still an asymptotic valid $p$-value, in the sense that $\limsup_{I\to\infty}\pr( \tilde{p}_{\Gamma_0;k} \le \alpha ) \le \alpha$ for $\alpha\in (0,1)$.
\end{lemma}

\begin{proof}[\bf Proof of Lemma \ref{lemma:bounded_null_sets_diff}]
We adopt the same notation as in the proof of Lemma \ref{lemma:bounded_null_sets_effect_increasing}. 

We first prove that $G_{\bs\Gamma^\star,\bs u^\star, \bs Y}(T)$ is stochastically larger than or equal to $\Unif(0,1)$.
By the same logic as \eqref{eq:tail_prob_ineq_diff}, we have 
\begin{align*}
    G_{\bs\Gamma^\star,\bs u^\star, \bs Y}(T)
    \ge 
    G_{\bs\Gamma^\star,\bs u^\star,\bs Y(0)}(t(\bs Z,\bs Y(0))) 
    = 
    G(t(\bs Z,\bs Y(0))). 
\end{align*}
Because $G(t(\bs Z,\bs Y(0)))$ is stochastically larger than or equal to $\Unif(0,1)$, $G_{\bs\Gamma^\star,\bs u^\star, \bs Y}(T)$ must also be stochastically larger than or equal to $\Unif(0,1)$.

Following the same steps as in the proof of Lemma \ref{lemma:bounded_null_sets_effect_increasing}, 
we can then derive Lemma \ref{lemma:bounded_null_sets_diff}. 
\end{proof}

\begin{proof}[\bf Proof of Theorem \ref{thm:bounded_null_sets_effect_diff_increasing}]
From the property of convergence in probability \citep[see, e.g.,][Theorem 2(d)]{ferguson2017course}, 
if Condition \ref{134in_prob} holds, 
then
for any subsequence $\{I_m,m=1,2,\cdots\}$, there exists a further subsequence $\{I_{m_j},j=1,2,\cdots\}$ such that Condition \ref{cond:134as} holds. 
From Lemmas \ref{lemma:bounded_null_sets_effect_increasing} and \ref{lemma:bounded_null_sets_diff}, we then have 
$\limsup_{j\to\infty}\pr( \tilde{p}_{\Gamma_0;k;I_{m_j}} \le \alpha ) \le \alpha$ for $\alpha\in (0,1)$.

From the above, for any subsequence $\{I_m,m=1,2,\cdots\}$, there exists a further subsequence $\{I_{m_j},j=1,2,\cdots\}$ such that $\limsup_{j\to\infty}\pr( \tilde{p}_{\Gamma_0;k;I_{m_j}} \le \alpha ) \le \alpha$ for $\alpha\in (0,1)$. 
This must imply that $\limsup_{I\to\infty}\pr( \tilde{p}_{\Gamma_0;k} \le \alpha ) \le \alpha$ for $\alpha\in (0,1)$. 
Therefore, Theorem \ref{thm:bounded_null_sets_effect_diff_increasing} holds. 
\end{proof}

To prove Theorem \ref{thm:bound_joint_ci_set}, we introduce the following condition.

\begin{condition}\label{cond:bound_joint_ci_as}
With the true control potential outcome vector $\bs{Y}(0)$ replaced by the imputed control potential outcome vector $\bs{Y}$ based on the sharp null ${\lxr H_{0}}$ (or equivalently pretending $\bs{Y}$ as the true control potential outcome vector), 
Condition \ref{cond:set_simultaneous} holds 
{\lxr almost surely},
in the sense that the {\lxr convergences} in these three conditions hold almost surely as $I \rightarrow \infty$. 
\end{condition}

\begin{proof}[\bf Proof of Theorem \ref{thm:bound_joint_ci_set}] 
To prove Theorem \ref{thm:bound_joint_ci_set}, it suffices to prove that, for any $\alpha\in (0,1)$, the confidence set in \eqref{eq:set_joint_set} covers the true hidden bias vector $\bs{\Gamma}^\sta$ with probability at least $1-\alpha$ asymptotically, i.e., 
\begin{align}\label{eq:bound_null_joint_ci_set}
    \liminf_{I \rightarrow \infty} \pr\Big( \bs{\Gamma}^\sta \in \bigcap_{\Gamma_0, k: \vardbtilde{p}_{\Gamma_0;k} \le \alpha} \mathcal{S}_{\Gamma_0;k}^\complement \Big) \ge 1-\alpha.
\end{align}
By the same logic as the proof of Theorem \ref{thm:bounded_null_sets_effect_diff_increasing}, 
it suffices to prove \eqref{eq:bound_null_joint_ci_set} under Condition \ref{cond:bound_joint_ci_as}. 
Below we will assume that Condition \ref{cond:bound_joint_ci_as} holds. 

First, we prove that, for any $c\ge 0$,
\begin{align}\label{eq:joint_bound_mean_sd}
    \pr\Big(
    \inf_{\underline{k}_I \le k\le I, \Gamma_0 \in [\Gamma_{(k)}^\sta, \overline{\Gamma}_I]} \left\{ \mu(\Gamma_0;k) + c \sigma(\Gamma_0; k) \right\}
    < \mu + c\sigma
    \Big) \rightarrow 0 \quad \text{as } I \rightarrow \infty. 
\end{align}
From Condition \ref{cond:bound_joint_ci_as} and Lemma \ref{lemma:set_joint_ci}, with probability $1$, 
$\inf_{\underline{k}_I \le k\le I, \Gamma_0 \in [\Gamma_{(k)}^\sta, \overline{\Gamma}_I]} \left\{ \mu(\Gamma_0;k) + c \sigma(\Gamma_0; k) \right\}$ 
{\lxr is greater than or equal to}
$\mu + c\sigma$ when $I$ is sufficiently large. 
That is, 
\begin{align*}
    \pr\Big( \bigcup_{m=1}^\infty\bigcap_{I=m}^\infty 
    \Big\{ \inf_{\underline{k}_I \le k\le I, \Gamma_0 \in [\Gamma_{(k)}^\sta, \overline{\Gamma}_I]} \left\{ \mu(\Gamma_0;k) + c \sigma(\Gamma_0; k) \right\} \ge \mu + c\sigma \Big\} \Big) = 1. 
\end{align*}
From Fatou's lemma and by the same logic as the proof of \eqref{eq:asymp_real_larger},
we can then derive \eqref{eq:joint_bound_mean_sd}. 

Second, from the proof of Theorem \ref{thm:set_joint_ci}, 
for any $\alpha\in (0,1)$, 
if $\bs{\Gamma}^\sta \in \bigcup_{\Gamma_0, k: \vardbtilde{p}_{\Gamma_0;k} \le \alpha} \mathcal{S}_{\Gamma_0;k}$, then we must have 
$
    T \ge \inf_{\underline{k}_I \le k\le I, \Gamma_0 \in [\Gamma_{(k)}^\sta, \overline{\Gamma}_I]} \left\{ \mu(\Gamma_0;k) + c \sigma(\Gamma_0; k) \right\} 
$
with $c = \max\{0, \Phi^{-1}(1-\alpha)\}$. 
This then implies that 
\begin{align*}
    & \quad \ \pr\Big( \bs{\Gamma}^\sta \in \bigcup_{\Gamma_0, k: \vardbtilde{p}_{\Gamma_0;k} \le \alpha} \mathcal{S}_{\Gamma_0;k} \Big)
    \\
    & \le 
    \pr\Big( 
    T \ge \inf_{\underline{k}_I \le k\le I, \Gamma_0 \in [\Gamma_{(k)}^\sta, \overline{\Gamma}_I]} \left\{ \mu(\Gamma_0;k) + c \sigma(\Gamma_0; k) \right\}
    \Big)
    \\
    & 
    \le 
    \pr\Big( 
    T \ge \inf_{\underline{k}_I \le k\le I, \Gamma_0 \in [\Gamma_{(k)}^\sta, \overline{\Gamma}_I]} \left\{ \mu(\Gamma_0;k) + c \sigma(\Gamma_0; k) \right\}, 
    \inf_{\underline{k}_I \le k\le I, \Gamma_0 \in [\Gamma_{(k)}^\sta, \overline{\Gamma}_I]} \left\{ \mu(\Gamma_0;k) + c \sigma(\Gamma_0; k) \right\}
    \ge \mu + c\sigma
    \Big)\\
    & \quad \ + \pr\Big(
    \inf_{\underline{k}_I \le k\le I, \Gamma_0 \in [\Gamma_{(k)}^\sta, \overline{\Gamma}_I]} \left\{ \mu(\Gamma_0;k) + c \sigma(\Gamma_0; k) \right\}
    < \mu + c\sigma
    \Big)\\
    & \le 
    \pr( 
    T \ge \mu + c\sigma
    )
    + \pr\Big(
    \inf_{\underline{k}_I \le k\le I, \Gamma_0 \in [\Gamma_{(k)}^\sta, \overline{\Gamma}_I]} \left\{ \mu(\Gamma_0;k) + c \sigma(\Gamma_0; k) \right\}
    < \mu + c\sigma
    \Big).
\end{align*}
{\lxr By the same logic as} \eqref{eq:limsup_T_mu_sigma} in the proof of Lemma \ref{lemma:bounded_null_sets_effect_increasing} and {\lxr from} \eqref{eq:joint_bound_mean_sd}, we must have \begin{align*}
    \limsup_{I\rightarrow \infty}
    \pr\Big( \bs{\Gamma}^\sta \in \bigcup_{\Gamma_0, k: \vardbtilde{p}_{\Gamma_0;k} \le \alpha} \mathcal{S}_{\Gamma_0;k} \Big)
    \le 1- \Phi(c) \le 1- \Phi(\Phi^{-1}(1-\alpha)) = \alpha. 
\end{align*}
This immediately implies \eqref{eq:bound_null_joint_ci_set}. 

From the above, Theorem \ref{thm:bound_joint_ci_set} holds. 
\end{proof}

\section{Technical Details for Remark \ref{rmk:match}}\label{sec:tech_match}

Below we give the technical details for Remark \ref{rmk:match}. 
Specifically, we will show that the model in \eqref{eq:true_treat_assign_mech} can be justified when we have independently sampled units in the original data prior to matching and the matching is exact.

Let $\{(Z_l, \bs{X}_l, Y_l(1), Y_l(0)): l=1,2,\ldots, L\}$ be the samples in the original dataset prior to matching. 
Assume that, given $\mathcal{F} \equiv \{(\bs{X}_l, Y_l(1), Y_l(0)): l=1,2,\ldots, L\}$, the treatment assignments $Z_l$'s are mutually independent, 
and define 
$
\pi_l \equiv \Pr(Z_l = 1\mid \mathcal{F}) 
$
for all $l$.

For simplicity, we consider the pair matching; similar arguments will work for matching with multiple controls as well. 
We first divide the units into $G$ groups so that units within the same group have the same covariate value and units in different groups have distinct covariate values, i.e., 
$\bs{X}_l = \bs{X}_k$ for any units $l,k$ within the same group
and $\bs{X}_l \ne \bs{X}_k$ for any units $l,k$ in different groups. 
We then perform random pair matching within each group. 
Algorithmically, 
within each group,
we can randomly permute the order of treated and control units, and match them based on their order within their permutations. 
That is, after the permutation, the first treated unit is matched with the first control unit, 
the second treated unit is matched with the second control unit, and so on. 
For each group, some control (or treated) units will not be matched if the group has more control (or treated) units. 
Finally, we obtain the matched dataset, and use $\mathcal{M}$ to denote the whole matching structure, which contains the information for the indices of units within each matched pair. 
Let 
$\bs{Z}_{\mathcal{M}}$ denote the treatment assignment vector for all matched units, 
and $\bs{Z}_{\mathcal{M}^\complement}$ denote the treatment assignment vector for all unmatched units. Moreover, 
we say  $\bs{Z}_{\mathcal{M}}\in \mathcal{E}$ if there is exactly one treated unit within each matched pair under $\mathcal{M}$; equivalently, $\mathcal{E}$ is the set of all possible assignments for matched units such that there is exactly one treated unit within each matched pair. 
Below we will investigate the conditional distribution of the treatment assignments for matched units, given the matching structure $\mathcal{M}$,  the potential outcomes and covariates for all units $\mathcal{F}$, the treatment assignments for unmatched units, and that each matched pair has only one treated unit, i.e., 
\begin{align}\label{eq:cond_match}
    \bs{Z}_{\mathcal{M}} \mid \mathcal{M}, \mathcal{F}, \bs{Z}_{\mathcal{M}^\complement},  \bs{Z}_{\mathcal{M}}\in \mathcal{E}.
\end{align}

To investigate \eqref{eq:cond_match}, we first investigate the conditional distribution of a particular matching structure given the treatment assignments for all units, i.e.,
\begin{align}\label{eq:match_given_Z}
    \Pr(\mathcal{M} = \mathpzc{m} \mid \bs{Z}=\bs{z}, \mathcal{F}).
\end{align}
By our matching algorithm described above, the probability in \eqref{eq:match_given_Z} takes value zero if $\bs{z}_{\mathpzc{m}} \notin \mathcal{E}$ or $\mathpzc{m}^\complement$ contains treated and control units with the same covariate value (note that in our matching algorithm we try to get as many exactly matched pairs as possible). 
Let {\lxr $\mathcal{Z}(\mathpzc{m})$} be the set of treatment assignment vector $\bs{z}$ such that (i) $\bs{z}_{\mathpzc{m}}\in \mathcal{E}$ and (ii) $\mathpzc{m}^\complement$ does not contain treated and control units with the same covariate value. 
It is not difficult to see that
$
    \Pr(\mathcal{M} = \mathpzc{m} \mid \bs{Z}=\bs{z}, \mathcal{F}) = \Pr(\mathcal{M} = \mathpzc{m} \mid \bs{Z}=\bs{z}', \mathcal{F})
$
for any $\bs{z}, \bs{z}'\in \mathcal{Z}(\mathpzc{m})$ with $\bs{z}_{\mathpzc{m}^\complement} = \bs{z}_{\mathpzc{m}^\complement}'$, 
due to the random pairing in our matching algorithm described before. 

We then investigate the conditional distribution in \eqref{eq:cond_match}, i.e., the following conditional probability for any given $\bs{z}$ and $\mathpzc{m}$: 
\begin{align}\label{eq:cond_z_m}
    \Pr( \bs{Z}_{\mathcal{M}} = \bs{z}_{\mathpzc{m}} \mid \mathcal{M} = \mathpzc{m}, \mathcal{F}, \bs{Z}_{\mathcal{M}^\complement} = \bs{z}_{\mathpzc{m}^\complement},  \bs{Z}_{\mathcal{M}} \in \mathcal{E} ).
\end{align}
Obviously, \eqref{eq:cond_z_m} is well-defined and can take nonzero value only if $\bs{z}\in \mathcal{Z}(\mathpzc{m})$. 
Below we will assume this is true. 
Note that 
\begin{align*}
    & \quad \ \Pr( \bs{Z}_{\mathcal{M}} = \bs{z}_{\mathpzc{m}}, \mathcal{M} = \mathpzc{m},  \bs{Z}_{\mathcal{M}^\complement} = \bs{z}_{\mathpzc{m}^\complement},  \bs{Z}_{\mathcal{M}} \in \mathcal{E} \mid \mathcal{F} )
    \\
    & =  
    \Pr( \bs{Z} = \bs{z}, \mathcal{M} = \mathpzc{m} \mid \mathcal{F} )
    = \Pr( \mathcal{M} = \mathpzc{m} \mid \bs{Z} = \bs{z},  \mathcal{F} ) 
    \Pr( \bs{Z} = \bs{z} \mid \mathcal{F} ),
\end{align*}
where the first equality holds due to the fact that $\bs{z}\in \mathcal{Z}(\mathpzc{m})$, 
and 
\begin{align*}
    \Pr( \mathcal{M} = \mathpzc{m},  \bs{Z}_{\mathcal{M}^\complement} = \bs{z}_{\mathpzc{m}^\complement},  \bs{Z}_{\mathcal{M}} \in \mathcal{E} \mid \mathcal{F} )
    & =  
    \sum_{\bs{z}'\in \mathcal{Z}(\mathpzc{m}): \bs{z}_{\mathpzc{m}^\complement}' = \bs{z}_{\mathpzc{m}^\complement}}
    \Pr(\bs{Z} = \bs{z}', \mathcal{M} = \mathpzc{m} \mid \mathcal{F} )
    \\
    & = \sum_{\bs{z}'\in \mathcal{Z}(\mathpzc{m}): \bs{z}_{\mathpzc{m}^\complement}' = \bs{z}_{\mathpzc{m}^\complement}} 
    \Pr( \mathcal{M} = \mathpzc{m} \mid \bs{Z} = \bs{z}',  \mathcal{F} ) 
    \Pr( \bs{Z} = \bs{z}' \mid \mathcal{F} ).
\end{align*}
We then have the following equivalent forms for \eqref{eq:cond_z_m}: 
\begin{align*}
    & \quad \ 
    \Pr( \bs{Z}_{\mathcal{M}} = \bs{z}_{\mathpzc{m}} \mid \mathcal{M} = \mathpzc{m}, \mathcal{F}, \bs{Z}_{\mathcal{M}^\complement} = \bs{z}_{\mathpzc{m}^\complement},  \bs{Z}_{\mathcal{M}} \in \mathcal{E} )
    \\
    & = 
    \frac{\Pr( \bs{Z}_{\mathcal{M}} = \bs{z}_{\mathpzc{m}}, \mathcal{M} = \mathpzc{m},  \bs{Z}_{\mathcal{M}^\complement} = \bs{z}_{\mathpzc{m}^\complement},  \bs{Z}_{\mathcal{M}} \in \mathcal{E} \mid \mathcal{F} )}{\Pr( \mathcal{M} = \mathpzc{m},  \bs{Z}_{\mathcal{M}^\complement} = \bs{z}_{\mathpzc{m}^\complement},  \bs{Z}_{\mathcal{M}} \in \mathcal{E} \mid \mathcal{F} )}
    \\
    & = 
    \frac{
    \Pr( \mathcal{M} = \mathpzc{m} \mid \bs{Z} = \bs{z},  \mathcal{F} ) 
    \Pr( \bs{Z} = \bs{z} \mid \mathcal{F} )
    }{
    \sum_{\bs{z}'\in \mathcal{Z}(\mathpzc{m}): \bs{z}_{\mathpzc{m}^\complement}' = \bs{z}_{\mathpzc{m}^\complement}} 
    \Pr( \mathcal{M} = \mathpzc{m} \mid \bs{Z} = \bs{z}',  \mathcal{F} ) 
    \Pr( \bs{Z} = \bs{z}' \mid \mathcal{F} )
    }
    \\
    & = 
    \frac{
    \Pr( \bs{Z} = \bs{z} \mid \mathcal{F} )
    }{
    \sum_{\bs{z}'\in \mathcal{Z}(\mathpzc{m}): \bs{z}_{\mathpzc{m}^\complement}' = \bs{z}_{\mathpzc{m}^\complement}} 
    \Pr( \bs{Z} = \bs{z}' \mid \mathcal{F} )
    },
\end{align*}
where the last equality holds due to the discussion before {\lxr for \eqref{eq:match_given_Z}}. 
By our assumption on the treatment assignments for units in the original dataset before matching, 
we can further write \eqref{eq:cond_z_m} as
\begin{align*}
    & \quad \ \Pr( \bs{Z}_{\mathcal{M}} = \bs{z}_{\mathpzc{m}} \mid \mathcal{M} = \mathpzc{m}, \mathcal{F}, \bs{Z}_{\mathcal{M}^\complement} = \bs{z}_{\mathpzc{m}^\complement},  \bs{Z}_{\mathcal{M}} \in \mathcal{E} )
    \\
    & = 
    \frac{
    \Pr( \bs{Z} = \bs{z} \mid \mathcal{F} )
    }{
    \sum_{\bs{z}'\in \mathcal{Z}(\mathpzc{m}): \bs{z}_{\mathpzc{m}^\complement}' = \bs{z}_{\mathpzc{m}^\complement}} 
    \Pr( \bs{Z} = \bs{z}' \mid \mathcal{F} )
    }
    =
    \frac{
    \Pr( \bs{Z}_{\mathpzc{m}} = \bs{z}_{\mathpzc{m}} \mid \mathcal{F} )
    }{
    \sum_{\bs{z}'\in \mathcal{Z}(\mathpzc{m}): \bs{z}_{\mathpzc{m}^\complement}' = \bs{z}_{\mathpzc{m}^\complement}} 
    \Pr( \bs{Z}_{\mathpzc{m}} = \bs{z}'_{\mathpzc{m}} \mid \mathcal{F} )
    }\\
    & = 
    \frac{
    \prod_{i=1}^I \prod_{j=1}^2 \pi_{ij}^{z_{ij}} (1-\pi_{ij})^{1-z_{ij}}
    }{
    \prod_{i=1}^I \{\pi_{i1}(1-\pi_{i2})+\pi_{i2} (1-\pi_{i1})\}
    }
    =
    \prod_{i=1}^I\frac{
    \prod_{j=1}^2 \{\pi_{ij}/(1-\pi_{ij})\}^{z_{ij}}
    }{
    \pi_{i1}/(1-\pi_{i1})+\pi_{i2}/(1-\pi_{i2})
    }\\
    & = 
    \prod_{i=1}^I
    \prod_{j=1}^2
    \left\{
    \frac{
    \pi_{ij}/(1-\pi_{ij})
    }{
    \pi_{i1}/(1-\pi_{i1})+\pi_{i2}/(1-\pi_{i2})
    }
    \right\}^{z_{ij}},
\end{align*}
where we use $ij$ with $1\le i\le I$ and $j=1,2$ to index the matched units under $\mathpzc{m}$. 
Obviously, the treatment assignments for matched units satisfy model \eqref{eq:true_treat_assign_mech}, with 
$p_{ij}^\star = \pi_{ij}/(1-\pi_{ij})/\{\pi_{i1}/(1-\pi_{i1})+\pi_{i2}/(1-\pi_{i2})\}$. 
Consequently, the hidden bias for each matched pair $i$ has the following forms: 
\begin{align*}
    \Gamma_i^\star = \frac{\max_j p_{ij}^\star}{ \min_k p_{ik}^\star}
    = 
    \frac{\max_j \pi_{ij}/(1-\pi_{ij})}{ \min_k \pi_{ik}/(1-\pi_{ik})}, \quad (1\le i\le I).
\end{align*}

Below we give some additional remarks on the 
{\lxr treatment probabilities}
$\pi_l$'s. 
If we further assume that the units are independent samples, in the sense that $(Z_l, \bs{X}_l, Y_l(1), Y_l(0))$'s are mutually independent across all $s$, 
then the 
{\lxr treatment probabilities}
simply to 
\begin{align*}
    \pi_l = \Pr(Z_l = 1 \mid \mathcal{F}) = \Pr(Z_l = 1 \mid \bs{X}_l, Y_l(1), Y_l(0) ), 
    \quad (1\le l \le L). 
\end{align*}
If we further assume that $(Z_l, \bs{X}_l, Y_l(1), Y_l(0))$'s are 
{\lxr i.i.d.~}samples from some superpopulation, then there must exist a function $\tilde{\lambda}(\cdot)$ such that 
\begin{align*}
    \pi_l = \Pr(Z_l = 1 \mid \mathcal{F}) = \Pr(Z_l = 1 \mid \bs{X}_l, Y_l(1), Y_l(0) ) = \tilde{\lambda}(\bs{X}_l, Y_l(1), Y_l(0)), 
    \quad (1\le l \le L). 
\end{align*}
If the covariates contain all the confounding, in the sense that $Z_l \ind (Y_l(1), Y_l(0)) \mid \bs{X}_l$, then there must exist a function $\lambda(\cdot)$ such that 
\begin{align*}
    \pi_l = \Pr(Z_l = 1 \mid \mathcal{F}) = \Pr(Z_l = 1 \mid \bs{X}_l, Y_l(1), Y_l(0) ) = \tilde{\lambda}(\bs{X}_l, Y_l(1), Y_l(0)) = \lambda(\bs{X}_l), 
    \quad (1\le l \le L). 
\end{align*}
Consequently, when matching is exact, units within the same matched set will have the same 
{\lxr treatment probabilities}, under which the hidden biases are all $1$ and the matched observational study reduces to a matched randomized experiment.
}

\section{Simulation studies}\label{sec:simulation}

{\add
In this section, we conduct simulations to investigate the validity and power of the proposed sensitivity analysis based on quantiles of hidden biases. 
Specifically, 
we will first conduct simulations under the situation where Fisher's null of no effect holds, to investigate the finite-sample performance of the asymptotic $p$-value in Corollary \ref{cor:sen_pval_set}.}
Following \citet{rosenbaum2010book}, we {\add then} conduct simulations under the ``favorable'' situation where there is nonzero causal effect but no unmeasured confounding, under which we hope the inferred causal effects to be significant and insensitive to 
unmeasured confounding.
In particular, 
we will {\lxr investigate}
{\add the power of our sensitivity analysis}
when we use different test statistics or 
when
the 
causal effects have different patterns. 

{\add 
\subsection{Sensitivity analysis when Fisher's null of no effect holds}

We investigate the finite-sample performance of the asymptotic $p$-value in Corollary \ref{cor:sen_pval_set} when Fisher's null $H_0$ of no effect in \eqref{eq:H0} holds. 
We consider a matched study with $I$ matched sets of equal sizes $n_1 = n_2 = \ldots = n_I \equiv n$ for $n=2$ and $3$, respectively.  
We consider the following two models for generating the potential outcomes: 
\begin{align}\label{eq:gen_outcome}
     \text{(i) normal: } Y_{ij}(1) = Y_{ij}(0) \iidsim \mathcal{N}(0, 1), 
     \quad 
     \text{(ii) binary: } Y_{ij}(1) = Y_{ij}(0) \iidsim \text{Bernoulli}(0.5), 
\end{align}
and the following three models for generating the true hidden biases $\bs{\Gamma}^\sta$: 
\begin{align}\label{eq:gen_hidden_bias}
    \text{(i)} & \text{ constant: } \Gamma_i^\sta = 5 \text{ for all } i, 
    \quad 
    \text{(ii) lognormal: } \Gamma_i^\sta \iidsim \max\{1, \text{Lognormal}(1.5, 0.2^2)\}, 
    \nonumber
    \\
    \text{(iii)} & \text{ outlier: } \Gamma_i^\sta = 500 \text{ for $5\%$ of the $I$ matched sets}, \text{ and } \Gamma_i^\sta = 5 \text{ for remaining sets}. 
\end{align}
We generate the hidden confounding $\bs{u}$ in the following way: 
\begin{align}\label{eq:gen_hidden_confound}
    u_{ij}^\sta = \frac{Y_{ij}(0) - \min_k Y_{ik}(0)}{\max_k Y_{ik}(0) - \min_k Y_{ik}(0)}, 
    \quad (1\le i \le I; 1\le j \le n).
\end{align}
To ensure the quantities in \eqref{eq:gen_hidden_confound} is well-defined, 
when generating binary potential outcomes as in (ii) in \eqref{eq:gen_outcome}, 
we will randomly select two units within each matched set and let their potential outcomes take values $0$ and $1$, respectively. 
Note that the binary potential outcomes in \eqref{eq:gen_outcome} are carefully designed so that our worst-case consideration can be achieved by the actual treatment assignment mechanisms. 
To mimic the finite population inference, the potential outcomes and hidden biases, once generated, are kept fixed, and the treatment assignments are randomly drawn from $\mathcal{M}(\bs{\Gamma}^\sta, \bs{u}^\sta)$ as defined in Definition \ref{def:sen_model}. 

Figures \ref{fig:pair_I200} and  \ref{fig:triple_I200} 
show the empirical distribution functions of the asymptotic $p$-value in Corollary \ref{cor:sen_pval_set} under various models for generating the potential outcomes and hidden biases as described in \eqref{eq:gen_outcome} and \eqref{eq:gen_hidden_bias}, over 500 simulated treatment assignments from $\mathcal{M}(\bs{\Gamma}^\sta, \bs{u}^\sta)$, 
for matched pairs ($n=2$) and matched triples ($n=3$), 
respectively, with $I=200$. 
In both Figures, we conduct sensitivity analysis for three quantiles of hidden biases: $\Gamma_{(I)}^\sta$ (black lines), $\Gamma_{(0.95I)}^\sta$ (red lines) and $\Gamma_{(0.9I)}^\sta$ (green lines), respectively, 
where the bounds on these quantiles of hidden biases are the true ones. 
From Figures \ref{fig:pair_I200} and  \ref{fig:triple_I200}, in most cases, the empirical distribution of the $p$-value is stochastically larger than or equal to the uniform distribution on $(0,1)$. 
In Figures \ref{fig:pair_I200}(d), \ref{fig:pair_I200}(f) and \ref{fig:triple_I200}(d), there is slight type-I inflation, especially when we consider significance level close to $0.5$. 
We therefore further conduct simulations in these three settings with the number of matched sets increased to $I=2000$. 
As shown in Figure \ref{fig:I2000}, with larger sample size, the empirical distribution of the $p$-value becomes closer to (or is still stochastically larger than or equal to) the uniform distribution on $(0,1)$. 

From the above, the asymptotic $p$-value in Corollary \ref{cor:sen_pval_set}
{\lxr reasonably controls the type-I error in finite samples.}

\begin{figure}[htb]
	\begin{subfigure}{.33\textwidth}
		\centering
		\includegraphics[width=1\linewidth]{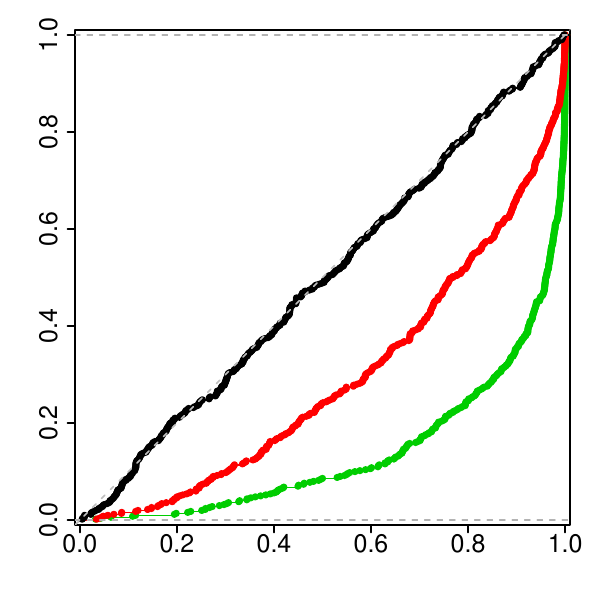}
		\caption{normal + constant}
	\end{subfigure}%
	\begin{subfigure}{.33\textwidth}
		\centering
		\includegraphics[width=1\linewidth]{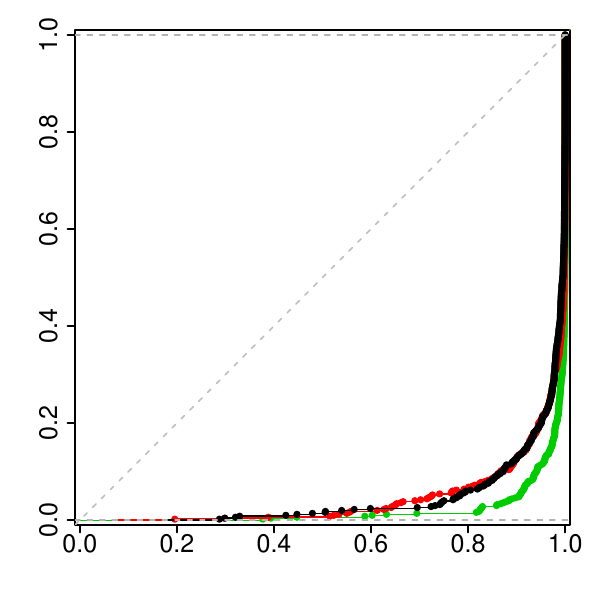}
		\caption{normal + lognormal}
	\end{subfigure}%
	\begin{subfigure}{.33\textwidth}
		\centering
		\includegraphics[width=1\linewidth]{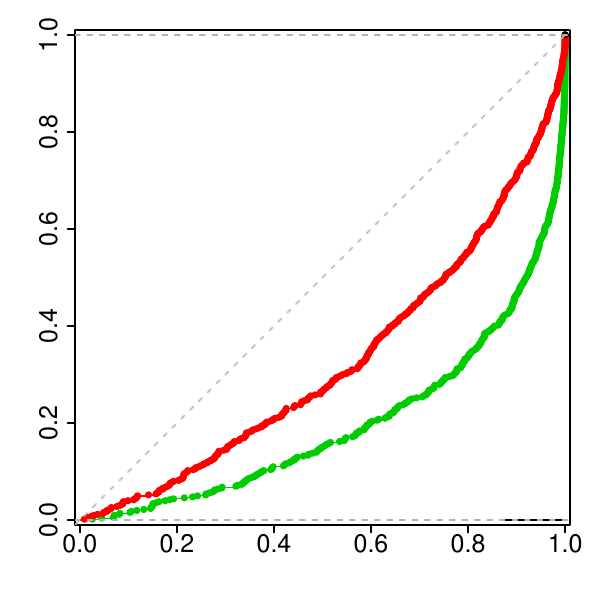}
		\caption{normal + outlier}
	\end{subfigure}
	\begin{subfigure}{.33\textwidth}
		\centering
		\includegraphics[width=1\linewidth]{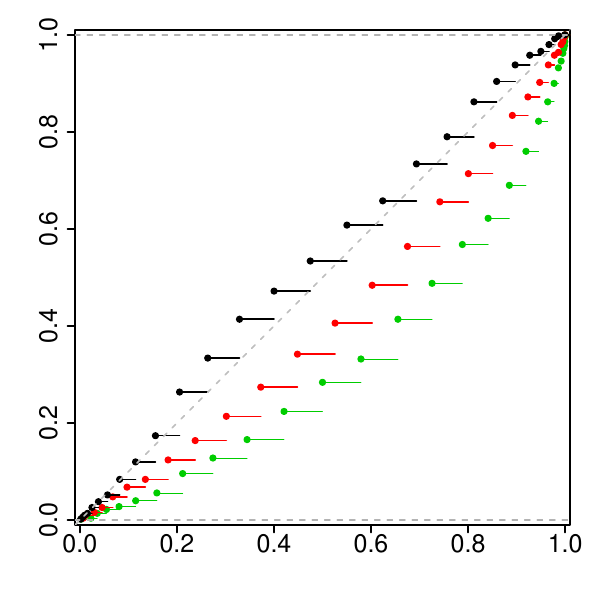}
		\caption{binary + constant} 
	\end{subfigure}
	\begin{subfigure}{.33\textwidth}
		\centering
		\includegraphics[width=1\linewidth]{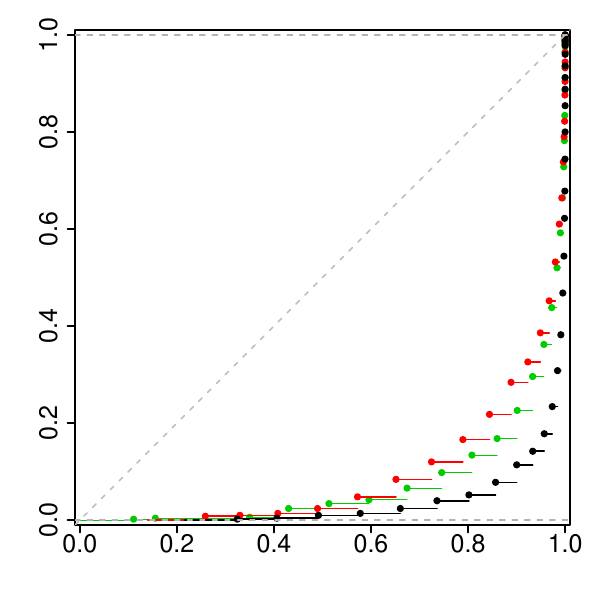}
		\caption{binary + lognormal} 
	\end{subfigure}%
	\begin{subfigure}{.33\textwidth}
	\centering
	\includegraphics[width=1\linewidth]{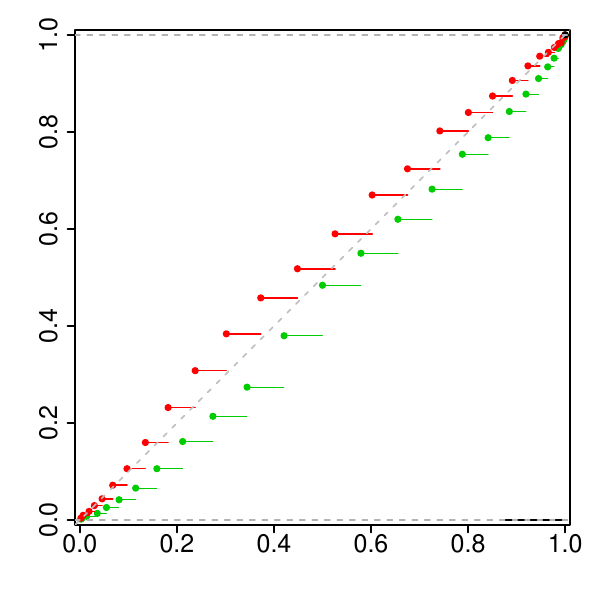}
	\caption{binary + outlier}
	\end{subfigure}
	\caption{
	Empirical distributions of the $p$-value in Corollary \ref{cor:sen_pval_set} under various data generating models in a matched pair study with $I=200$ and equal set size $n = 2$. 
	The caption of each subfigure describes the models for generating potential outcomes and hidden biases as in \eqref{eq:gen_outcome} and \eqref{eq:gen_hidden_bias}, respectively. 
	The black, red and green lines show the empirical distribution functions of the $p$-value under sensitivity analysis for $\Gamma_{(I)}^\sta$, $\Gamma_{(0.95I)}^\sta$ and $\Gamma_{(0.9I)}^\sta$, respectively. 
	The grey dashed lines show the distribution function of the uniform distribution on the interval $(0,1)$. 
	Note that in (c) and (f), the $p$-value from the sensitivity analysis based on $\Gamma_{(I)}^\sta$ are all $1$, under which its empirical distribution becomes a point mass at $1$ and cannot be easily seen in the figure. 
	}\label{fig:pair_I200}
\end{figure}

\begin{figure}[htbp]
	\begin{subfigure}{.33\textwidth}
		\centering
		\includegraphics[width=1\linewidth]{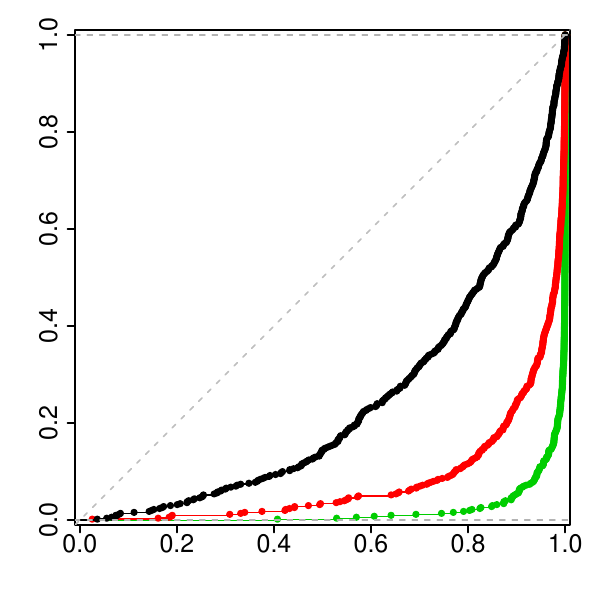}
		\caption{normal + constant}
	\end{subfigure}%
	\begin{subfigure}{.33\textwidth}
		\centering
		\includegraphics[width=1\linewidth]{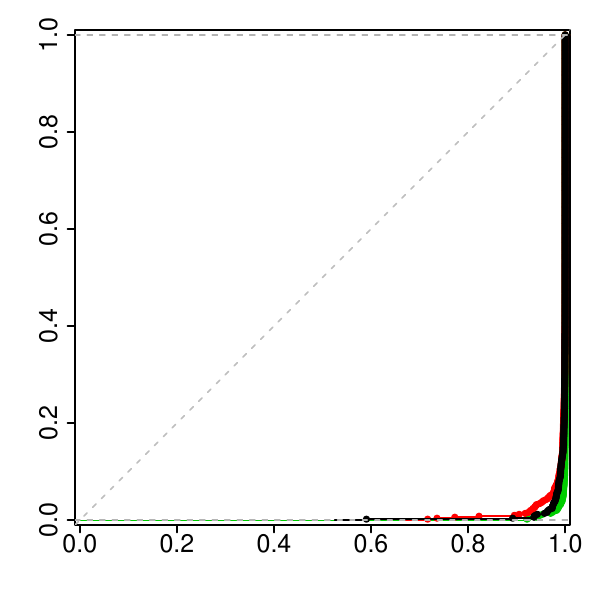}
		\caption{normal + lognormal}
	\end{subfigure}%
	\begin{subfigure}{.33\textwidth}
		\centering
		\includegraphics[width=1\linewidth]{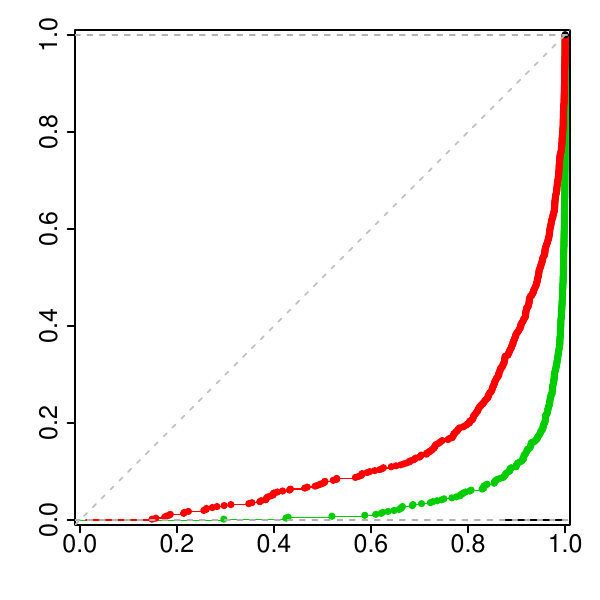}
		\caption{normal + outlier}
	\end{subfigure}
	\begin{subfigure}{.33\textwidth}
		\centering
		\includegraphics[width=1\linewidth]{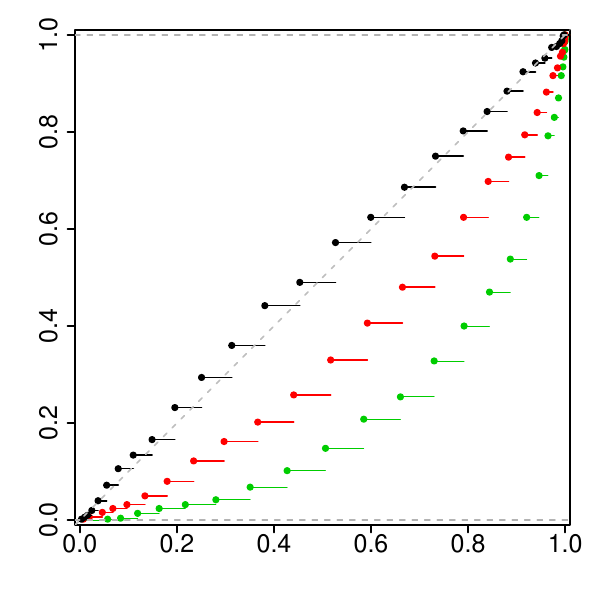}
		\caption{binary + constant} 
	\end{subfigure}
	\begin{subfigure}{.33\textwidth}
		\centering
		\includegraphics[width=1\linewidth]{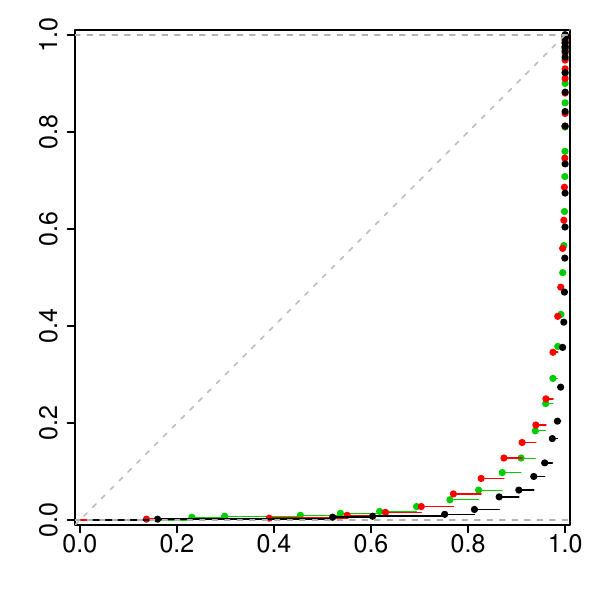}
		\caption{binary + lognormal} 
	\end{subfigure}%
	\begin{subfigure}{.33\textwidth}
	    \centering
	    \includegraphics[width=1\linewidth]{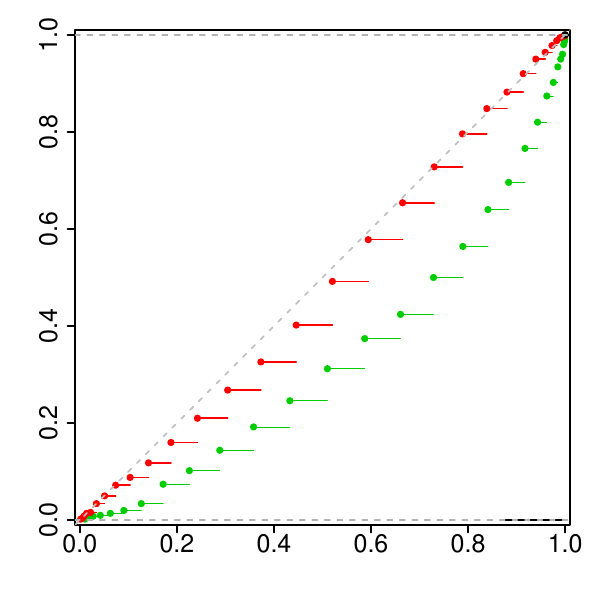}
	    \caption{binary + outlier}
	\end{subfigure}
	\caption{
	Empirical distributions of the $p$-value in Corollary \ref{cor:sen_pval_set} under various data generating models in a matched triple study with $I=200$ and equal set size $n = 3$. 
	The caption of each subfigure describes the models for generating potential outcomes and hidden biases as in \eqref{eq:gen_outcome} and \eqref{eq:gen_hidden_bias}, respectively. 
	The black, red and green lines show the empirical distribution functions of the $p$-value under sensitivity analysis for $\Gamma_{(I)}^\sta$, $\Gamma_{(0.95I)}^\sta$ and $\Gamma_{(0.9I)}^\sta$, respectively. 
	The grey dashed lines show the distribution function of the uniform distribution on the interval $(0,1)$. 
	Note that in (c) and (f), the $p$-value from the sensitivity analysis based on $\Gamma_{(I)}^\sta$ are all $1$, under which its empirical distribution becomes a point mass at $1$ and cannot be easily seen in the figure. 
	}\label{fig:triple_I200}
\end{figure}

\begin{figure}[htbp]
	\begin{subfigure}{.34\textwidth}
		\centering
		\includegraphics[width=0.97\linewidth]{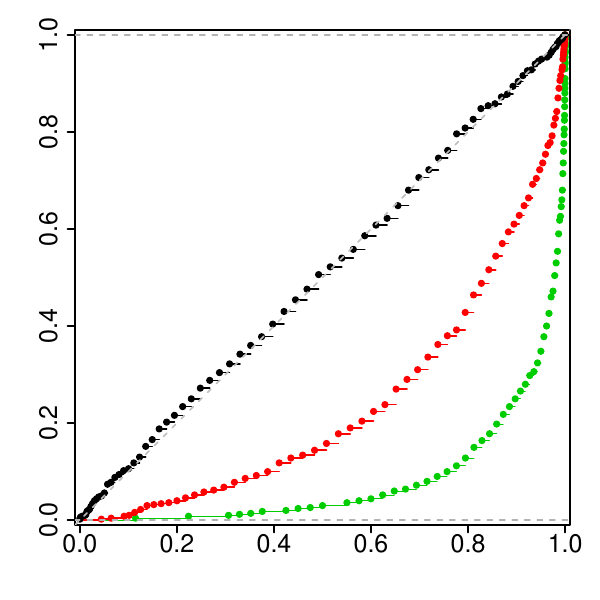}
		\caption{binary + constant (2)}
	\end{subfigure}%
	\begin{subfigure}{.34\textwidth}
		\centering
		\includegraphics[width=0.97\linewidth]{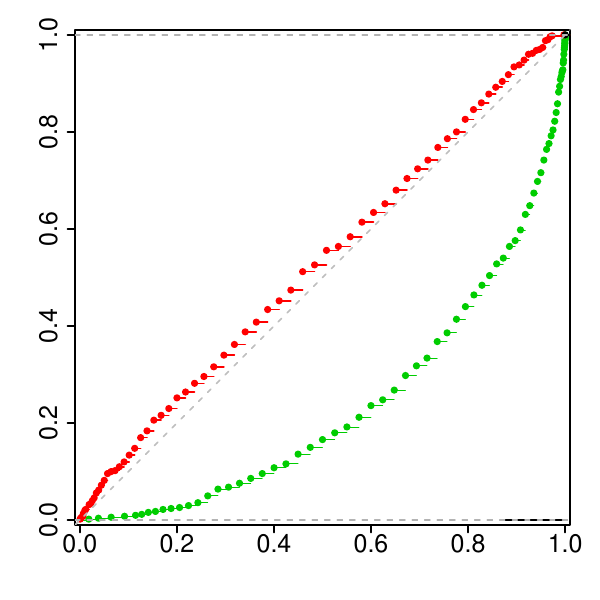}
		\caption{binary + outlier (2)}
	\end{subfigure}%
	\begin{subfigure}{.34\textwidth}
		\centering
		\includegraphics[width=0.97\linewidth]{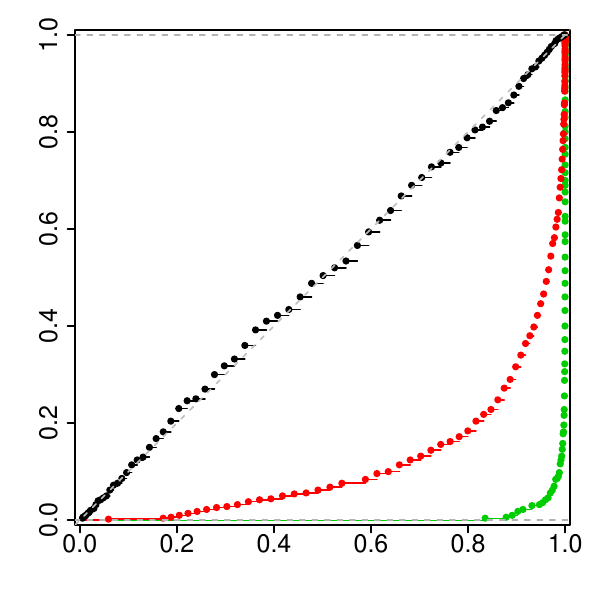}
		\caption{binary + constant (3)}
	\end{subfigure}
	\caption{
	Empirical distributions of the $p$-value in Corollary \ref{cor:sen_pval_set} under various data generating models in a matched study with $I=2000$ and equal set size. 
	The caption of each subfigure describes the models for generating potential outcomes and hidden biases  as in \eqref{eq:gen_outcome} and \eqref{eq:gen_hidden_bias}, respectively, with the number in the parentheses representing the matched set size. 
	The black, red and green lines show the empirical distribution functions of the $p$-value under sensitivity analysis for $\Gamma_{(I)}^\sta$, $\Gamma_{(0.95I)}^\sta$ and $\Gamma_{(0.9I)}^\sta$, respectively. 
	The grey dashed lines show the distribution function of the uniform distribution on the interval $(0,1)$. 
	Note that in (b), the $p$-value from the sensitivity analysis based on $\Gamma_{(I)}^\sta$ are all $1$, under which its empirical distribution becomes a point mass at $1$ and cannot be easily observed in the figure. 
	}\label{fig:I2000}
\end{figure}

}

\subsection{Sensitivity analysis using different {\lxr choices} of test statistics}\label{sec:simu_stat}
We 
investigate the power of our sensitivity analysis using various {\lxr choices} of test statistics. In particular, we will investigate how the inner trimming in the m-statistic affects the sensitivity analysis. We consider a matched pair study with $I=500$ pairs, assume that the treatment effects are constant $0.5$ across all units, and generate the control potential outcomes $Y_{ij}(0)$'s as i.i.d.~samples from Gaussian distribution with mean $0$ and variance $0.5$.

\begin{figure}[htb]
    \centering
        \includegraphics[width=.325\textwidth]{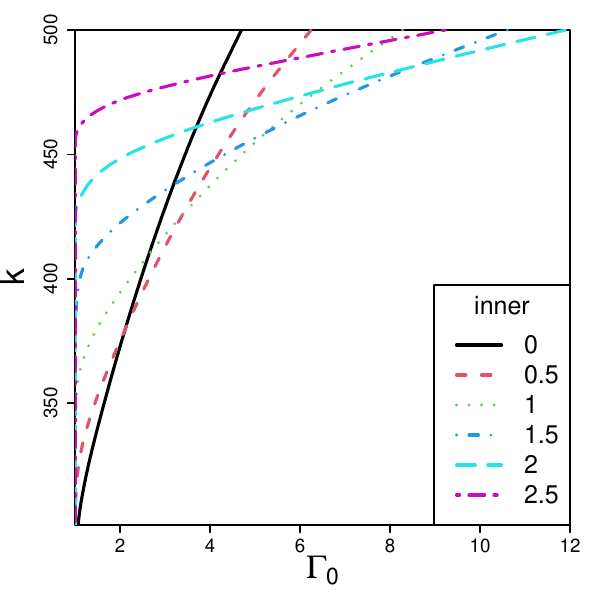}
    \caption{$95\%$ lower confidence limits for quantiles of hidden biases in a matched pair study assuming that Fisher's null $H_0$ holds, averaging over all simulated datasets. 
    The data generation is described in 
    {\lxr Appendix \ref{sec:simu_stat}}.
    Different lines correspond to m-statistics with various inner trimming. 
    }
    \label{fig:sen_simu_diff_test}
\end{figure}

We use the m-statistic discussed in Section \ref{sec:test_rand_sen} to conduct sensitivity analysis for quantiles of hidden biases. 
We fix the outer trimming at $3$, and vary the inner trimming from 0 to 2.5 by a grid of 0.5. 
Assuming Fisher's null holds, 
Figure \ref{fig:sen_simu_diff_test} shows the $95\%$ lower confidence limits for quantiles of hidden biases using various inner trimming, 
averaging over 200 simulated datasets. 
From Figure \ref{fig:sen_simu_diff_test}, 
the average lower confidence limit for the maximum hidden bias first increases when the inner trimming increases, and it achieves its maximum when inner trimming equals to 2. 
It then decreases when inner trimming further increases to 2.5. 
This is intuitive. On the one hand, inner trimming focuses more on matched pairs with larger treatment effects and thus improves the power of sensitivity analysis for maximum hidden bias; on the other hand, it also reduces the effective sample size. 
With large sample size, more inner trimming generally improves the power of conventional sensitivity analysis that focuses on the maximum hidden bias. 

\begin{table}[htb]
    \centering
    \small
    \caption{Average of the $95\%$ lower confidence limits for all quantiles of hidden biases from a matched pair study,  averaging over all simulated datasets from the model described in 
    {\lxr Appendix \ref{sec:simu_stat}}.
    With $\tilde{\Gamma}_{(k)}$ denoting the lower confidence limit for $\Gamma^\sta_{(k)}$, the 2nd--6th columns report the average values of $g^{-1}( I^{-1}\sum_{k=1}^I g(\tilde{\Gamma}_{(k)}) )$ over all simulated datasets, where $g$ is specified in the 1st column. 
    }
    \label{tab:average_sensitivity}
    \begin{tabular}{ccccccccc}
    \toprule
    inner & $0$ & $0.5$ & $1$ & $1.5$ & $2$ & $2.5$\\
    \midrule
    $g(x)=x$
    & 2.591072 & 2.932655& 3.074130& 2.930389 & 2.460424 & 1.652486\\[5pt]
    $g(x) = \log(x)$
    & 2.345861 & 2.459622 & 2.287098 & 1.922898 & 1.543919 & 1.244339 \\[5pt]
    $g(x) = x/(1+x)$
    & 2.256807 & 2.291832 & 2.056154 & 1.706557 & 1.401990 & 1.190056 \\
    \bottomrule
\end{tabular}
\end{table}
However, as shown in Figure \ref{fig:sen_simu_diff_test}, greater inner trimming, although improves the power of sensitivity analysis for maximum hidden bias, can deteriorate the power of sensitivity analysis for quantiles of hidden biases. 
Specifically, when the inner trimming is 2, although the sensitivity analysis for maximum hidden bias becomes most powerful, with causal effects being insensitive to hidden biases at {\lxr a} magnitude {\lxr of} about 12, the sensitivity analysis for quantiles of hidden biases becomes less powerful. Specifically, the lower confidence limit for the $430/500 = 86\%$ quantiles of hidden biases is almost the noninformative 1. 
From Figure \ref{fig:sen_simu_diff_test}, in general, the smaller the inner trimming, the more informative the sensitivity analysis is for lower quantiles of hidden biases. 
Obviously, there is a trade-off for the choice of inner trimming. 
In Table 2, we report the average value of the lower confidence limits (or their transformations) for all quantiles of hidden biases, averaging over all simulated datasets. 
From the discussion at the end of Section \ref{sec:simultaneous_pair}, these are actually $95\%$ simultaneous lower confidence limits for the average hidden biases, averaging over all the simulations. 
From Table \ref{tab:average_sensitivity}, inner trimming at 
$0.5$ 
seems to provide the most informative sensitivity analysis for the average hidden biases.

\subsection{Simulation for sensitivity analysis under various treatment effect patterns}\label{sec:simu_effect}

\begin{figure}[htbp]
    \centering
        \centering
        \includegraphics[width=.33\textwidth]{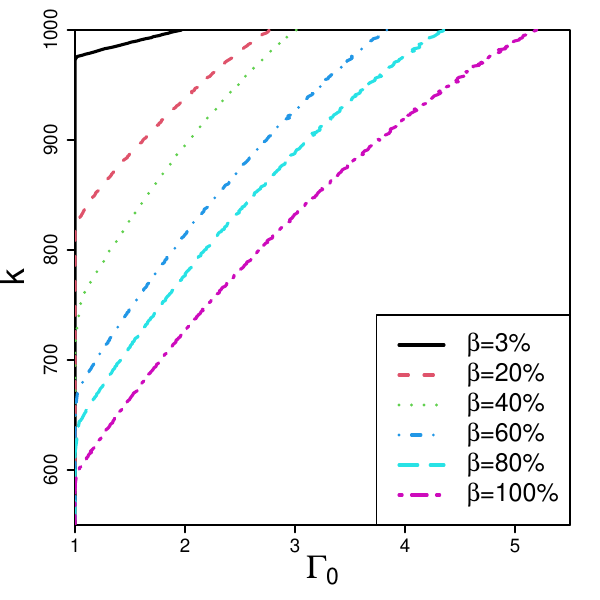}
    \caption{$95\%$ lower confidence limits for quantiles of hidden biases in a matched pair study assuming that Fisher's null $H_0$ holds, averaging over all simulated datasets. 
    The data are generated from the model described in 
    {\lxr Appendix}
    \ref{sec:simu_effect} with various proportion $\beta$ of units that are affected by the treatment. 
    }
    \label{fig:sen_simu_treatment_pattern}
\end{figure}

We consider a matched pair study where the individual treatment effects for units across sets have the following form: 
\begin{align*}
   Y_{ij}(1)-Y_{ij}(0)=
\begin{cases}
c/\beta&1\le i \le \beta I,j=1,2;\\
0& \beta I+1\le i \le I,j=1,2,
\end{cases} 
\end{align*}
where $c>0$ and $\beta >0$. 
Specifically, we assume that $1-\beta$ proportion of the units have zero treatment effects, while the remaining $\beta$ proportion of units have positive effects. 
Note that the average treatment effect is fixed at $c$, independent of $\beta$. 
Intuitively, when $\beta=1$, the treatment has a constant effect $c$ for all units. 
When $\beta$ becomes smaller, the treatment effects become more heterogeneous and are nonzero for a smaller proportion of the units, with the total treatment effects kept fixed. 
In the following, we will fix $I=1000$, and generate control potential outcomes $Y_{ij}(0)$'s as 
{\lxr i.i.d.~}samples from Gaussian distribution $\mathcal{N}(0, 0.5)$. 
Intuitively, the treated-minus-control difference within each matched pair will follow a Gaussian distribution with mean $c/\beta$ or $0$ and standard deviation 1. 
We use the difference-in-means as our test statistic in \eqref{eq:test_stat}, which corresponds to the m-statistic discussed in Section \ref{sec:test_rand_sen} {\lxr and Appendix \ref{sec:choice_test_stat_appendix}} {\lxr without any inner or outer trimming.}

Assuming Fisher's null holds, 
Figure \ref{fig:sen_simu_treatment_pattern} {\lxr shows} the $95\%$ lower confidence limits for the quantiles of hidden biases averaging over 100 simulated datasets when the average effect $c=0.5$, with the proportion of units with positive effects $\beta$ taking values $3\%, 20\%, 40\%, 60\%, 80\%,$ $100\%$. 
From Figure \ref{fig:sen_simu_treatment_pattern}, we can see that more homogeneous positive effects provide us more powerful sensitivity analysis results, under which the inferred causal effects will be more robust to {\lxr the} existence of unmeasured confounding. 
To explain away the observed association between treatment and outcome, 
the maximum hidden bias needed increases with $\beta$, 
and this is further manifested in the quantiles of hidden biases. 
Specifically, {\lxr when only 3\% of units are affected by the treatment, even with a large effect size,} the resulting study will be {\lxr highly} sensitive to hidden confounding if the investigator suspects that a few matched sets may have exceptionally large hidden biases.
From Figure \ref{fig:sen_simu_treatment_pattern}, the study loses power to detect significant treatment effects when about $3\%$ of the sets are suspected to have large hidden biases. 
Intuitively, a matched observational study with more homogeneous nonzero treatment effects across matched sets (given that the total effects are kept fixed) can provide us more powerful sensitivity analysis and thus more robust causal conclusions, especially when we suspect extreme hidden biases in some or a 
{\lxr certain}
proportion of the matched sets.

\section{Matching for studying the effect of smoking}\label{app:match}

We use the data from the 2005–2006 National Health and Nutrition Examination Survey, which are also available in \citet{bigmatch2020}. 
\begin{figure}[h]
    \centering
    \includegraphics[width=1\textwidth]{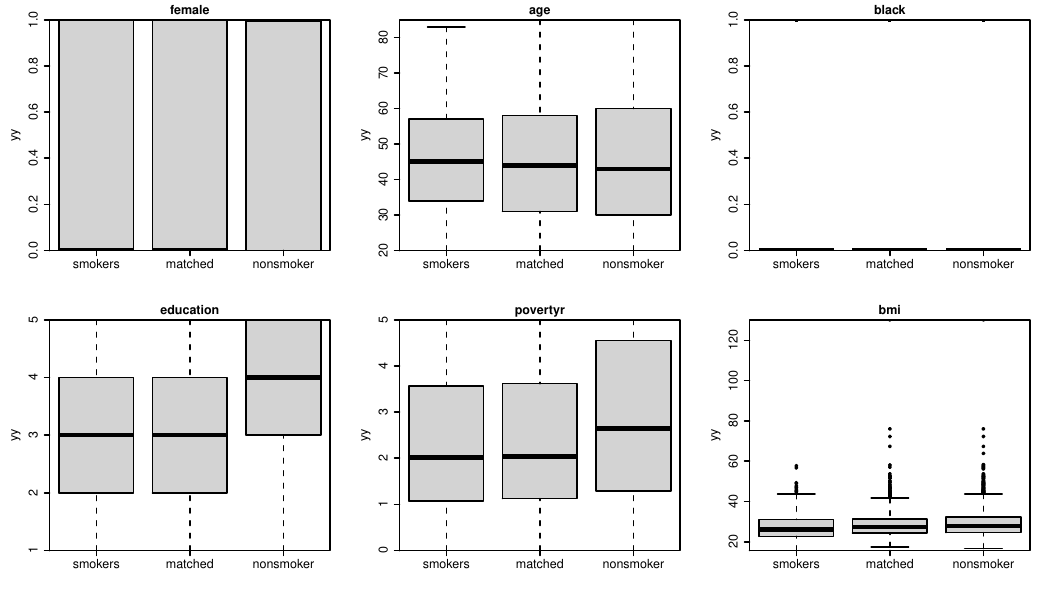}
    \caption{Box plots of covariates for smokers, matched nonsmokers and original nonsmokers. 
    These six plots, from left to right and from top to bottom, are for covariates gender, age, race, education level, household income level and body mass index, respectively. 
    }
    \label{fig:boxplot}
\end{figure}
The data include 2475 observations, including 521 daily smokers and 1963 nonsmokers. 
We perform the minimum-distance near-fine matching in \citet{YuRosenbaum2019}, and construct a matched comparison consisting of 512 smoker-nonsmoker matched {\lxr sets}.
The matching takes into account six covariates including gender, age, race, education level, household income level and body mass index; for a more detailed description of these covariates, see, e.g., \citet{bigmatch2020}. 
Figure \ref{fig:boxplot} shows the box plots of the six covariates for smokers, matched nonsmokers and original nonsmokers. 
From Figure \ref{fig:boxplot}, the balance of these covariates has been substantially improved by matching.

\end{document}